\documentclass{birkjour}

\usepackage{bm}%
\usepackage{stmaryrd}%
\usepackage{cite}%
\usepackage{booktabs}
\usepackage{mathtools}
\usepackage{mathrsfs}
\usepackage{bbm}
\usepackage{color}

\newcommand{\R}{\mathbb{R}}%

\renewcommand{\vec}[1]{\boldsymbol{#1}}

\theoremstyle{plain}%
\newtheorem{theorem}{Theorem}%
\newtheorem{lemma}[theorem]{Lemma}%
\newtheorem{corollary}[theorem]{Corollary}%
\newtheorem{proposition}[theorem]{Proposition}%
\theoremstyle{definition}%
\newtheorem{definition}[theorem]{Definition}%
\newtheorem{algorithm}[theorem]{Algorithm}%
\newtheorem{example}[theorem]{Example}%
\theoremstyle{remark}%

\numberwithin{equation}{section}%
\numberwithin{theorem}{section}%

\begin{document}

\title[State transition graph structure and hysteresis]%
{The structure of state transition graphs in systems with return point memory:  I. General Theory}

\author[M.~Mungan]{Muhittin Mungan}%
\address{%
  Institut f\"{u}r angewandte Mathematik \\
  Bonn University \\
  Endenicher Allee 60 \\
  53115 Bonn, Germany}%
\email{mungan@iam.uni-bonn.de}%

\thanks{MMT acknowledges partial support through NSF grant CHE \#1464926, MM  was supported by the
    German Research Foundation (DFG) in the Collaborative Research Centre 1060 "The Mathematics
    of Emergent Effects" and DFG Project No. 398962893. This research was also supported in part through NSF 
    grant NSF PHY 17-48958.}

\author[M.M.~Terzi]{M. Mert Terzi}%
\address{%
  Physics Department \\
  Carnegie Mellon University \\
  Pittsburgh PA 15213, USA}%
\email{mterzi@andrew.cmu.edu}

 

\begin{abstract}
We consider the athermal quasi-static dynamics (AQS) of disordered systems driven by an external field. 
Our interest is in an automaton description (AQS-A) that represents the AQS dynamics via a 
graph of state transitions triggered by the field and in the presence of    
return point memory (RPM) property, a tendency for the system to return to the same microstate 
upon cycling driving. The existence of three conditions, (1) a partial order on the set of 
configuration; 
(2) a no-passing property; and (3) an adiabatic response to monotonously changing fields, implies RPM. 
When periodically driven, such systems settle into a cyclic response after a transient of at most one period. 
While sufficient, conditions (1) - (3) are not necessary. We show that the AQS dynamics  provides 
a more selective partial order which, due to its explicit connection to hysteresis loops,  is a natural choice for 
establishing the RPM property. This enables us to 
consider AQS-A exhibiting RPM without   
necessarily possessing the no-passing property. We call such automata $\ell$AQS-A and work out the structure of their state transition graphs.
We find that the RPM property constrains 
the {\em intra-loop} structure of hysteresis loops, namely its hierarchical organization into  sub loops, but 
not the {\em inter-loop} structure. We prove that the topology of the intra-loop structure can be represented as 
an ordered tree and show that the corresponding state transition graph is planar. On the other hand, the RPM property does not 
significantly restrict the  inter-loop transitions. A system exhibiting RPM and 
subject to periodic forcing can undergo a large number of transient cycles before settling into a periodic response. Such systems can 
even exhibit subharmonic response. 
\end{abstract}

\maketitle


\section{Introduction}
\label{sec:intro}

When a deformable structure in a disordered medium is subjected 
to athermal, time-periodic forcing such as shearing, its dynamical state can  
evolve into a limit-cycle. The system subjected to repeated 
cycling, explores the energy landscape by traversing energy barriers 
into nearby states. The wandering around the energy landscape can 
last many cycles of irreversible changes, the ``training phase'', 
before a dynamic state is reached where the system is locked into a limit cycle, thereafter responding 
periodically to the driving. 
Such phenomena have been experimentally observed for a  
variety of systems, such as charge-density 
waves \cite{Fleming86}, magnetic multilayers \cite{Pierce2003, Pierce2005, Pierce2007, Hauet2008}, 
particle diffusions in sheared viscous  
suspensions \cite{Pine2005, Chaikin2008, Keimetal2014}, granular matter \cite{toiya2004transient,mueggenburg2005behavior,ren2013reynolds}
and amorphous solids \cite{haw1998colloidal, petekidis2002rearrangements, Lundberg2008}. 
Many of these systems involve particles with sizes of the order of $10^{-6}$ to $10^{-3}$m and thermal 
effects are often negligible. 
Minimal models have been formulated  and simulated    
\cite{Tangetal87,SNCLittlewood87, Chaikin2008, Dahmenetal2009, Keimetal2011, Libaletal2012,
Keimetal2013, Fioccoetal2014, Fiocco2015, Regev2015, RossoWyart2015,Sastry2017}
which qualitatively reproduce the experimental findings.   

One can think of the asymptotic response to the 
cyclic forcing as some form of memory that 
encodes features of the forcing, such as its 
amplitude\cite{Keimetal2011}. Indeed, experiments as well as numerical 
simulations suggest that the control parameter $\gamma$ 
of the forcing  can impose a type of order on the 
set of dynamically accessible cycle-states, 
organizing these into a hierarchy.  Thus for example 
in the experiments and models of diffusion in sheared 
viscous suspensions 
\cite{Chaikin2008,Keimetal2011, Keimetal2013,Keimetal2014} 
a periodic state reached 
upon driving at some amplitude $\gamma$, also retains its periodicity  
for smaller amplitudes $\gamma' < \gamma$,  
meaning that once the limit cycle at amplitude $\gamma$ 
is reached, subsequently 
reducing the amplitude to $\gamma' < \gamma$ causes the 
system to remain on a limit cycle, which now 
 traverses a subsegment of the original 
$\gamma$-cycle. On the other hand, subjecting 
the $\gamma$ limit-cycle to a forcing at an amplitude 
$\gamma' > \gamma$, a transient to a 
new limit-cycle sets in. This is the phenomenon of {\em marginal stability}, well-known 
within the context of charge-density-waves \cite{Tangetal87}. A recent comprehensive review on memory formation in 
matter can be found in \cite{keim2018memory}.

A strict hierarchical organization of limit-cycles and marginal stability
are features emerging in the hysteretic response of certain  
ferromagnets to oscillations of an applied magnetic 
field. They are generally associated with  
{\em return point memory} (RPM) \cite{Barker1983, Sethna93, SethnaNoiseReview2001, SethnaHyst2006}, 
in which a system periodically returns to a set of microstates and hence ``remembers'' these. 

It was shown theoretically  
that the presence of three conditions is sufficient for the 
RPM property to emerge \cite{Sethna93}: (1) existence of 
a partial ordering of states; (2) a no-passing (NP) property that 
implies that two ordered configurations will retain their 
ordering 
under a monotonous 
change in the forcing \cite{NoPassing}; (3) 
adiabaticity, a rate-independent response to 
any monotonous change in the forcing.  
Assumptions (1) - (3)  emerge naturally from the athermal 
relaxational dynamics of many disordered systems. Condition (3) in particular implies that 
the dynamics is athermal quasi-static (AQS) \cite{AQSMalandroLacks1999,AQSMaloney2006}. 
All of these conditions are not necessary however for RPM. 
Indeed, RPM has been observed in systems where the no-passing property (2)  does not hold, 
such as in the presence of anti-ferromagnetic interactions\cite{Deutschetal2004,KatzgraberZimanyi2006}. 
RPM effects have also been observed experimentally at the microscopic level in 
systems such as magnetic multilayers \cite{Pierce2003, Pierce2005, Pierce2007, Hauet2008},  
artificial spin ice\cite{Hanetal2008,Gilbert2016}, and recently colloidal suspensions \cite{Keim2018return}, where
due to kind of interactions, one again does not expect NP to hold. 
Given the possibility of experimentally and numerically observing the microscopic configurations of such 
systems,  it is of theoretical interest to pursue minimal models 
(AQS automata) that exhibit RPM without the NP property. This is the main goal of the present paper. 

Key to our results is the observation that the no-passing property 
relies on the existence of a partial order on the space of configurations that is then 
preserved by the dynamics. The AQS dynamics  
provides an alternative and arguably more natural partial order, the {\em dynamic partial order}. 
This partial order is compatible with the intrinsic partial order, but generally stronger: 
any pair of configurations ordered with respect to the dynamic order is also ordered 
with respect to the intrinsic partial order, but not necessarily {\em vice versa}. Moreover, 
the dynamic partial order is closely connected with the hysteresis loops in the state transition graphs. 
It therefore is a suitable starting point to  
define the RPM property. To emphasize its dynamic origin and the connection with loops, we will refer to this type of 
RPM as {\em loop} RPM ($\ell$RPM). Roughly speaking, the $\ell$RPM property prescribes whether a state can return to 
a previously visited state or not, namely whether there {\em exists} some forcing that will achieve this. The NP 
property implies that such a forcing exists and in addition imposes constraints on what this forcing 
should be. 

The introduction of $\ell$RPM allows us to disentangle features of systems exhibiting 
RPM but not NP, from those where RPM follows as a result of NP. 
Our central finding is that $\ell$RPM  {\em does not} impose strong constraints on the 
{\em inter-loop} structure, namely the way transitions occur between configurations belonging to 
{\em different} hierarchies of loops and sub loops. As a consequence, systems exhibiting $\ell$RPM and subject to 
time-periodic forcing, can suffer long transients before eventually settling into a limit cycle. This is in sharp 
contrast to RPM as a result of NP, where due to the NP property  
the length of the transient response to periodic forcing is at most one cycle. 

The organization of the paper and our main results are as follows. We start in Section \ref{sec:damamodels} with 
a formal definition of AQS automata (AQS-A), Def.~\ref{def:AQSA}. In this section we also 
introduce the dynamic partial order and define the $\ell$RPM property, Defs.~\ref{def:RPM}--\ref{def:RPMII}. 
We motivate our approach using the random field Ising model (RFIM) as an example, 
whose dynamics we reformulate as an AQS-A. We return 
to the RFIM throughout this section for illustration purposes. 
In Section \ref{sec:damaloop} we then develop the necessary concepts to describe the {\em intra-loop structure} 
of AQS-A possessing the $\ell$RPM property, working out the structure of the state transitions graph. 
The main results we establish here are:

\begin{itemize}
 \item[(1)] There exists a well-defined partition of a hysteresis loop into sub loops. 
 This leads to the representation of the hierarchical organization of sub-loops as an ordered tree, Thm.~\ref{prop:mert}.
\item [(2)] The state transition graph associated with a hysteresis loop and its sub loops is planar, Thm.~\ref{thm:planar}.
\end{itemize}
As we will show in Section \ref{sec:maxloopetc}, the $\ell$RPM property by itself does not restrict strongly 
the transitions ``between'' loops belonging to different hierarchies, the {\em maximal loops}. The inter-loop 
transition graph describes the transitions between these maximal loops, Def.~\ref{def:interLoop}. 
We construct two AQS-A examples where the inter-loop transitions involve an arbitrarily 
large number of maximal loops that are transient,  
Examples \ref{ex:shortTrans} and \ref{ex:longTrans}. 
Such inter-loop structures influence the length of the transient response and 
in Section \ref{sec:PF} we turn to the AQS-A subject to time-periodic forcing. 
As we will demonstrate, 
when $\ell$RPM is present but the no-passing property does not hold, the length of the transient can 
become arbitrarily long. A subharmonic response where the period of the limit cycle is an integral multiple of the period of the forcing is 
possible as well \cite{Deutschetal2003,Nagel2017}, Example \ref{ex:subharmonic}. 
Under time-periodic forcing, a short transient is ensured in the following case, however:
\begin{itemize}
 \item[(3)] Under $\ell$RPM and the marginal stability condition, 
Defs.~\ref{def:loopmarginality} and \ref{def:lAQSAmarginality}, a one period transient is always realized 
once the dynamics has reached a confining loop and  irrespective of the state in the loop it started out, 
  Thm.~\ref{thm:trans}.
\end{itemize}
Prop.~\ref{prop:AQSmarginality} provides a necessary and sufficient condition for marginal stability. 
In the absence of marginality, a limit-cycle will be reached inside the confining loop but the transients can be longer 
than one period. 
We next turn to the role of the no-passing property. NP is a statement about what can happen when the confinement assumption 
of Theorem \ref{thm:trans} does not hold and 
the state trajectory has to eventually leave the loop. It implies that the loops entered and left become increasingly more 
confining so that a confining loop is guaranteed to be found, Thm.~\ref{thm:NPInterloop}. For the case of periodic forcing, this  
means that such a confining loop is found within one period. Thus one role of NP in systems exhibiting the $\ell$RPM property is 
to constrain the inter-loop transitions so that the destination loop is always more confining then the source loop. 
This provides an alternative way to arrive at the result of \cite{Sethna93}.   
We conclude with a summary and discussion of our results in Section \ref{sec:discussion}, where we will also revisit the question of 
testing for RPM in the presence or absence of NP. In order not to interrupt the flow of the argument in the main 
text, the proofs of the Lemmas, Propositions and Theorems are collected  in Appendix \ref{app:proofs}. 

The description of  RPM by means of the state transitions on a countable state space is rather general, 
and hence applicable to a broader class of disordered systems, including those with continuous degrees of freedom, {\em e.g.}  the 
zero-temperature subthreshold behavior of CDW-type models such as the Fukuyama-Lee-Rice \cite{FuLee,LeeRice} or the 
toy model for depinning \cite{KMShort13,KMLong13}. In such systems the states of the automaton become  
equivalence classes of configurations that can be continuously deformed into each other. A quantitative analysis of the 
structure of the state transition graph of depinning toy model will be carried out in a companion paper \cite{TerziMungan2018}.  




\section{AQS automata and the return-point-memory property}
\label{sec:damamodels}

Before beginning with the formal set-up and definitions, it is useful to present a concrete example to motivate our approach. Throughout this section 
we will return to this example for illustration purposes. Readers interested in a more formal development of the subject can safely ignore these 
parts. We will explicitly indicate where an example begins and use  a $\maltese$ to mark where it ends. 
We will use this convention throughout the other sections as well.

\begin{example}[The random field Ising Model as an AQS automaton]
\label{ex:RFIM_I}

Consider the zero-temperature behavior of the random-field Ising model (RFIM) 
\cite{Sethna93, SethnaNoiseReview2001, SethnaHyst2006} with 
ferromagnetic interactions, a random impurity field and  
subject to a time-varying external magnetic field. 
Denote by $\mathcal{S}$ the set of configurations 
$\vec{\sigma}$ that are stable at {\em some} field $F$\footnote{
In order to be consistent with the notation used in the remainder of the paper, we will denote the 
magnetic field by $F$, instead of the customary $H$.}. 
A configuration $\vec{\sigma}$ is stable, if each spin $\sigma_i \in \pm 1$ is aligned with its  
local field, which is the total field experienced by $\sigma_i$, including also the contributions due to its neighboring spins. 
Hence there is an interval of external magnetic fields $F^-(\vec{\sigma}) < F <  F^+(\vec{\sigma})$ over which 
a configuration $\vec{\sigma}$ is stable. 
Given a stable configuration $\vec{\sigma}$ and increasing the magnetic field 
slowly (adiabaticity) to  $F^+(\vec{\sigma})$ will trigger an avalanche of spin flips. At the end of the avalanche 
a new stable configuration $\vec{\sigma'}$ is attained. Since the configuration $\vec{\sigma'}$ is stable 
at $F = F^+(\vec{\sigma})$, it follows that  $F^+(\vec{\sigma'}) > F^+(\vec{\sigma}) $. In a similar manner, 
one can consider a field decrease to $F^-(\vec{\sigma})$, leading to a stable configuration $\vec{\sigma''}$ and 
the relation $F^-(\vec{\sigma''}) < F^-(\vec{\sigma})$.

Given a stable state $\vec{\sigma}$ we can write the state transitions described above 
 in terms of a pair of maps $U$ and $D$ on $\mathcal{S}$ as
\begin{align*}
 \vec{\sigma'} &= U\vec{\sigma},   \\
 \vec{\sigma''} &= D\vec{\sigma}, 
\end{align*}
along with the inequalities
\begin{align*}
 F^+(\vec{\sigma}) &< F^+(\vec{U\sigma}), \\
 F^-(D\vec{\sigma}) &< F^-(\vec{\sigma}). 
\end{align*}
Thus $U\vec{\sigma}$ is the state that $\vec{\sigma}$ transits to when the field is increased to $F = F^+[\vec{\sigma}]$, 
and likewise $D\vec{\sigma}$ is the state that $\vec{\sigma}$ transits to upon a field decrease to $F = F^-[\vec{\sigma}]$.
With an initial configuration $\vec{\sigma}$ given, the AQS response of the 
system to an arbitrary field-history $F_t$ can therefore be translated into a sequence of $U$ and $D$ operations 
applied to $\vec{\sigma}$. The disorder enters this description through the specific form of the maps $U$ and $D$, so that  
the maps themselves are random objects. $\maltese$

\end{example}

We start with a formal definition of the AQS automaton (AQS-A).
We consider a finite set of configurations $\mathcal{S}$, and label its elements by Greek bold letters, $\vec{\sigma}, \vec{\mu},$ {\em etc}. 
We introduce a pair of maps $F^\pm: \mathcal{S} \to \R \cup \{-\infty,\infty\}$, such that for all $\vec{\sigma} \in \mathcal{S}$, 
\begin{equation}
 F^-[\vec{\sigma}] < F^+[\vec{\sigma}].
 \label{eqn:stabilityI}
\end{equation}
The pair $F^\pm[\vec{\sigma}]$ associated with a configurations will be called its {\em trapping fields}.  We shall say that a configuration 
$\vec{\sigma}$ is stable at field $F$, if $ F \in (F^-[\vec{\sigma}], F^+[\vec{\sigma}])$. The AQS dynamics is implemented through a pair 
of maps $U$ and $D$: given a configuration $\vec{\sigma}$, we denote the configuration into which it transits when 
$F = F^+[\vec{\sigma}]$ as $U\vec{\sigma}$. Likewise, the transition when $F = F^-[\vec{\sigma}]$ will be denoted by $D\vec{\sigma}$. 
If $F^+[\vec{\sigma}] = \infty$, or $F^-[\vec{\sigma}] = -\infty$, then we say that $\vec{\sigma}$ is absorbing under $U$, respectively 
$D$. In this case it  is convenient to define $U\vec{\sigma} = \vec{\sigma}$, respectively $D\vec{\sigma} = \vec{\sigma}$, so that $U$ and $D$ 
become maps from $\mathcal{S}$ into $\mathcal{S}$. 
The AQS dynamics also requires a set of compatibility conditions between the pair of maps $(U,D)$ and the trapping fields:
\begin{align}
   F^+[U\vec{\sigma}] &>  F^+[\vec{\sigma}], \quad (U\vec{\sigma} \neq \vec{\sigma}),  \label{eqn:stabilityII}\\
   F^-[U\vec{\sigma}] &< F^+[\vec{\sigma}]
   , \label{eqn:stabilityIII}\\
   F^-[D\vec{\sigma}] &<  F^-[\vec{\sigma}], \quad (D\vec{\sigma} \neq \vec{\sigma}), \label{eqn:stabilityIV}\\
   F^+[D\vec{\sigma}] &> F^-[\vec{\sigma}]
   . \label{eqn:stabilityV}
\end{align}
Condition \eqref{eqn:stabilityII} and \eqref{eqn:stabilityIII} assure that the state $U\vec{\sigma}$ is indeed stable at force $F=F^+[\vec{\sigma}]$, 
{\em cf.} \eqref{eqn:stabilityI}, while \eqref{eqn:stabilityIV} and \eqref{eqn:stabilityV} assure the same for the transition to 
$D\vec{\sigma}$ when $F=F^-[\vec{\sigma}]$. Given a finite set of configurations $\mathcal{S}$, we say that a pair of trapping fields $F^\pm$ 
and a pair of maps $(U,D)$ are compatible with each other, if for all $\vec{\sigma} \in \mathcal{S}$, conditions \eqref{eqn:stabilityII} -- 
\eqref{eqn:stabilityV} hold. 

\begin{definition}[AQS automaton (AQS-A)]
 \label{def:AQSA}
 An AQS automaton is the quintuple $(\mathcal{S},F^\pm, U, D)$, where $\mathcal{S}$ is a finite set of configurations, the trapping field 
 $F^\pm: \mathcal{S} \to \R$ satisfy \eqref{eqn:stabilityI} and the maps $U$ and $D$ satisfy the compatibility conditions \eqref{eqn:stabilityII} -- 
\eqref{eqn:stabilityV}.
\end{definition}

Note that the AQS dynamics associates with each state $\vec{\sigma}$ two transitions, $\vec{\sigma} \mapsto U\vec{\sigma}$ and 
$\vec{\sigma} \mapsto D\vec{\sigma}$. 
It is thus natural to consider the {\em functional graphs} associated with the maps $U$ and $D$. These are the graphs whose vertex set is given 
by $\mathcal{S}$ and whose directed edges are the pairs $(\vec{\sigma},U\vec{\sigma})$ and $(\vec{\sigma},D\vec{\sigma})$, respectively. 
We will also refer to such graphs as state transition graphs. The state transition graphs underlie the adiabatic dynamics under some arbitrary time varying field $F(t)$. The AQS nature of the dynamics and conditions by \eqref{eqn:stabilityII} -- \eqref{eqn:stabilityV} then imply that 
the functional graphs are acyclic, except for the fixed points which are mapped into themselves under $U$ or $D$. The particular way in which 
the directed  acyclic graphs $U$ and $D$ are laid out on the common 
vertex set $\mathcal{S}$, determines the hysteretic properties of the system: how and to what extent it can respond to 
periodic forcing.

For simplicity let us henceforth assume that 
each of the maps $U$ and $D$ has precisely one such fixed point, which we will denote 
by $\vec{\omega}$ and $\vec{\alpha}$, respectively, 
so that\footnote
{
  In the case of the RFIM model Example \ref{ex:RFIM_I}, these fixed points are the saturated states 
  $\vec{\alpha} = \vec{-1}$ and $\vec{\omega} = \vec{+1}$. 
}
\begin{align}
 U\vec{\omega} &= \vec{\omega}, \label{eqn:Uomega}\\
 D\vec{\alpha} &= \vec{\alpha},\label{eqn:Dalpha}
\end{align}
and the corresponding trapping fields thus are
\begin{equation}
 F^+[\vec{\omega}] = +\infty, \quad \quad F^-[\vec{\alpha}] = -\infty.
 \label{eqn:Fpnabsorbing}
\end{equation}
From \eqref{eqn:stabilityII} and \eqref{eqn:stabilityIV} it now 
immediately follows that the functional graphs of $U$ and $D$ are trees rooted at $\vec{\omega}$ and $\vec{\alpha}$. 

The states $\vec{\alpha}$ and $\vec{\omega}$ are {\em absorbing}, since  
for every $\vec{\sigma} \in \mathcal{S}$, there exists a pair of least positive integers 
$n_\sigma$ and $m_\sigma$, such that $U^{n_\sigma} \vec{\sigma} = \vec{\omega}$, and $D^{m_\sigma} \vec{\sigma} = \vec{\alpha}$. 
\begin{figure}[t!]
  \begin{center}
    \includegraphics[width=4.2in]{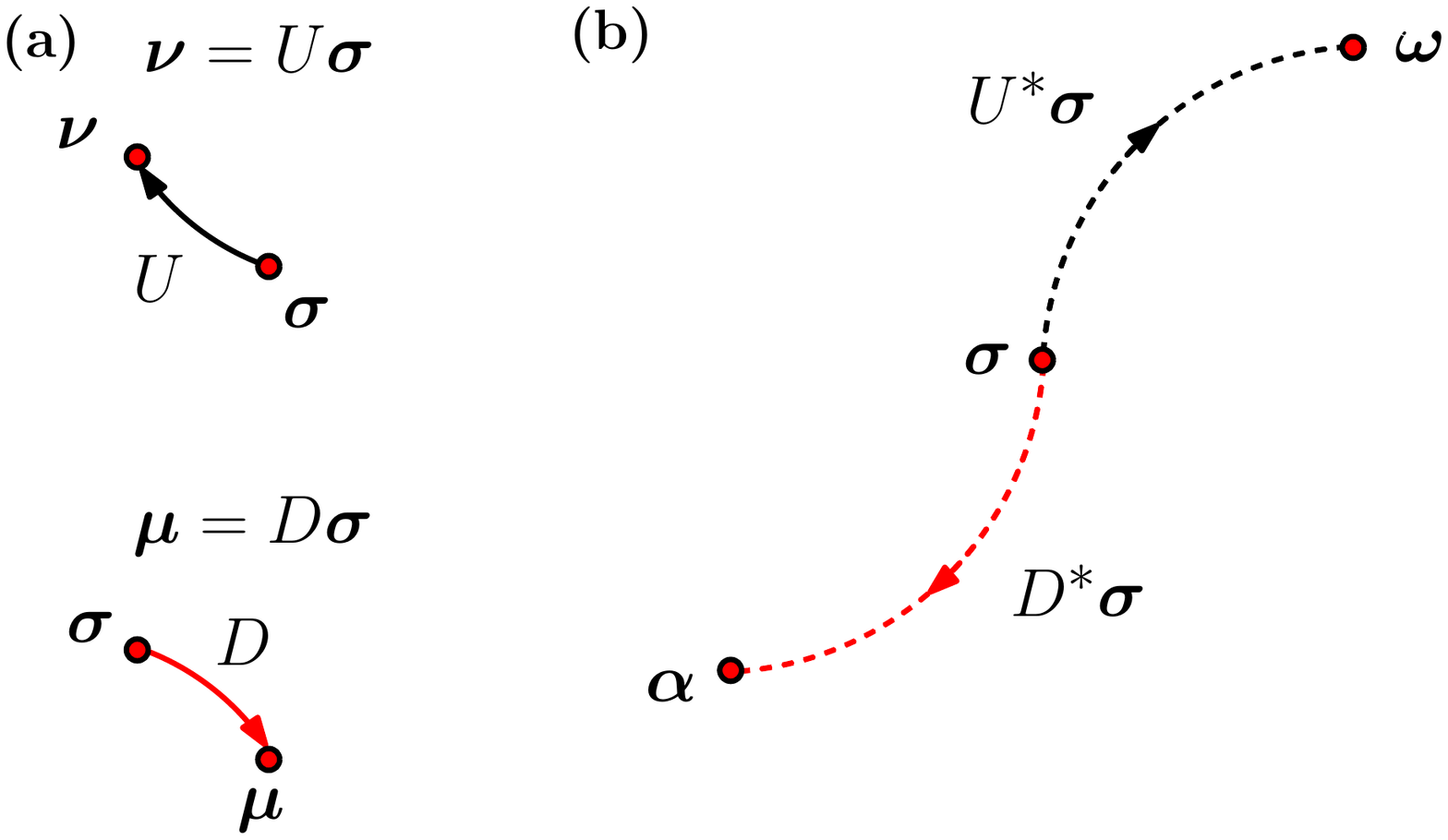} 
  \end{center}
  \caption{(a) Graphical representation of the relations $\vec{\nu} = U\vec{\sigma}$ and $\vec{\mu} = D\vec{\sigma}$. We use solid black and red arrows 
  to represent transitions under $U$ and $D$. (b) Graphical representation of the orbits $U^*\vec{\sigma}$ and $D^*\vec{\sigma}$. These orbits 
  necessarily terminate at the absorbing states $\vec{\omega}$ and $\vec{\alpha}$, respectively. We depict orbits by dashed black and red arrows.} 
  \label{fig:UDorbits}
\end{figure}
Given a state $\vec{\sigma}$, it is useful to define the sequence of states obtained by successive applications of $U$ and $D$, respectively.

\begin{definition}[$U$- and $D$-orbit]
 Given a state $\vec{\sigma}$ we define its $U$-orbit $U^*\vec{\sigma}$ as the sequence of distinct states
 \begin{equation}
  U^*\vec{\sigma} = (\vec{\sigma}, U\vec{\sigma}, U^2\vec{\sigma}, \ldots U^{n_\sigma}\vec{\sigma} = \vec{\omega}), 
 \end{equation}
 where, $n_\sigma$ is the smallest positive integer for which $U^{n_\sigma} \vec{\sigma} = \vec{\omega}$. 
Likewise, the $D$-orbit of $\vec{\sigma}$ is defined as
 \begin{equation}
  D^*\vec{\sigma} = ( \vec{\sigma}, D\vec{\sigma}, D^2\vec{\sigma}, \ldots D^{m_\sigma} \vec{\sigma} = \vec{\alpha}), 
 \end{equation}
with $m_\sigma$ being the smallest positive integer for which $D^{m_\sigma} \vec{\sigma} = \vec{\alpha}$.
\end{definition}

Fig.~\ref{fig:UDorbits} (a) and (b) depict the graphical representation of the actions of $U$ and $D$ on a state $\vec{\sigma}$ and the associated 
orbits $U^*\vec{\sigma}$ and $D^*\vec{\sigma}$. Sometimes it will be useful to think of the orbits as sets and we will write $\vec{\eta} \in U^*\vec{\sigma}$ 
to indicate that the configuration $\vec{\eta}$ is in the $U$-orbit of $\vec{\sigma}$, that is, $\vec{\eta} = U^i \vec{\sigma}$, for some integer 
$i \geq 0$. We next define a loop and its boundary states. 

\begin{figure}[t!]
  \begin{center}
    \includegraphics[width=3.2in]{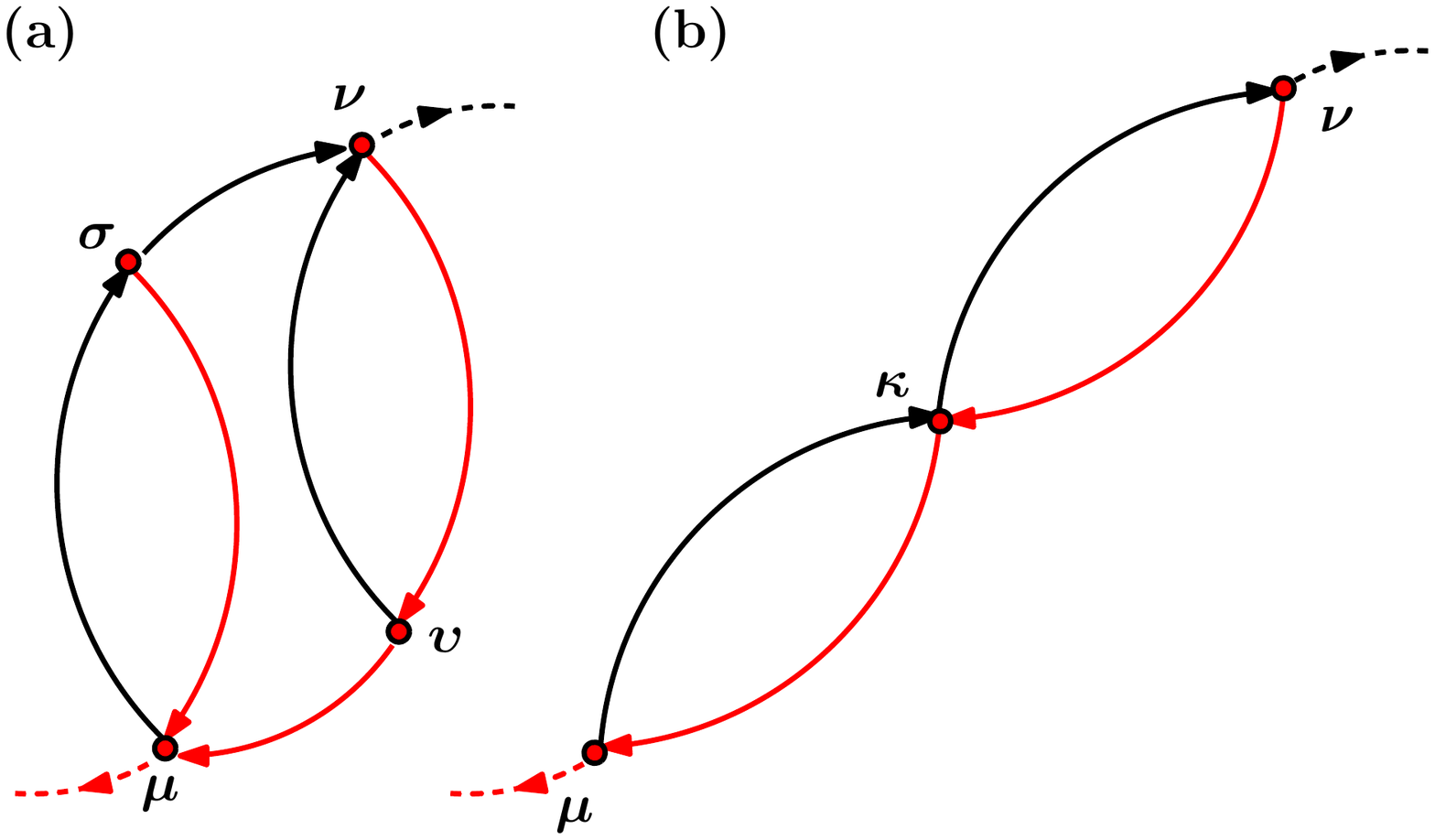} 
  \end{center}
  \caption{Examples for loops with and without intersection points, panels (a) and (b), respectively. Shown in each panel is a loop 
  $(\vec{\mu},\vec{\nu})$, where $\vec{\mu}$ and 
  $\vec{\nu}$ are its end points. (a) The $U$-boundary of the loop consists of the boundary states $\vec{\mu}, \vec{\sigma}$, and $\vec{\nu}$. 
  The state $\vec{\sigma}$ not being an 
  end point, is called an intermediate boundary state of $(\vec{\mu},\vec{\nu})$. (b) Example when the loop  contains an intersection point $\vec{\kappa}$. 
  The intermediate state $\vec{\kappa}$ is an intersection point, since it lies on both boundaries of the loop.} 
  \label{fig:loopDef}
\end{figure}

\begin{definition}[$(\vec{\mu},\vec{\nu})$-loop, loop boundaries, and boundary states]
\label{def:loopdef}
 A pair of states $(\vec{\mu},\vec{\nu})$ forms a $(\vec{\mu},\vec{\nu})$-loop, if $\vec{\nu} \in U^*\vec{\mu}$ and $\vec{\mu} \in D^*\vec{\nu}$. In 
 this case we call $\vec{\mu}$ and $\vec{\nu}$ the {\em endpoints} of the loop. 
 We will refer to the subsets $\{\vec{\mu}, U\vec{\mu}, \ldots U^n\vec{\mu} = \vec{\nu} \} \subset U^*\vec{\mu}$ and 
 $\{\vec{\nu}, D\vec{\mu}, \ldots D^m\vec{\nu} = \vec{\mu} \} \subset D^*\vec{\nu}$ as the $U$-, respectively $D$-{\em boundaries} of the loop and 
 denote the states they contain boundary states. Boundary states that are not endpoints will be referred to as {\em intermediate 
 states}. An intermediate node belonging to both boundaries will be called an {\em intersection state}.  
\end{definition}
As an illustration, consider the two loops depicted in Fig.~\ref{fig:loopDef}. The loop in (a) has endpoints $\vec{\mu}$ and $\vec{\nu}$. 
The $U$- and $D$-boundaries of this loop are formed by the nodes 
$(\vec{\mu},\vec{\sigma},\vec{\nu})$ and $(\vec{\mu},\vec{\upsilon},\vec{\nu})$. The boundary nodes $\vec{\sigma}$ and $\vec{\upsilon}$ are intermediate states. 
In the loop depicted in panel (b), the state $\vec{\kappa}$ is an intersection node, as it lies on both boundaries of the loop.

Consider next the dynamics under a time-varying field $(F_t)_{0 \leq t \leq T}$ for some time $T > 0$. 
We will always assume 
that the forcing is such that the initial state is stable at $F_0$\footnote{
This is consistent with AQS dynamics which assumes that for almost all times the system 
is in a quasi-static configuration $\vec{\sigma} \in \mathcal{S}$, punctuated by abrupt 
transitions between these. 
 }.
Suppose first that $(F_t)_{0 \leq t \leq T}$ is non-decreasing. Given an initial state 
$\vec{\sigma}$, we define the action of $(F_t)_{0 \leq t \leq T}$ as an evolution along the orbit $U^*\vec{\sigma}$ until the 
first state $\vec{\upsilon}$ is reached for which $F^+(\vec{\upsilon}) > F_T$. We write this as 
\begin{equation}
 \mathcal{U}[F_T]\vec{\sigma} = \min \{ \vec{\eta} \in U^*\vec{\sigma} : F^+(\vec{\eta}) > F_T \},
 \label{eqn:FUaction}
\end{equation}
and note that the action of $\min$ on the right hand side is supposed to extract the first element in the ordered sequence of states $U^*\vec{\sigma}$ for which 
the given condition is satisfied\footnote{The choice of the $\min$ and $\max$ notation in \eqref{eqn:FUaction} and \eqref{eqn:FDaction} alludes to a 
natural partial order on an orbit and will be discussed further below. }. 
Note that a
 state transition occurs only if $F_T \geq F^+(\vec{\sigma})$. From \eqref{eqn:FUaction} we also see that resulting state  
 only depends on the value of the force $F_T$ at the endpoint, as required by the adiabaticity assumption underlying our AQS dynamics \cite{Sethna93}.  
Supposing next that $(F_t)_{0 \leq t \leq T}$ is non-increasing, we define its action on an initial state $\vec{\sigma}$ as the first state 
$\vec{\eta}$ on the orbit $D^*\vec{\sigma}$ for which $F^-(\vec{\eta}) < F_T$ and write this 
as
\begin{equation}
 \mathcal{D}[F_T]\vec{\sigma} = \max \{ \vec{\eta} \in D^*\vec{\sigma} : F^-(\vec{\eta}) < F_T \}.
 \label{eqn:FDaction}
\end{equation}
Having defined the AQS-dynamics under monotone forces, the generalization to some force $(F_t)_{0 \leq t \leq T}$ with finitely many intervals of monotonicity is clear: 
We break the time interval $[0,T]$ into the disjoint intervals in which the forcing is non-decreasing or non-increasing and accordingly apply \eqref{eqn:FUaction} and \eqref{eqn:FDaction}.

We now turn to the no-passing (NP) property which requires a partial order\footnote
{Recall that a partial order $\preceq$ on a set $\mathcal{S}$ is a reflexive, anti-symmetric and transitive 
order relation: for all $a, b, c \in \mathcal{S}$, 
\begin{align*}
 a \preceq a, \quad &\mbox{(reflexivity)}, \\
 a \preceq b, \, b \preceq a  \,\, \Rightarrow a = b, \quad &\mbox{(anti-symmetry)}, \\
 a \preceq b, \, b \preceq c  \,\, \Rightarrow a \preceq c, \quad &\mbox{(transitivity)}. 
\end{align*}
The strict partial order $a \prec b$ is defined by  $a \preceq b$ and $a \neq b$.
}
%
%
on $\mathcal{S}$ . 
Suppose therefore that $\mathcal{S}$ admits a partial order $\preceq$ such that $\vec{\alpha}$ and $\vec{\omega}$ are 
the minimal, 
respectively maximal elements with respect to this partial ordering:  $\vec{\alpha} \preceq \vec{\sigma}$ and $\vec{\sigma} \preceq \vec{\omega}$, 
for all $\vec{\sigma} \in \mathcal{S}$. 
We say that the partial order $\preceq$ 
is compatible with $U$ and $D$ if the following hold:
\begin{align}
 \vec{\sigma} &\preceq U\vec{\sigma}, 
 \label{eqn:NCU}\\
 D\vec{\sigma} &\preceq \vec{\sigma}, 
 \label{eqn:NCD}
\end{align}
with the equality holding in \eqref{eqn:NCU} and \eqref{eqn:NCD}, if and only if $\vec{\sigma} = \vec{\omega}$  and 
$\vec{\sigma} = \vec{\alpha}$, respectively.

\begin{definition}[No-passing property for AQS dynamics]
\label{def:NP}
Let $(\mathcal{S}, F^\pm, U, D)$ be an AQS-A for which $\preceq$ is a compatible partial order. We say $(\mathcal{S}, F^\pm, U, D)$
has the no-passing property with respect to $\preceq$, if for any pair of states $\vec{\eta}_0 \preceq \vec{\sigma}_0$  
and their respective forcings $(F^{(1)}_t)_{0 \leq t \leq T}$ and  $(F^{(2)}_t)_{0 \leq t \leq T}$ such that $F^{(1)}_t \leq F^{(2)}_t$
for all times $0 \leq t \leq T$, 
we have 
\begin{equation*}
 \vec{\eta}_t \preceq \vec{\sigma}_t, \quad \quad 0 \leq t \leq T.
\end{equation*}
\end{definition}

Given a partial order $\preceq$ compatible with $U$ and $D$, Middleton's {\em no-passing property} \cite{NoPassing} implies the 
following:

\begin{lemma}[No-passing property for AQS-A orbits]\label{lem:NP}
 Let $(\mathcal{S}, F^\pm, U, D)$ be an AQS-A satisfying the no-passing property with the  compatible partial order $\preceq$. 
 Then for 
 any triplet of distinct states $\vec{\mu}, \vec{\sigma}, \vec{\nu}$ such that 
 \begin{equation}
  \vec{\mu} \prec \vec{\sigma} \prec \vec{\nu},
  \label{eqn:tripletOrder}
 \end{equation}
the following two conditions hold:
\begin{itemize}
 \item [(a)] if $\vec{\mu} \in D^*\vec{\nu}$ then $\vec{\mu} \in D^*\vec{\sigma}$,
 \item [(b)] if $\vec{\nu} \in U^*\vec{\mu}$ then $\vec{\nu} \in U^*\vec{\sigma}$.
\end{itemize}
\end{lemma}

\begin{figure}[t!]
  \begin{center}
    \includegraphics[width = 3in]{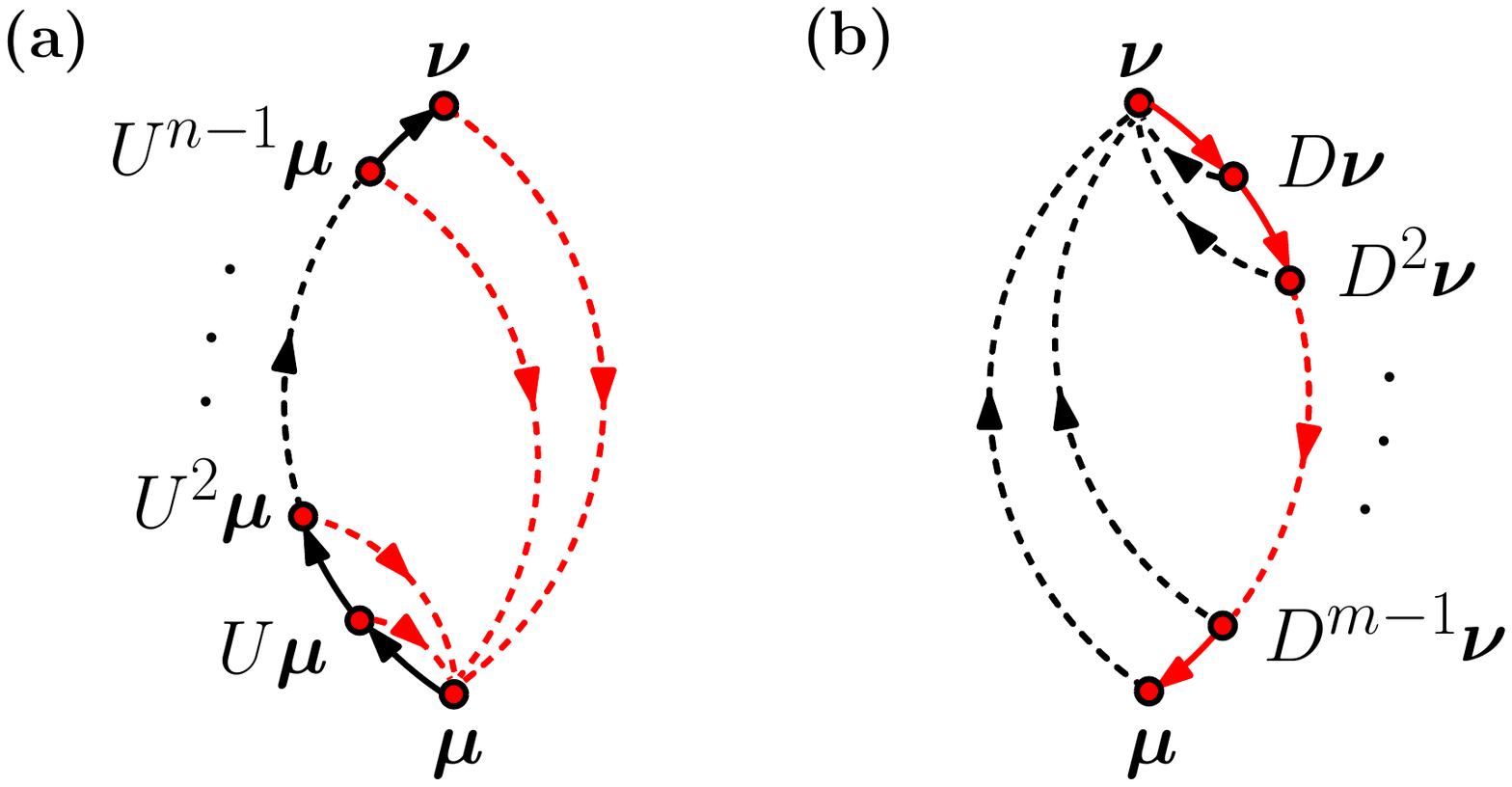} 
  \end{center}
  \caption{Illustration of the return-point memory (RPM) property Def.~\ref{def:RPM} for the loop formed by   
  states $\vec{\mu}$ and $\vec{\nu}$. The panels (a) and (b) corresponds to the two cases of Def.~\ref{def:RPM}. (a) 
  For any  state $\vec{\sigma} = U^i\vec{\mu}$ located on the orbit $U^*\vec{\mu}$ between $\vec{\mu}$ and $\vec{\nu}$, 
    its $D$-orbit $D^*\vec{\sigma}$ has to pass through $\vec{\mu}$. (b) For any state 
   $\vec{\sigma} = D^j\vec{\nu}$ located on the orbit $D^*\vec{\nu}$ and between 
  $\vec{\mu}$ and $\vec{\nu}$, its orbit $U^*\vec{\sigma}$ has to pass through  $\vec{\nu}$.
  } 
  \label{fig:RPMUD}
\end{figure}

\begin{example}[The random-field Ising Model -- Partial order and no-passing]
 \label{ex:RFIM_PO}

The following partial order $\preceq$
can be defined on the set of 
configurations: we say that $\vec{\sigma} \preceq \vec{\sigma'}$, if  $\sigma_i \leq \sigma'_i$ for all spin 
components.
The AQS dynamics of the RFIM is compatible with this partial order: the transition $\vec{\sigma'} = U\vec{\sigma}$ 
triggered by the minimal field increment to 
$F = F^+(\vec{\sigma})$ produces a configuration $\vec{\sigma'}$  with $\vec{\sigma} \preceq \vec{\sigma'}$. An analogous 
result holds for minimal field decrements to $F = F^-(\vec{\sigma})$, resulting in a  state with 
$\vec{\sigma''} \preceq \vec{\sigma}$. As can be easily shown  \cite{Sethna93}, the AQS dynamics of the RFIM has the no-passing property, 
Def.~\ref{def:NP}, with respect to this partial order. $\maltese$ 
\end{example}

The return point memory (RPM) property emerges as a special case 
of Lemma \ref{lem:NP} in which the ordering \eqref{eqn:tripletOrder} is achieved by 
assuming that all three states are either on the orbit $U^*\vec{\mu}$ or $D^*\vec{\nu}$. To be more precise, let us assume 
that all three states are on $U^*\vec{\mu}$. Lemma \ref{lem:NP}(a) then asserts that starting at $\vec{\nu}$, decreasing the field to reach $\vec{\mu}$, 
and then performing a {\em switch-back} at $\vec{\mu}$ by increasing the field  to reach the state $\vec{\sigma}$, a subsequent switch-back by decreasing the field  
will cause the system to eventually reach the state $\vec{\mu}$ of the first switch-back. In other words, the first switch-back $\vec{\mu}$ 
is remembered, since a second switch-back at $\vec{\sigma}$, satisfying \eqref{eqn:tripletOrder} causes the state trajectory to return to it. A similar 
argument applies when all three states are on $D^*\vec{\nu}$, we start in $\vec{\mu}$, increase the field to reach $\vec{\nu}$ and perform a switch-back to 
reach-state $\vec{\sigma}$. Lemma \ref{lem:NP}(b) then asserts that upon a second switch-back the first switch-back state $\vec{\nu}$ is eventually reached. 
We now define the RPM property in the form that will be useful in what follows. 
\begin{definition}[loop RPM property $\ell$RPM]\label{def:RPM}
 Let $(\mathcal{S}, F^\pm, U, D)$ be an AQS-A
 We say that the AQS-A 
 possesses the (loop) {\em return point memory property}, if the following two hold:
 \begin{itemize}
  \item [(a)] For any two states $\vec{\sigma}$ and $\vec{\nu}$ on the $U$-orbit of $\vec{\mu}$, such that 
  $\vec{\sigma}$ is reached before $\vec{\nu}$: 
  if $\vec{\mu} \in D^*\vec{\nu}$ then $\vec{\mu} \in D^*\vec{\sigma}$. 
  \item [(b)] For any two states $\vec{\upsilon}$ and $\vec{\mu}$ on the $D$-orbit of $\vec{\nu}$, such that 
  $\vec{\upsilon}$ is reached before $\vec{\mu}$:  
  if $\vec{\nu} \in U^*\vec{\mu}$ then $\vec{\nu} \in U^*\vec{\upsilon}$. 
  \end{itemize}
\end{definition}

It is clear that the NP property, Def.~\ref{def:NP}, implies via Lemma \ref{lem:NP} the $\ell$RPM property of Def.~\ref{def:RPM}, 
since 
any three distinct states on a $U$ or $D$ orbit are necessarily ordered with respect to $\prec$. The $\ell$RPM property 
is illustrated in Fig.~\ref{fig:RPMUD}. 

Observe that the $\ell$RPM property as given by Def.~\ref{def:RPM} does not require the presence of a partial order on $\mathcal{S}$. It is 
a statement about a pair of states $(\vec{\mu},\vec{\nu})$ forming the endpoints of a loop, Def.~\ref{def:loopdef}, 
as depicted in Fig.~\ref{fig:RPMUD}. A loop can be naturally associated with a partial order that is moreover based entirely 
on the maps $U$ and $D$ and hence induced by the dynamics. 

\begin{definition}[Dynamic partial orders, $\preceq_U$, $\preceq_D$ and $<$]\label{def:dynPO}
 Let $(\mathcal{S}, F^\pm, U,D)$ be an AQS-A. The orders $\preceq_U$ and $\preceq_D$ associated with the maps $U$ and $D$ are defined as 
 \begin{align}
  \vec{\nu} \in U^*\vec{\mu} \quad &\Leftrightarrow \quad \vec{\mu} \preceq_U \vec{\nu}, \label{eqn:dynU}\\ 
  \vec{\mu} \in D^*\vec{\nu} \quad &\Leftrightarrow \quad \vec{\mu} \preceq_D \vec{\nu}. \label{eqn:dynD} 
 \end{align}

The intersection of $\prec_U$ and $\prec_D$, namely the set of pairs of states $(\vec{\mu}, \vec{\nu})$ for which 
both $\vec{\mu} \prec_U \vec{\nu}$ and $\vec{\mu} \prec_D \vec{\nu}$, is another partial order which will be called the {\em loop partial order} and 
denoted by $<$. Thus 
\begin{equation}
 \vec{\mu} < \vec{\nu} \quad \Leftrightarrow \quad \vec{\mu} \prec_U \vec{\nu} \,\, \mbox{and} \,\, \vec{\mu} \prec_D \vec{\nu}.  \label{eqn:loopOrder}
\end{equation}
\end{definition}

It is readily checked that by nature of their acyclic functional graphs, the orders $\preceq_U$ and $\preceq_D$, induced by $U$ and $D$, are indeed partial orders. 
Hence the intersection $<$ of these partial orders defined in \eqref{eqn:loopOrder} is a partial order, too\footnote{We prefer to keep $<$ as a strict partial order, since the case of equality 
implies a loop whose endpoints coincide, a {\em singleton loop}.}. 
Recall that 
given a pair of partial orders defined on a set $\mathcal{S}$ one is said to be {\em stronger} than the other, if 
the pair of states ordered with respect to the former is a subset of those ordered with respect to the latter order. 
Thus the partial order $<$ is stronger than 
 $\prec_U$ and $\prec_D$. Moreover, the partial order $<$ is necessarily stronger than any partial order $\prec$ on $\mathcal{S}$ that is compatible with $U$ and $D$, {\em i.e.} a partial order 
that satisfies both \eqref{eqn:NCU} and \eqref{eqn:NCD}, since $\prec_U$ and $\prec_D$ are by definition each stronger than $\prec$. 

\begin{example}[The random-field Ising Model -- Partial orders $\prec$ and $<$]
 \label{ex:RFIM_PO_II}
 For the RFIM the partial order  $\prec$ has been defined in Example \ref{ex:RFIM_PO}. Here we will show that in general the partial 
 order $<$ is stronger than $\prec$.  
 First note that $\vec{\mu} < \vec{\nu}$ implies that $\vec{\mu} \prec \vec{\nu}$. However the reverse is generally not true. 
 There are states that are comparable with respect to $\prec$ but not $<$. To see this, observe that not every  pair of states such that 
 $\vec{\mu} \prec \vec{\nu}$ has to lie on a $U$- or $D$-orbit. Hence such pairs of states need  not be 
 comparable with respect to $\prec_U$ or $\prec_D$. Consequently such pairs need not be comparable with respect to $<$. Thus the 
 partial order $<$ is generally stronger than $\prec$. This is not surprising, since the partial $\prec$ of the RFIM is based entirely on 
 the spin configurations and  therefore is defined without 
 the dynamics, while the partial order $<$ is a consequence of the dynamics.  
 $\maltese$ 
\end{example}

It is apparent now that the $\ell$RPM property in Def.~\ref{def:RPM} is a statement about the dynamic 
partial orders $<$, $\prec_U$ and $\prec_D$ and can therefore equivalently be defined in terms of these as:
\begin{definition}[$\ell$RPM via dynamic partial orders]\label{def:RPMII}
 The  AQS-A $(\mathcal{S}, F^\pm, U, D)$ 
 possesses the (loop) {\em return point memory property}, if the following two hold: 
 
 For any pair of states $\vec{\mu} < \vec{\nu}$ and some state $\vec{\sigma}$, 
 \begin{align}
  \vec{\mu} \prec_U \vec{\sigma} \prec_U \vec{\nu} \quad &\Rightarrow \quad \vec{\mu} < \vec{\sigma}, \\ 
  \vec{\mu} \prec_D \vec{\sigma} \prec_D \vec{\nu} \quad &\Rightarrow \quad \vec{\sigma} < \vec{\nu}.   
 \end{align}
\end{definition}
At the same time, this formulation makes clear why AQS dynamics with the NP property can only imply $\ell$RPM, but $\ell$RPM does not necessarily imply NP. 
The partial order $\prec$ associated with the NP property is only required to be compatible with $U$ and $D$ and therefore necessarily weaker than the partial order $<$ induced by the dynamics. 
Since on the other hand Definition \ref{def:RPMII} for $\ell$RPM only involves the dynamic partial orders, the presence of $\ell$RPM will not necessarily imply an NP property, or if it does, its 
scope will be restricted to loops. 

\begin{definition}[$\ell$AQS-A]\label{def:lAQSA}
An $\ell$AQS-A is an AQS automaton $(\mathcal{S},F^\pm,U,D)$ having the $\ell$RPM property, as defined by  
Defs.~\ref{def:dynPO} and \ref{def:RPMII}. 
\end{definition}

Unless otherwise noted, for the rest of this paper we shall only consider $\ell$AQS-A. We will use 
the terms RPM and $\ell$RPM interchangeably, choosing the latter if we want to emphasize its origin in the dynamic partial orders $(\preceq_U, \preceq_D, <)$, 
as introduced in Def.~\ref{def:dynPO}. 
We conclude this section with some further definitions and turn in the following section to an analysis of the 
structure of the state transition graph associated with $\ell$AQS-A loops.

Given a state $\vec{\sigma}$, we can obtain a new state $\vec{\nu}$ by applying a sequence of $U$ and $D$ operations
\begin{equation}
 \vec{\nu} = U^{n_p} D^{m_p} U^{n_{p-1}} D^{m_{p-1}} \cdots   D^{m_1} U^{n_1} \vec{\sigma}, \label{eqn:fieldH}  
\end{equation}
for some positive integers $p$,  $n_1, m_1, \ldots, m_{p-1}, n_{p-1}, m_p$, and  
$n_p \geq 0$. We will denote the sequence of operations in \eqref{eqn:fieldH} as 
a {\em field-history} and denote these by capital Greek letters. 
The finiteness of $\mathcal{S}$ assures that any field-history can always be reduced to a finite number of $U$ and $D$ operations 
by eliminating multiply traversed intermediate states. If there 
exists a field-history $\Phi$ that maps $\vec{\sigma}$ to $\vec{\nu}$, we say that $\vec{\nu}$ is {\em field-reachable} from $\vec{\sigma}$.

Among the set of all configurations $\mathcal{S}$, there is a special subset $\mathcal{R}$ of states that is field-reachable from $\vec{\alpha}$ or $\vec{\omega}$. 
We call this set of states $(\vec{\alpha},\vec{\omega})$ -- reachable, or more simply, reachable states. 

\begin{definition}[Reachable states $\mathcal{R}$]
 Let $(\mathcal{S}, F^\pm,U, D)$ be an $\ell$AQS-A. The subset of reachable states $\mathcal{R} $ is defined as
 \begin{equation}
  \mathcal{R} = \left \{ \vec{\sigma} \in \mathcal{S} : \mbox{there exists a field-history } \Phi {:}\, \vec{\alpha} \mapsto \vec{\sigma} \right \}.
 \end{equation}

\end{definition}

Note that by their absorbing properties, $\vec{\alpha}, \vec{\omega} \in \mathcal{R}$. Thus if a state is field-reachable from 
$\vec{\alpha}$, it necessarily is also reachable from $\vec{\omega}$ and 
{\em vice versa}. Moreover, on $\mathcal{R}$ field-reachability is an equivalence 
relation \cite{MagniBasso2005,Bertottietal2007,Bortolottietal2008}. 

\begin{lemma}[Mutual reachability on $\mathcal{R}$ as an equivalence relation]
\label{lem:mutreach}
 Consider any pair of reachable states  $\vec{\sigma}_1, \vec{\sigma}_2 \in \mathcal{R}$. Then there exist field histories  
 that map $\vec{\sigma}_1$ and $\vec{\sigma}_2$ into each other and we say that $\vec{\sigma}_1$ and $\vec{\sigma}_2$ 
 are mutually reachable. Moreover, mutual reachability is an 
 equivalence relation and 
 $\mathcal{R}$ is closed under the actions of $U$ and $D$. 
\end{lemma}

\begin{figure}[t!]
  \begin{center}
    \includegraphics[width=4in]{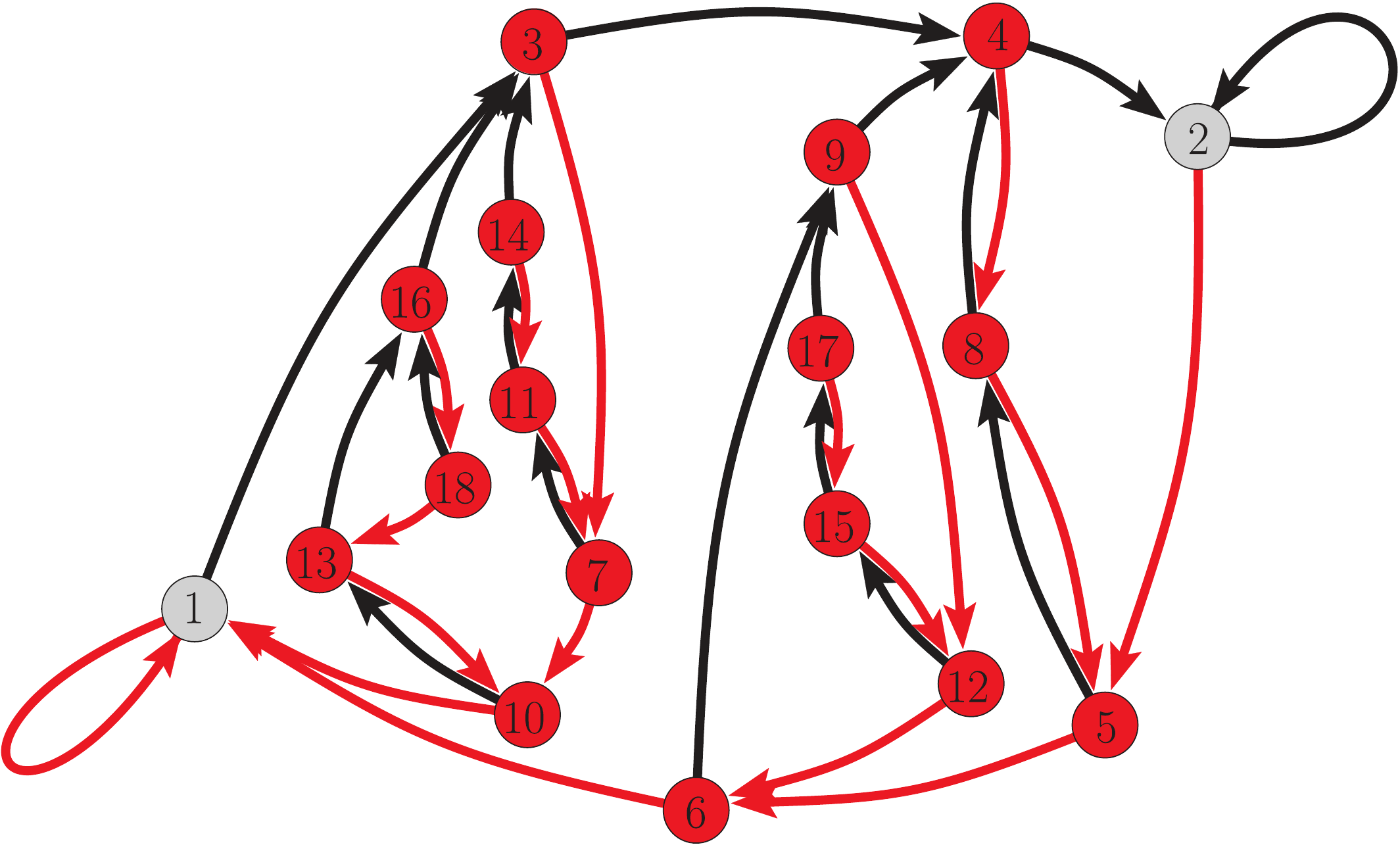} 
  \end{center}
  \caption{Sample reachability graph obtained from the toy model of depinning \cite{KMShort13,KMLong13}. Each vertex is a reachable configuration.  The transitions 
  under $U$ and $D$ correspond to  increments and decrements of the driving force. 
  They are indicated by black and red arrows, respectively.
  The $D$- and $U$-absorbing states are colored in gray and labeled as vertices 1 and 2. 
  } 
  \label{fig:reachabilitygraph}
\end{figure}

By restricting the vertex set to $\mathcal{R}$, we obtain the $\ell$AQS-A state transition graph on the set of reachable states, the 
{\em reachability graph}. Figure \ref{fig:reachabilitygraph} depicts a typical reachability graph 
that was actually obtained from a realization of the toy model for depinning \cite{KMShort13, KMLong13}. 
As in the RFIM example, the functional graphs obtained by considering only transitions under $U$ (black arrows), or transitions under $D$ (red arrows)
are trees rooted at the upper and lower absorbing states $\vec{\omega}$ and $\vec{\alpha}$, respectively. In the 
figure these absorbing states are labeled as $2$ and $1$, respectively. 
The state transition graph $\mathcal{R}$ has the a 
double-tree structure, {\em i.e.} two trees laid out on a common set of vertices. 
The intertwining of these two trees results in the loops and their nesting is a consequence of the RPM property.

\section{$\ell$AQS-A intra-loop structure}
\label{sec:damaloop}

In this section we derive the intra-loop structure of $\ell$AQS-A. 
The $\ell$RPM property allows us to describe the topology of loop substructure
and we start by defining field histories confined to a loop $(\vec{\mu},\vec{\nu})$. 

\begin{definition}[Loop-confined field-history]
\label{def:munuconfined}
 Let  $(\vec{\mu},\vec{\nu})$ be the endpoints of a loop and consider a field-history starting from $\vec{\mu}$ and leading to 
 $\vec{\sigma}$, 
 \begin{equation}
    \vec{\sigma} = U^{n_p} D^{m_{p-1}} U^{n_{p-1}} D^{m_{p-2}} \cdots   D^{m_1} U^{n_1} \vec{\mu},  
    \label{eqn:sigmaFH}
 \end{equation}
 with integers $p, n_p \geq 0$ and $n_1, m_1, n_2, \ldots, n_{p-1}, m_{p-1} > 0$. 
 The sequence of states 
 \begin{align*}
  U\vec{\mu}&, U^2\vec{\mu}, \ldots, U^{n_1}\vec{\mu}, DU^{n_1}\vec{\mu}, D^2U^{n_1}\vec{\mu}, \ldots D^{m_1}U^{n_1}\vec{\mu}, \\
  \ldots&, 
  U^{n_p-1} D^{m_{p-1}} U^{n_{p-1}} D^{m_{p-2}} \cdots   D^{m_1} U^{n_1} \vec{\mu}
 \end{align*}
 are the {\em transit states} associated with the field-history. We say that the field-history starting from an endpoint of the loop $(\vec{\mu}, \vec{\nu})$ is 
{\em loop}- or $(\vec{\mu}, \vec{\nu})$-confined, if the transit states of the trajectory avoid the endpoints, that is 
$\vec{\mu}$ and $\vec{\nu}$ are not transit states of the field-history. The subset of transit states given by  
\begin{align}
 \label{eqn:sigmaSB}
 \vec{\nu}_1 &= U^{n_1} \vec{\mu}, \nonumber \\
 \vec{\nu}_2 &= U^{n_2} D^{m_1} U^{n_1} \vec{\mu}, \nonumber \\
     &\cdots \nonumber \\
 \vec{\nu}_{p-1} &= U^{n_{p-1}} D^{m_{p-2}} \cdots U^{n_2} D^{m_1} U^{n_1} \vec{\mu}, \nonumber \\ 
 \\
 \vec{\mu}_1 &= D^{m_1} U^{n_1} \vec{\mu}, \nonumber \\
 \vec{\mu}_2 &= D^{m_2}U^{n_2} D^{m_1} U^{n_1} \vec{\mu}, \nonumber \\
     &\cdots \nonumber \\
 \vec{\mu}_{p-1} &= D^{m_{p-1}}U^{n_{p-1}} D^{m_{p-2}} \cdots U^{n_2} D^{m_1} U^{n_1} \vec{\mu}. \nonumber 
\end{align}
are called  the {\em switch-back states} of the field-history. 
For a field-history starting from $\vec{\nu}$, 
 \begin{equation}
    \vec{\sigma} = U^{n_p} D^{m_{p-1}} U^{n_{p-1}} D^{m_{p-2}} \cdots   D^{m_1} \vec{\nu},  
    \label{eqn:sigmaFH2}
 \end{equation}
transient states, switch-back states and $(\vec{\mu}, \vec{\nu})$-confined field histories are defined in a similar manner.
\end{definition}

Note that if a field-history \eqref{eqn:sigmaFH} is $(\vec{\mu},\vec{\nu})$-confined then so are field histories terminating at any of its 
transient states. We define next the set of $(\vec{\mu},\vec{\nu})$-reachable states $\mathcal{R}_{\vec{\mu}, \vec{\nu}}$. 

\begin{definition}[$(\vec{\mu},\vec{\nu})$-reachable states]
\label{def:munureachable}
 Let  $(\vec{\mu},\vec{\nu})$ be the endpoints of a loop. The set $\mathcal{R}_{\vec{\mu}, \vec{\nu}}$ of 
 $(\vec{\mu},\vec{\nu})$ -- reachable states is defined as  
\[
 \mathcal{R}_{\vec{\mu}, \vec{\nu}} = \{ \vec{\sigma} \in \mathcal{S} : 
\mbox{there is a  
 $(\vec{\mu}, \vec{\nu})$-confined field-history leading to $\vec{\sigma}$} 
\}.
\]
\end{definition}

\begin{proposition}[Orbits of $(\vec{\mu},\vec{\nu})$ -- reachable states]\label{prop:munuorbitLemma}
Let $\mathcal{R}_{\vec{\mu}, \vec{\nu}}$ be the set of reachable states associated with the loop  $(\vec{\mu},\vec{\nu})$. 
For any $\vec{\sigma} \in \mathcal{R}_{\vec{\mu}, \vec{\nu}}$, 
\[
\vec{\nu} \in U^*\vec{\sigma}\quad \mbox{and} \quad \vec{\mu} \in D^*\vec{\sigma}. 
\]
\end{proposition}

Note that the proposition is trivial 
in the special case when  the endpoints of the loop are $\vec{\mu} = \vec{\alpha}$ and $\vec{\nu} = \vec{\omega}$, since the result follows 
from the absorbing property of these states. The proof of the proposition for general endpoints $(\vec{\mu},\vec{\nu})$ relies on the 
$\ell$RPM property.  
Proposition \ref{prop:munuorbitLemma} asserts that the orbits of all $(\vec{\mu}, \vec{\nu})$-reachable states leave 
$\mathcal{R}_{\vec{\mu},\vec{\nu}}$ through the endpoints of the loop. Thus any pair of $(\vec{\mu}, \vec{\nu})$-reachable states 
are mutually reachable, {\em cf. } Lemma \ref{lem:mutreach} and its proof. At the same time Proposition \ref{prop:munuorbitLemma} 
provides a notion for the ''interior`` of a loop.

\begin{definition}[Interior states of a loop]\label{def:boundaryinterior}
 Consider the set $\mathcal{R}_{\vec{\mu}, \vec{\nu}}$ of reachable states associated with a loop $(\vec{\mu},\vec{\nu})$. A state 
 $\vec{\sigma} \in \mathcal{R}_{\vec{\mu}, \vec{\nu}}$ that is not a boundary state will be called an {\em interior state}. Thus elements of 
 $\mathcal{R}_{\vec{\mu}, \vec{\nu}}$ are either interior or boundary states. 
\end{definition}

We now have the following proposition for the nesting of interior loops, as illustrated in Fig.~\ref{fig:Nesting}:

\begin{figure}[t!]
  \begin{center}
    \includegraphics[width = 0.8\textwidth]{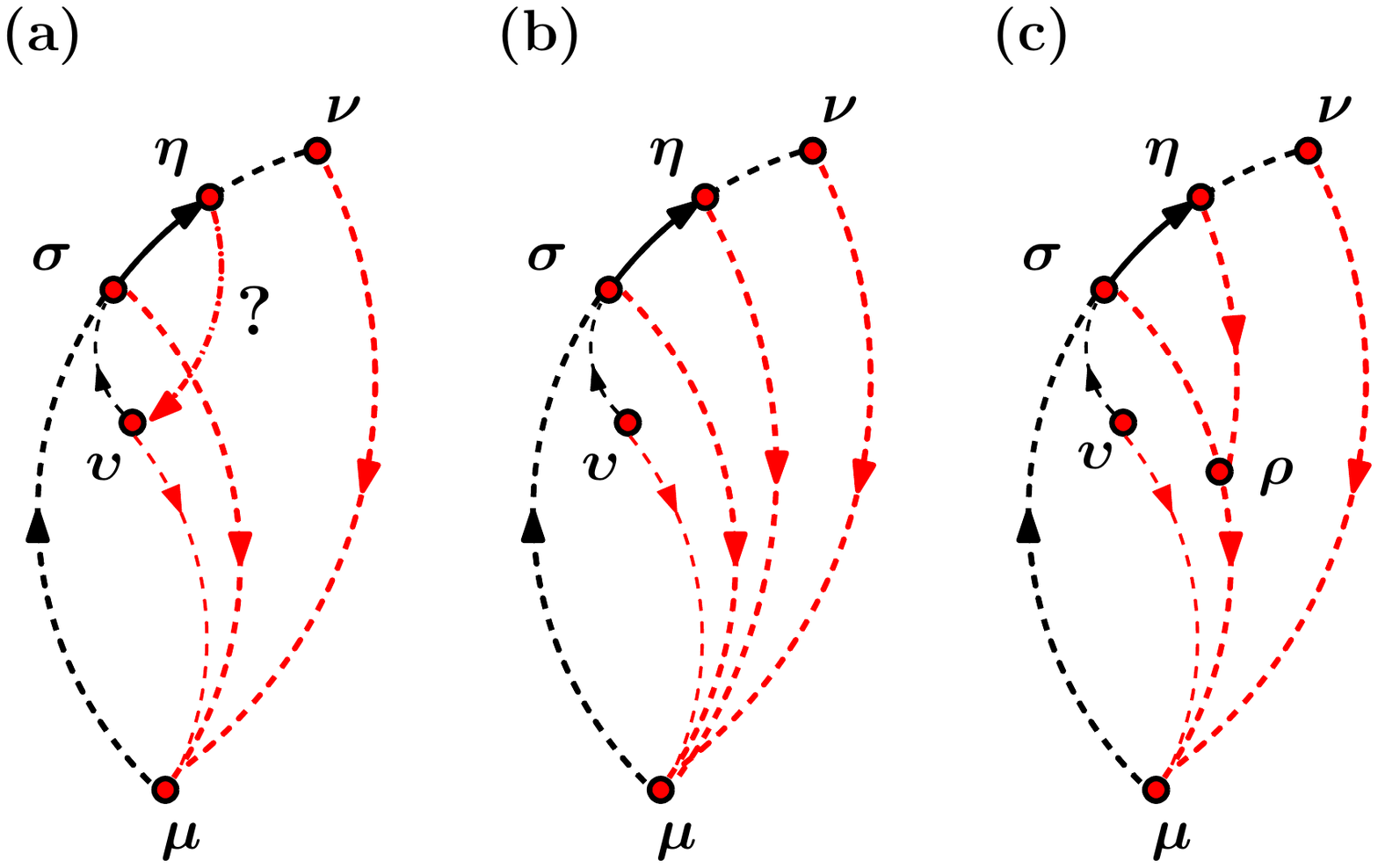} 
  \end{center}
  \caption{Nesting of two adjacent loops $(\vec{\mu},\vec{\sigma})$ and $(\vec{\mu},\vec{\eta})$, where $\vec{\eta} = U \vec{\sigma}$ and  
  $\vec{\upsilon}$ is an interior state of the $(\vec{\mu},\vec{\sigma})$ loop. The figure illustrates part (a) of Proposition \ref{prop:nestinglemma}
  which is about the orbit $D^*\vec{\eta}$. Panel (a): the orbit $D^*\vec{\eta}$ cannot reach the interior node $\vec{\upsilon}$, but instead 
  must intersect the $D$-boundary of the $(\vec{\mu},\vec{\sigma})$ loop. This can happen in two ways as shown in panels (b) and (c): either $D^*\vec{\eta}$ 
  intersects the $D$-boundary of $(\vec{\mu},\vec{\sigma})$ at its endpoint (b), or at an intermediate point $\vec{\rho}$, panel (c).
  } 
  \label{fig:Nesting}
\end{figure}

\begin{proposition}[Nesting of interior loops]\label{prop:nestinglemma}
 Let $(\vec{\mu},\vec{\nu})$ form a loop and let $\vec{\sigma}$ be an intermediate boundary state. Suppose
 \begin{itemize}
  \item [(a)] $\vec{\sigma} \in U^*\vec{\mu}$ and  let $\vec{\eta} = U\vec{\sigma}$. The following 
  holds for the orbit $D^*\vec{\eta}$:
\begin{itemize}
 \item [(i)] $\vec{\mu} \in D^*\vec{\eta}$,
 \item [(ii)] $D^*\vec{\eta}$ cannot contain any interior nodes of the loop $(\vec{\mu},\vec{\sigma})$, 
\end{itemize}  

\item[(b)] $\vec{\sigma} \in D^*\vec{\nu}$ and let $\vec{\eta} = D\vec{\sigma}$. The following holds for the orbit $U^*\vec{\eta}$:
\begin{itemize}
 \item [(i)] $\vec{\nu} \in U^*\vec{\eta}$,
 \item [(ii)] $U^*\vec{\eta}$ cannot contain any interior nodes of the loop $(\vec{\sigma}, \vec{\nu})$, 
\end{itemize}

 \end{itemize}

\end{proposition}

Prop.~\ref{prop:nestinglemma} indicates how a (sub) loop $(\vec{\mu},\vec{\sigma})$ can be augmented to a loop $(\vec{\mu},U\vec{\sigma})$ 
and what happens to its boundaries under this transformation. 
Note that by the parts (i) of the proposition, there is at least one state of $D^*\vec{\eta}$ ($U^*\vec{\eta}$) that is also on 
the $D$-($U$)-boundary of the loop $(\vec{\mu},\vec{\sigma})$, respectively $(\vec{\sigma},\vec{\nu})$. 
We work out next some useful consequences of Prop.~\ref{prop:nestinglemma}.

 Consider a $(\vec{\mu},\vec{\nu})$ -- loop and let $\vec{\sigma}$ be a point on its boundary. Recall that if $\vec{\sigma}$ is not an endpoint, 
 we call it an intermediate boundary state. Label the intermediate boundary points of the $U$-boundary as $\vec{\nu}_1, \vec{\nu}_2, \ldots \vec{\nu}_{n-1}$, 
 \begin{equation}
  \vec{\nu}_i = U^i \vec{\mu}, \quad U^n\vec{\mu} = \vec{\nu}.
  \label{eqn:nuidef}
 \end{equation}
Likewise, denote the $D$-boundary intermediate points by  $\vec{\mu}_1, \vec{\mu}_2, \ldots \vec{\mu}_{m-1}$,
 \begin{equation}
  \vec{\mu}_i = D^{m-i} \vec{\nu}, \quad D^m\vec{\nu} = \vec{\mu}.
  \label{eqn:muidef}
 \end{equation}
 Note that\footnote{
 If we want to emphasize that a boundary node 
 is an intersection point, we will use the letter $\vec{\kappa}$.}
 \begin{align}
  \vec{\mu} \prec_U \vec{\nu}_1 \prec_U &\ldots \prec_U \vec{\nu}_{n-1} \prec_U \vec{\nu}, \nonumber \\
  \vec{\mu} \prec_D \vec{\mu}_1 \prec_D &\ldots \prec_D \vec{\mu}_{m-1} \prec_D \vec{\nu}. \label{eqn:numuiorder}
 \end{align}
For the rest of this paper, we will use the above labeling and indexing convention for states on the $U$ and $D$-boundaries of a loop.

Given a $(\vec{\mu},\vec{\nu})$ -- loop, the nesting property, Prop. \ref{prop:nestinglemma}, allows us to define two mappings from the 
states of one of its boundaries to those of the other. By the $\ell$RPM property, this mapping has certain monotonicity 
properties. 

\begin{definition}[Boundary-to-boundary maps $\psi_\pm$]
\label{def:psidef}
Let states $\vec{\nu}_i$ and $\vec{\mu}_j$ be the boundary states of the loop $(\vec{\mu}, \vec{\nu})$. We 
define the following two mappings from one boundary of the loop into the other: 
 \begin{align*}
  \psi_-(\vec{\nu}_i) &= {\max}_D \left \{  D^*\vec{\nu}_i \cap D^*\vec{\nu} \right \}, \\
  \psi_+(\vec{\mu}_j) &= {\min}_U \left \{ U^*\vec{\mu}_j \cap U^*\vec{\mu} :  \right \}.
\end{align*}
 
 \end{definition}

 In words, the maps $\psi_\pm$ return the points where the orbit originating from one boundary is incident on the opposite boundary for the first time. In particular,  
points $\vec{\kappa}$ for which $\psi_\pm(\vec{\kappa}) = \vec{\kappa}$ are intersection points. This is established in Lemma \ref{lem:iminter} below.

\begin{proposition}[Monotonicity of the maps $\psi_\pm$]
\label{prop:bondarymon}
Let $(\vec{\mu},\vec{\nu})$ be the endpoints of a loop.  The mappings $\psi_\pm$ have the following properties:  
\begin{itemize}
 \item [(i)] $\psi_\pm (\vec{\mu}) = \vec{\mu}$ and $\psi_\pm (\vec{\nu}) = \vec{\nu}$.
 \item [(ii)] $\psi_\pm$ are non-decreasing: 
 \begin{align*}
  \vec{\nu}_i \prec_U \vec{\nu}_j &\Rightarrow \psi_-(\vec{\nu}_i) \preceq_D \psi_-(\vec{\nu}_j), \\
  \vec{\mu}_i \prec_D \vec{\mu}_j &\Rightarrow \psi_+(\vec{\mu}_i) \preceq_U \psi_+(\vec{\mu}_j).
 \end{align*}
\end{itemize}
\end{proposition}

Regarding the intersection points, the following result is useful.

\begin{lemma}[Images of intersection points]\label{lem:iminter}
 Let $\vec{\kappa}$ be an intermediate boundary point of a loop $(\vec{\mu},\vec{\nu})$. The following are equivalent:
 \begin{itemize}
  \item [(i)] $\vec{\kappa}$ is an intersection point.
  \item [(ii)] $\psi_-(\vec{\kappa}) = \vec{\kappa}$.
  \item [(iii)] $\psi_+(\vec{\kappa}) = \vec{\kappa}$. 
 \end{itemize}
\end{lemma}

Next, consider a point $\vec{\mu}_i$ on the $D$-boundary of a loop and consider the corresponding 
inverse maps $\psi_-^{-1}(\vec{\mu}_i)$, the {\em set} of points on 
the $U$-boundary that are mapped under $\psi_-$ into $\vec{\mu}_i$. Given a set of adjacent boundary points, 
$\vec{\nu}_i,\vec{\nu}_{i+1}, \ldots \vec{\nu}_{i+j}$, we will denote this as $[\vec{\nu}_i, \vec{\nu}_{i+j}]$ and call it an interval. 
We have the following Lemma:

\begin{lemma}[Pre-images $\psi^{-1}_\pm$]
\label{lem:invImages}
 Let $\vec{\mu}_i$ and $\vec{\nu}_j$ be points on the $U$, respectively $D$-boundary of a loop $(\vec{\mu}, \vec{\nu})$. Then, 
 the following hold:
 \begin{itemize}
  \item[(i)] the 
 pre-images $\psi^{-1}_-(\vec{\mu}_i)$ and $\psi^{-1}_+(\vec{\nu}_j)$ are intervals on the $U$-, respectively $D$-boundary of the loop. 
 That is, 
 \begin{align}
 \psi^{-1}_-(\vec{\mu}_i) &= [\vec{\nu}_{k'}, \vec{\nu}_k],  \label{eqn:downTriangleLoop} \\
 \psi^{-1}_+(\vec{\nu}_j) &= [\vec{\mu}_{l'}, \vec{\mu}_l],  \label{eqn:downTriangleLoop_2} 
 \end{align}
  for some pair of integers $k' < k$ and $l' < l$.
 
 \item[(ii)] The pairs $(\vec{\mu}_i, \vec{\nu}_k)$ and $(\vec{\mu}_{l'}, \vec{\nu}_j)$ are sub loops of $(\vec{\mu}, \vec{\nu})$.  
 \end{itemize}

\end{lemma}

The following two Lemmas describe an ``adjacency'' property of loops whose endpoints are located on the boundary of another loop. This is 
illustrated in Fig.~\ref{fig:StackingProof}

\begin{figure}[t!]
  \begin{center}
    \includegraphics[width = 3.2in]{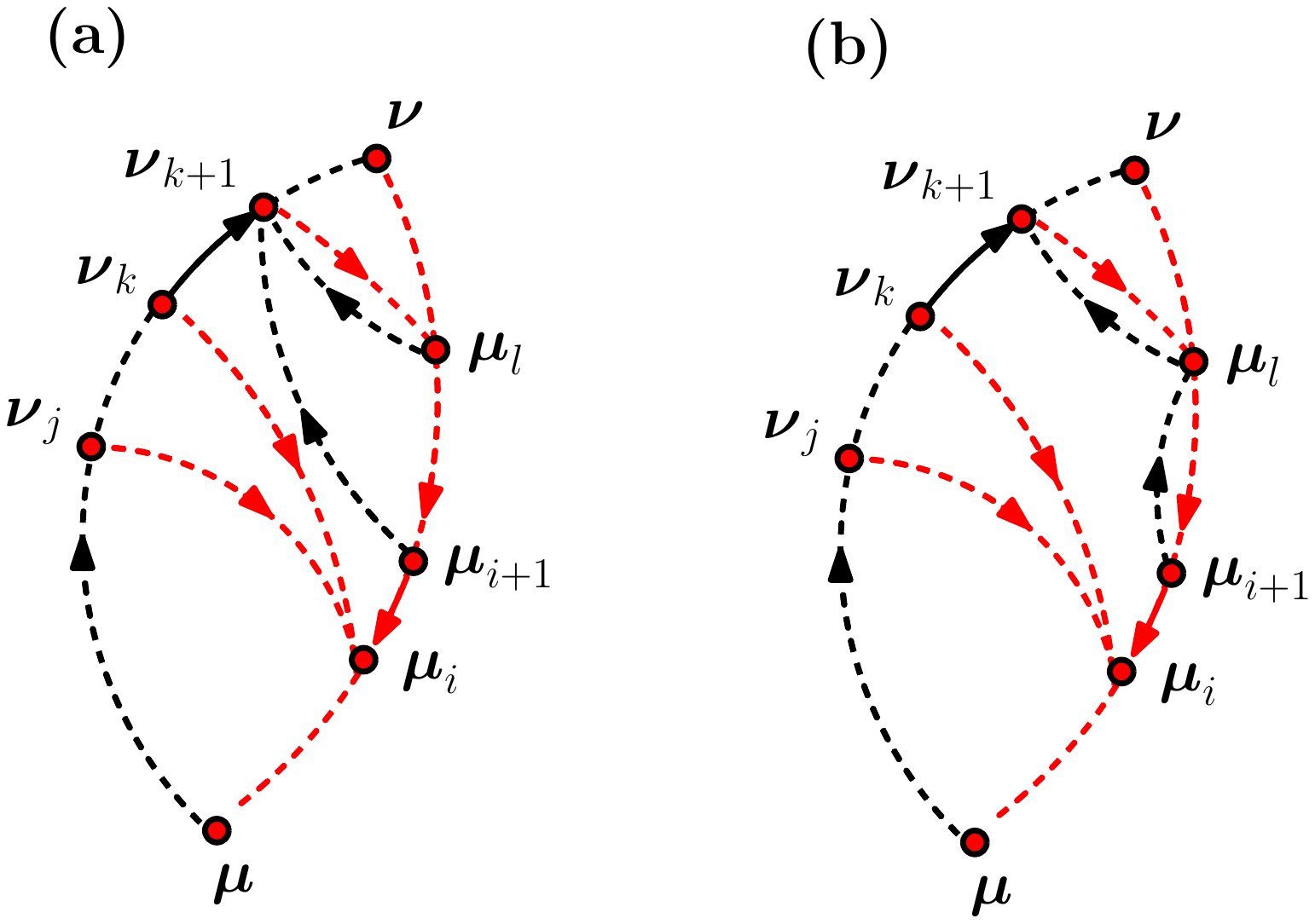} 
  \end{center}
  \caption{Illustration of the Adjacency Lemmas \ref{lemma:stacking} and \ref{lem:stacking}. In both panels it is assumed that $\psi^{-1}_-(\vec{\mu}_i) = [\vec{\nu}_j, \vec{\nu}_k]$, with 
  $\vec{\mu}_i$ and $\vec{\nu}_k$ being intermediate boundary points on the $D$- and $U$-boundary of the loop $(\vec{\mu}, \vec{\nu})$. The state $\vec{\mu}_l = \psi_-(\vec{\nu}_{k+1})$, 
  is the state where the boundary point  
  $\vec{\nu}_{k+1} = U\vec{\nu}_k$ is incident on the $D$-boundary of the loop $(\vec{\mu}, \vec{\nu})$. 
  Lemma \ref{lemma:stacking} 
  asserts that $\psi_+(\vec{\mu}_{i+1}) = \vec{\nu}_{k+1}$, where  $D\vec{\mu}_{i+1} = \vec{\mu}_i$. 
  Panels (a) and (b) 
  show two possible courses for the orbit $U^*\vec{\mu}_{i+1}$. In (a) apart from $\vec{\mu}_{i+1}$, the orbit $U^*\vec{\mu}_{i+1}$ 
  does not contain any other $D$-boundary state. In (b) $U^*\vec{\mu}_{i+1}$ contains 
  at least one other state on the $D$-boundary:  $\vec{\mu}_l$. Lemma \ref{lem:stacking} 
  asserts that $\psi_-(\vec{\nu}_{k}) = \vec{\mu}_{i}$.
  } 
  \label{fig:StackingProof}
\end{figure}

\begin{lemma}[Adjacency of sub loops I]\label{lemma:stacking}
 Let $(\vec{\mu},\vec{\nu})$ be the end points of a loop, let 
 $\vec{\mu}_i \prec_D \vec{\nu}$ be an intermediate point on its $D$-boundary such that the interval 
 \[
 \psi^{-1}_-(\vec{\mu}_i) = [\vec{\nu}_j, \vec{\nu}_k]  
 \]
is non-empty. Denote by $\vec{\mu}_{i+1}$ the boundary point adjacent to $\vec{\mu}_i$ so that $D\vec{\mu}_{i+1} = \vec{\mu}_i$. 
Then   
\[
  \psi_+(\vec{\mu}_{i+1}) = U \vec{\nu}_k = \vec{\nu}_{k+1}
 \]
and $\vec{\mu}_{i+1}$ is the lower endpoint of the sub loop $(\vec{\mu}_{i+1}, \vec{\nu}_{k+1})$. 

Moreover, the loops $(\vec{\mu}, \vec{\nu}_k)$ and $(\vec{\mu}_{i+1}, \vec{\nu}_{k+1})$ are disjoint and so are therefore the  
sets $\mathcal{R}_{\vec{\mu}, \vec{\nu}_k}$ and $\mathcal{R}_{\vec{\mu}_{i+1}, \vec{\nu}_{k+1}}$.

\end{lemma}

Note that by Prop.~\ref{prop:bondarymon} (i) $\psi_-(\vec{\nu}) = \vec{\nu}$. Thus, if $\vec{\mu}_i \prec_D \vec{\nu}$ and the set
 $\psi^{-1}_-(\vec{\mu}_i) = [\vec{\nu}_j, \vec{\nu}_k]$ is non-empty, this implies that $\vec{\nu}_k \prec_U \vec{\nu} $ and 
 hence equality is impossible, $\vec{\nu}_k \neq \vec{\nu}$.

Observe that the loop $(\vec{\mu}_{i+1}, \vec{\nu}_{k+1})$ in Lemma \ref{lemma:stacking}  is 
``adjacent'' to $(\vec{\mu}_i,\vec{\nu}_k)$, since its lower and upper endpoints each are adjacent on the 
corresponding boundaries of the loop $(\vec{\mu}, \vec{\nu})$.  
Lemma \ref{lemma:stacking} essentially tells us that a loop  $(\vec{\mu},\vec{\nu})$ can be broken down or ``stacked'' into a sequence of adjacent sub loops that are 
all disjoint: each $(\vec{\mu},\vec{\nu})$-reachable state $\vec{\sigma}$ belongs to one and only one of these sub loops. 
Also observe that when $l > i + 1$, so that there are states on the $D$-boundary between $\vec{\mu}_i$ and $\vec{\mu}_l$, the orbit $U^*\vec{\mu}_{r}$ can contain an 
 interval of states $\vec{\mu}_{r+1}, \vec{\mu}_{r+2}, \ldots, \vec{\mu}_s$ for some $r < s \leq l$, so that for the $(\vec{\mu}_{i+1},\vec{\nu}_{k+1})$-loop 
 these states are intersection points. One such case is shown in panel (b) of Fig.~\ref{fig:StackingProof}. 
 
 Upon exchange of the boundaries as well as the maps $\psi_\pm$, a result analogous to Lemma \ref{lemma:stacking} is found, {\em cf.} Fig.~\ref{fig:StackingProof}.  

\begin{lemma}[Adjacency of sub loops II]\label{lem:stacking}
 Let $(\vec{\mu},\vec{\nu})$ be the end points of a loop, let 
 $\vec{\nu}_{k+1} \prec_U \vec{\nu}$ be an intermediate point on its $U$-boundary  such that the interval 
 \[
 \psi^{-1}_+(\vec{\nu}_{k+1}) = [\vec{\mu}_{i+1}, \vec{\mu}_l]  
 \]
is non-empty. Denote by $\vec{\nu}_k$ the boundary point adjacent to $\vec{\nu}_{k+1}$, so that $U\vec{\nu}_k = \vec{\nu}_{k+1}$. 
Then   
\[
  \psi_-(\vec{\nu}_k) = \vec{\mu}_i,
 \]
where $\vec{\mu}_i = D \vec{\mu}_{i+1}$, 
and $\vec{\mu}_i$ is the lower endpoint of sub loop $(\vec{\mu}_i, \vec{\nu}_k)$. 

Moreover, the loops $(\vec{\mu}_{i+1}, \vec{\nu})$ and $(\vec{\mu}_i, \vec{\nu}_k)$ are disjoint. In other words, the 
sets $\mathcal{R}_{\vec{\mu}_{i+1}, \vec{\nu}}$ and $\mathcal{R}_{\vec{\mu}_i, \vec{\nu}_k}$ are disjoint.
 
\end{lemma}

We can now state our main result on the partition of a loop into sub loops.

\begin{theorem}[Standard partitioning of a loop into sub loops]\label{prop:mert}
  Let $(\vec{\mu},\vec{\nu})$ be the endpoints of a loop.
  The 
  $(\vec{\mu},\vec{\nu})$ loop can be decomposed into two, possibly three (adjacent) sub loops, $\ell_-,\ell_+$, and  
  $\ell_0$ by  the following {\em standard partition procedure}:
  
  Let the 
  inverse images of $\vec{\mu}$ and $\vec{\nu}$ under $\psi_{\pm}$ be 
  \begin{align}
   \psi_{-}^{-1}(\vec{\mu}) &= [\vec{\mu}, \vec{\nu}_i], \label{eqn:psimim}\\
   \psi_{+}^{-1}(\vec{\nu}) &= [\vec{\mu}_j, \vec{\nu}], \label{eqn:psiplm}
  \end{align}
  for some $i$ and $j$, and where some of the intervals might contain just one element. 
  Denote by $\ell_-$ and $\ell_+$ the sub loops $(\vec{\mu}, \vec{\nu}_i)$ and $(\vec{\mu}_j,\vec{\nu})$, respectively. 
  Consider the boundary state $\vec{\nu}_{i+1} = U\vec{\nu}_i$. 
  If $\vec{\nu}_{i+1} \prec_U \vec{\nu}$, 
  then there is a third sub loop $\ell_0$ with endpoints $(\vec{\mu_1}, \vec{\nu}_{n-1})$, where 
   $D\vec{\mu_1} = \vec{\mu}$ and $U\vec{\nu}_{n-1} = \vec{\nu}$.
   
  The sub loops $\ell_\pm$, and if present $\ell_0$, form a (disjoint) partition of the set of $(\vec{\mu},\vec{\nu})$-reachable 
  states $\mathcal{R}_{\vec{\nu}, \vec{\mu}}$.   Moreover, the process of repeatedly applying the standard partition procedure 
  to the resulting sub loops of $(\vec{\mu},\vec{\nu})$ has a natural representation as an ordered tree.
\end{theorem}

Fig.~\ref{fig:CanonicalPart}(a) and (b) depict the two partitioning possibilities described in 
Thm.~\ref{prop:mert}.

\begin{figure}[t!]
  \begin{center}
    \includegraphics[width = 3in]{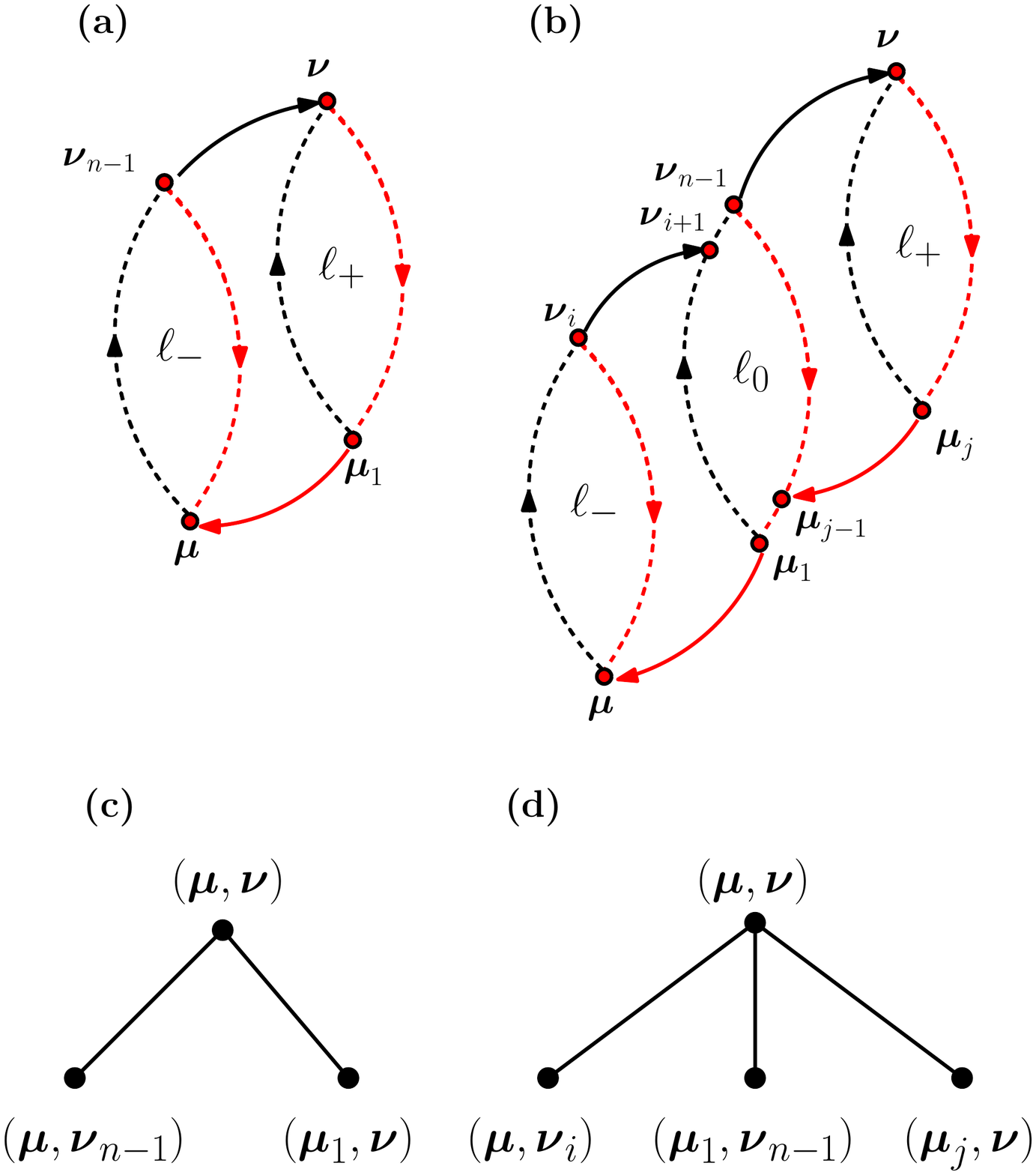} 
  \end{center}
  \caption{The two possibilities for the standard partitioning of a loop $(\vec{\mu}, \vec{\nu})$ into two (a) or three (b) 
  sub loops. The sub loops are marked as $\ell_-, \ell_+$ and $\ell_0$. Panels (c) and (d) depict the tree representation of 
  the partitioning of the parent loop $(\vec{\mu}, \vec{\nu})$ into two or three off-spring sub loops. Note that the ordering 
  of the sub-loops inside the main parent loop in (a) and (b) matters, thus the representations of the decomposition in (c) and (d) are 
  therefore ordered trees.
  } 
  \label{fig:CanonicalPart}
\end{figure}

 One or more of the sub loops in the partition may consist of a single point, a {\em singleton loop}. This can occur for example  
when one of the boundaries does not contain any intermediate states. In this case the partition produces two loops and one of these 
is a singleton. 

 If the boundary contains intersection points and we perform a standard partition, these intersection points will always end up in the middle loop 
 $\ell_0$, as is easily seen: by Lemma  \ref{lem:iminter} we have  $\psi^{-1}_\pm(\vec{\kappa}) = \{ \vec{\kappa}\}$, so if $\vec{\kappa}$ is an intersection point 
 of the loop $(\vec{\mu}, \vec{\nu})$ then $\vec{\kappa}$ is neither in $\psi^{-1}_-(\vec{\mu})$ nor in $\psi^{-1}_+(\vec{\nu})$. Hence it 
 must be in the middle loop. We thus have the following Corollary:
 
\begin{corollary}[Standard partition and intersection points] 
 Let $(\vec{\mu}, \vec{\nu})$ be the endpoints of a loop $\ell$. If the boundary of $\ell$ contains an intersection point then 
 the standard partitioning  must yield three loops. 
\end{corollary}

The important feature of Thm.~\ref{prop:mert} is that the partitioning can be done unambiguously and 
hence is well-defined. 
Each of these sub loops can in turn be  
partitioned {\em etc.} In this way a loop can be broken down into singleton states/loops\footnote{We allow decomposing a
$2$-loop or a $3$-loop into two, respectively three, singleton loops.}.  
The relation between the parent loop and its off-spring loops resulting from the partitioning is shown 
in Fig.~\ref{fig:CanonicalPart} (c) and (d). Note that the off-springs nodes are ordered: the left and right nodes 
correspond to the loops $\ell_-$ and $\ell_+$, respectively.
 \begin{figure}[t!]
  \begin{center}
    \includegraphics[width = 0.6\textwidth]{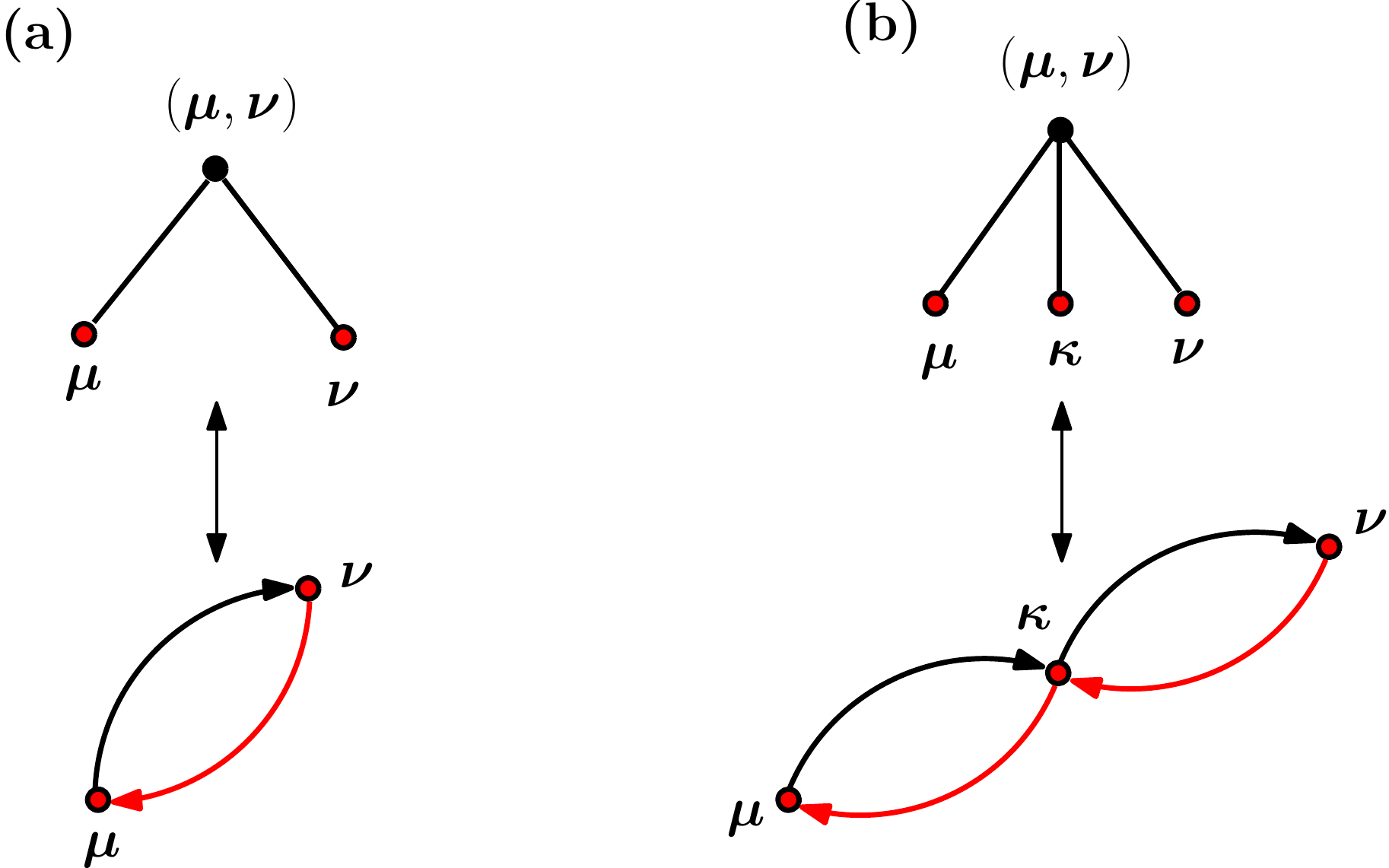} 
  \end{center}
  \caption{(a) and (b) the two possibilities for the leaf-nodes and their parent loop node $(\vec{\mu}, \vec{\nu})$ in the tree representation of the standard loop partitioning. (a) The 
  loop node having two leaf-offsprings corresponds to the  2-loop, while the other possibility of three leaf-offspring, shown in (b), corresponds to the  $3$-loops, 
  containing in addition the intersection point $\vec{\kappa}$.
  } 
  \label{fig:LoopComp}
\end{figure}
The leaves of the decomposition tree are the nodes that do not have any offsprings. These are the singleton loops. Since the standard partitioning is also 
a partition of the set of states $\mathcal{R}_{\vec{\mu},\vec{\nu}}$ associated with the parent loop $(\vec{\mu},\vec{\nu})$, 
the set of leave nodes is identical to $\mathcal{R}_{\vec{\mu}, \vec{\nu}}$. 
 The leaves of the 
 tree have either a $2$- or $3$-loop as their parent, as shown in Fig~\ref{fig:LoopComp} (a) and (b). In terms of the loop order $<$ of 
 Def.~\ref{def:dynPO}, these leaf states obey the order relations $\vec{\mu} < \vec{\nu}$ and  $\vec{\mu} < \vec{\kappa} < \vec{\nu}$, 
 respectively. 

 The decomposition tree of the reachability graph of Fig.~\ref{fig:reachabilitygraph} is 
given in Fig.~\ref{fig:ReachTree}. The intermediate loops are shown as small black circles with the end points of the corresponding loops 
indicated next to them. The numbering of states corresponds to that of Fig.~\ref{fig:reachabilitygraph}. The leaf nodes are the states of the 
graph and are shown in blue. Note that on the decomposition tree a path from a  singleton loop $\vec{\sigma}$ 
to the root is a sequence of expanding loops, one contained inside the next. This picture is useful when we want to understand 
the response of an initial state to periodic forcing, Section \ref{sec:PF}.

\begin{figure}[t!]
  \begin{center}
    \includegraphics[width = 0.8\textwidth]{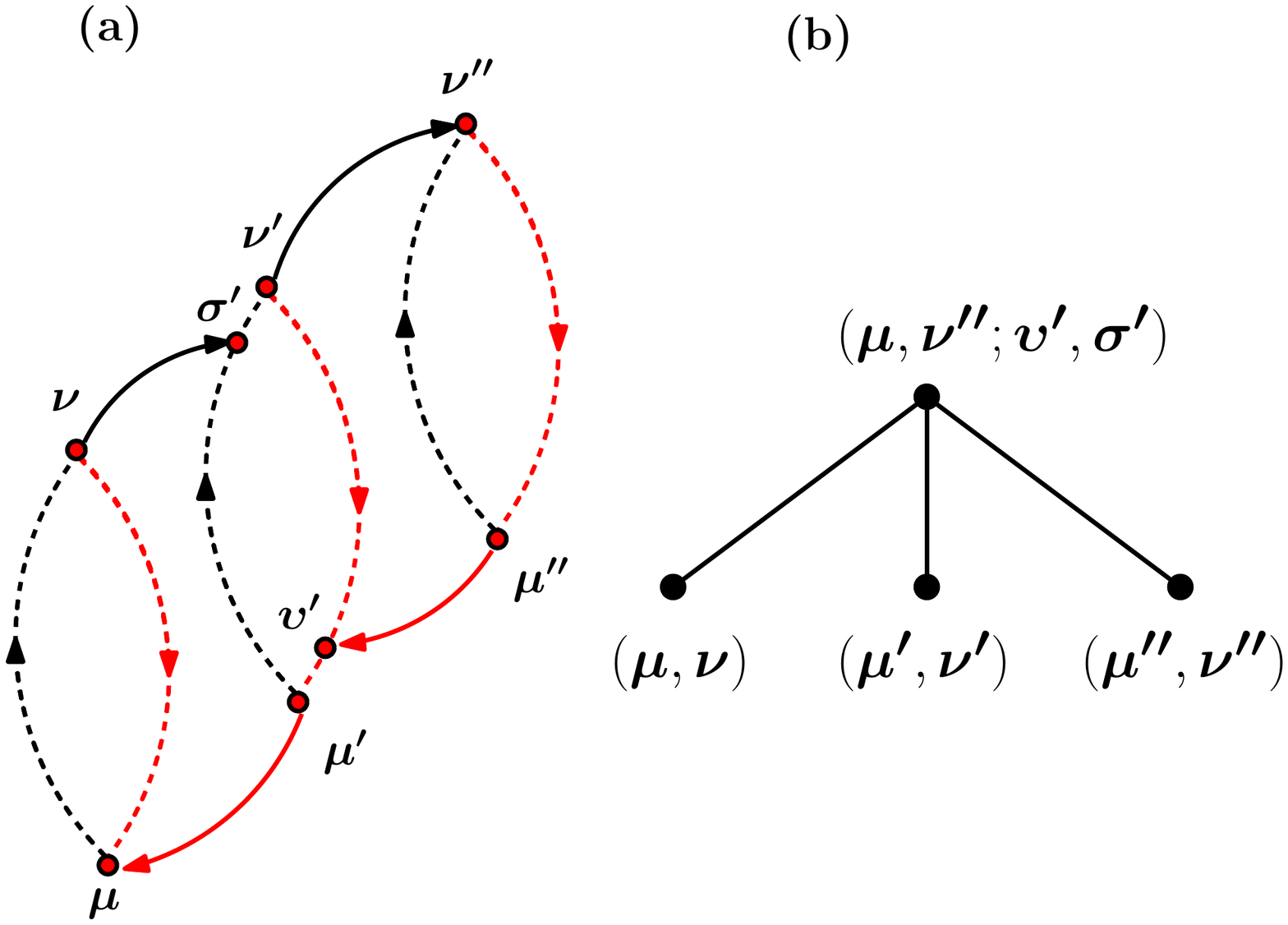} 
  \end{center}
  \caption{Glue-points in the case the standard partition results in three offspring loops. The loop $(\vec{\mu},\vec{\nu''})$ 
  is partitioned into the three sub loops $(\vec{\mu}, \vec{\nu})$, $(\vec{\mu'}, \vec{\nu'})$, and $(\vec{\mu''}, \vec{\nu''})$, 
  panel (a). Panel (b) shows the tree representation of this partition. The reverse procedure of reassembling the parent 
  loop from the three offspring loops is not unique. We need to specify also where the upper end point of  $(\vec{\mu}, \vec{\nu})$
  and the lower endpoint of $(\vec{\mu''}, \vec{\nu''})$ are connected to the corresponding boundaries of the middle 
  loop $(\vec{\mu'}, \vec{\nu'})$. These are the boundary points $\vec{\sigma'}$ and $\vec{\upsilon'}$. In this example  
  $\vec{\sigma'}$ and $\vec{\upsilon'}$ are intermediate boundary points of the loop $(\vec{\mu'}, \vec{\nu'})$ and we 
  will refer to them as glue points. Note that the glue points are also boundary nodes of the parent loop $(\vec{\mu},\vec{\nu''})$.
  In the tree decomposition we indicate the glue points in the label of the parent loop, as shown in (b).
  } 
  \label{fig:CanonicalDecomp}
\end{figure}

As is already apparent from Fig.~\ref{fig:CanonicalPart}, if the state transition graphs associated with the off-spring loops are planar, 
so is the parent loop. In fact, we have 

\begin{theorem}[Planarity of state transition graph associated with a loop] 
\label{thm:planar}
Let the pair $(\vec{\mu}, \vec{\nu})$ form a loop. The state transition graph associated with the set of nodes 
$\mathcal{R}_{\vec{\mu}, \vec{\nu}}$ is planar. 
\end{theorem}

Note that the planarity of the state transition graph is actually a direct consequence of the Adjacency Lemmas \ref{lemma:stacking} and \ref{lem:stacking}. 
This can be seen from Fig.~\ref{fig:StackingProof}, illustrating the adjacency property. The pair of loops $(\vec{\mu}, \vec{\nu}_k)$ 
and $(\vec{\mu}_{i+1},\vec{\nu})$ furnishes a partition of the loop $(\vec{\mu}, \vec{\nu})$ into two components. A reasoning similar to that of the proof 
of Thm.~\ref{thm:planar} then implies planarity of $(\vec{\mu}, \vec{\nu})$, 

We conclude this section with remarks on the reverse procedure of assembling a loop from its sub loops by following the decomposition tree and 
the possibility of inferring the maps $U$ and $D$ from it. 
The reverse procedure is ambiguous when the partition had given rise to three offspring loops.
This can be seen in Fig.~\ref{fig:CanonicalPart} (b). The 
choice of the ``glue points'', the states $\vec{\nu}_{i+1}$ and $\vec{\mu}_{j-1}$ on the $U$- and $D$-boundaries of the middle 
loop $\ell_0$, is not unique. In order to be able to reconstruct the original loop from its sub loops, we 
therefore also have to specify the points where the corresponding 
endpoints $\vec{\nu}_i$ and $\vec{\mu}_j$ of $\ell_-$ and $\ell_+$ are  ``glued'' to the boundary of the  middle loop. This 
ambiguity does not arise when the decomposition yields two loops, Fig.~\ref{fig:CanonicalPart} (a). In this case 
endpoints are glued to endpoints. We will only explicitly indicate the glue-points when at least on of them is not the endpoint of the parent loop. The labeling convention for glue points is depicted in Fig.~\ref{fig:CanonicalDecomp} (a) and (b). 

\begin{figure}[t!]
  \begin{center}
    \includegraphics[width = 3.5in]{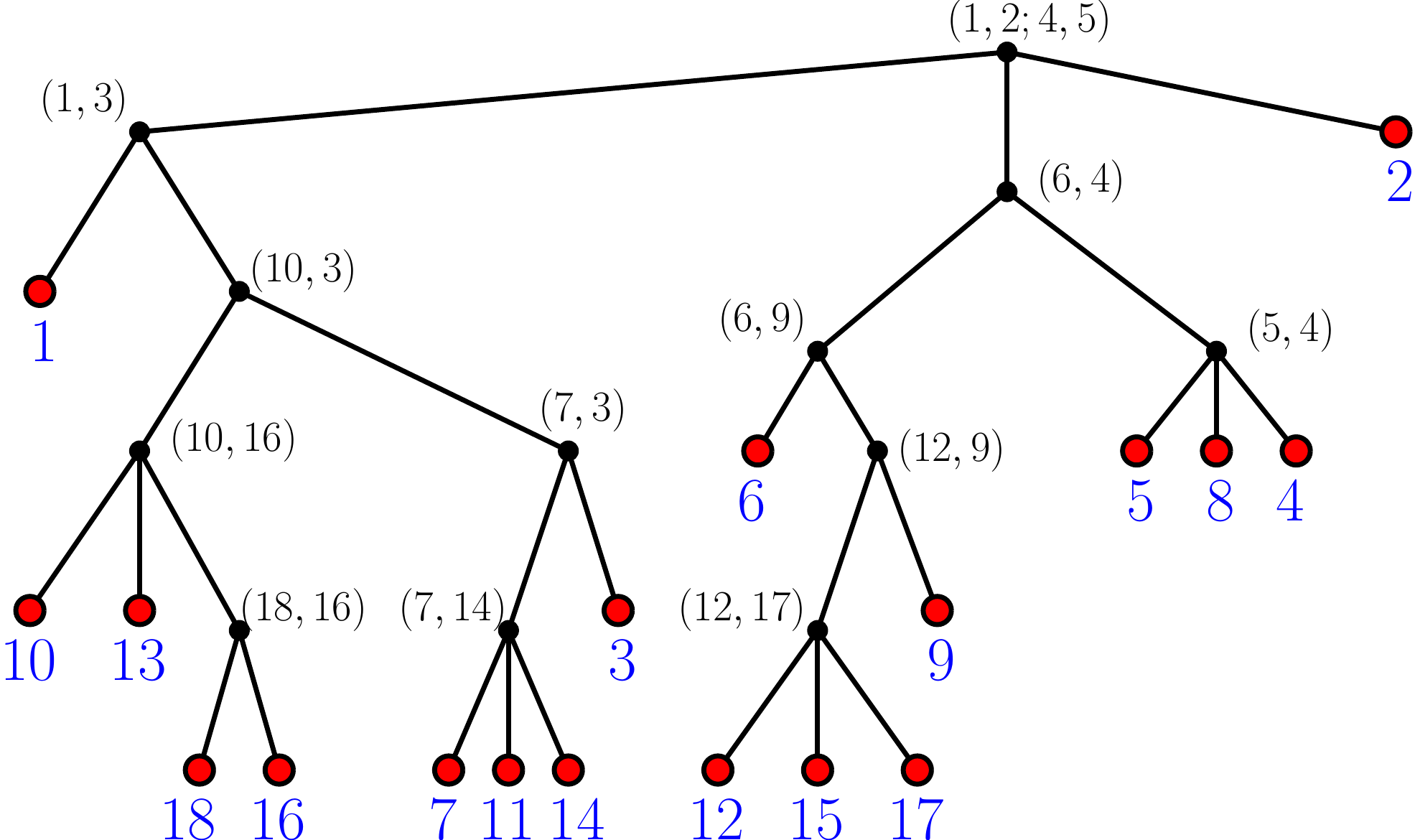} 
  \end{center}
  \caption{Tree representation of the standard partitioning of the state transition graph of Fig.~\ref{fig:reachabilitygraph}. 
  The numbers next to the nodes correspond to the labeling of states in Fig.~\ref{fig:reachabilitygraph}.
  The root of the tree is the major loop $(1,2)$ of the transition graph.  
  The standard partitioning, Thm.~\ref{prop:mert}, breaks down a loop into two or three sub loops. 
  In the tree representation, these  are the off-spring sub loops, shown as small black circles, 
  and the endpoints of these loops are indicated next to the nodes of the tree.
  The leaves of the tree (blue circles) correspond to the singleton loops and 
  at the same time they are the states of the transition graph. 
  Leaf nodes that are glue points are indicated in the labeling of the parent loop where they are used, 
  {\em e.g.} the loop $(1,2;4,5)$ in the figure. We adopt the convention to specify only glue-points that are not end-points. 
  } 
  \label{fig:ReachTree}
\end{figure}

\begin{example}[Non-trivial glue points] 
The state-transition graph of Fig.~\ref{fig:reachabilitygraph} contains one non-trivial glue point, meaning 
that all other connections between offspring loops forming a parent loop are endpoint to endpoint. The standard partitioning decomposes the 
$(1,2)$-loop into the  
three sub loops $(1,3)$, $(6,4)$ and $(2,2)$.
The state $5$ is a glue point at which the $(2,2)$ loop is connected to the $D$-boundary of the $(6,4)$ loop. Alternatively, the 
lower endpoint of the $(2,2)$ loop could also have been connected to the points, $4$, $8$, or $6$, {\em i.e.} anywhere else on the $D$-boundary. 
In the reachability graph of Fig.~\ref{fig:reachabilitygraph} the endpoint of the loop $(1,3)$, is connected to $4$, the upper endpoint 
of the $(6,4)$ loop. Fig.~\ref{fig:ReachTree} depicts the tree representation of the standard partition of this loop. Since the partition of the 
loop $(1,2)$ involves the non-trivial glue point $5$, this is indicated in the label next to that loop as $(1,2;4,5)$.  
$\maltese$
\end{example}

One can think of the standard partition procedure as a systematic way of removing $U$ and $D$ transitions among pairs of states so that the 
resulting state transition graph consists of   
two or three disconnected components. In graph theory the edges whose removal gives rise to disconnected components are called cut-edges. Continuing in 
this manner, once all the transitions have been removed, we are left with singleton loops, the states of the state transition graph. To illustrate this, 
consider the two possibilities for the standard partitioning of a loop, represented in Fig.~\ref{fig:CanonicalPart} (a) and (b) and 
with their corresponding tree representations in (c) and (d), respectively. 
For the partitioning resulting in two sub loops,  Fig.~\ref{fig:CanonicalPart} (a), the partitioning in effect removes the two transitions
\[
U\vec{\nu}_{n-1} = \vec{\nu} \quad \mbox{and} \quad D\vec{\mu}_1 = \vec{\mu}. 
\]
Likewise, for a partitioning resulting in three loops, as shown in Fig.~\ref{fig:CanonicalDecomp},  the four transitions 
\[
U\vec{\nu} = \vec{\sigma'}, \quad D\vec{\mu'} = \vec{\mu}, \quad U\vec{\nu'} = \vec{\nu''}, \quad \mbox{and} \quad D\vec{\mu''} = \vec{\upsilon'} 
\]
are removed. As explained before, out of the latter 
four, the first and fourth require additionally the specification of glue points, $\vec{\sigma'}$ and $\vec{\upsilon'}$. All other transitions are 
between endpoints of the constituent loops. This indicates that the (decorated) tree representation of the standard partitioning also encodes all the 
transitions between states under $U$ and $D$. This encoding can be inferred easily from comparing the corresponding panels in each of the columns  
of Fig.~\ref{fig:CanonicalPart}. 
The reachability graph associated with a loop $(\vec{\mu}, \vec{\nu})$ and its associated decorated decomposition tree 
are two complementary ways of representing the reachable states $\mathcal{R}_{\vec{\mu}, \vec{\nu}}$ and the transitions between them under 
$U$ and $D$.

\section{Maximal loops and inter-loop structure}
\label{sec:maxloopetc}

In the previous two sections our focus has been on what happens inside a loop.  We have shown that a loop $(\vec{\mu}, \vec{\nu})$ can be decomposed 
into a hierarchical structure of sub loops and that this hierarchy has a natural representation as an ordered tree. In this section we are interested in 
what happens ``outside'' of loops at the {\em inter-loop} level. The natural thing to ask first is whether the 
maps $U$ and $D$ are such that  $(\vec{\mu}, \vec{\nu})$ is a sub loop  of another loop. 

Suppose that $(\vec{\mu}, \vec{\nu})$ forms a loop so that $\mathcal{R}_{\vec{\mu}, \vec{\nu}}$ is the set of  $(\vec{\mu}, \vec{\nu})$-reachable states, 
as defined in Def. 
\ref{def:munureachable}. We can ask 
for the largest loop $(\vec{\mu'}, \vec{\nu'})$-loop such that $\mathcal{R}_{\vec{\mu}, \vec{\nu}} \subset \mathcal{R}_{\vec{\mu'}, \vec{\nu'}}$. 
We will call such a loop a {\em maximal loop}. Below is an algorithm to obtain the maximum loop containing a given loop $(\vec{\mu}, \vec{\nu})$.

\begin{algorithm}\label{alg:maxloop}

 Let the current loop be $(\vec{\tilde{\mu}}, \vec{\tilde{\nu}}) = (\vec{\mu}, \vec{\nu})$. 
 \begin{itemize}
  \item [(A1)] Determine the largest  state (with respect to $\prec_U$) on the orbit $U^*\vec{\tilde{\nu}}$ whose $D$-orbit returns to the $D$-boundary 
  of the current loop:  
     \[
    \vec{\nu'} = {\max}_U \{ \vec{\sigma} \in U^*\vec{\tilde{\nu}} : (\vec{\tilde{\mu}}, \vec{\sigma}) \mbox{ forms a loop}\}.
   \]
   If $\vec{\tilde{\nu}} \prec_U \vec{\nu'}$, make the update  $ \vec{\nu'} \to \vec{\tilde{\nu}}$.
   \item [(A2)] Determine the smallest state $\vec{\mu'}$ (with respect to $\prec_D$) on the orbit $D^*\vec{\tilde{\mu}}$ whose $U$-orbit 
   returns to the $U$-boundary of the current loop:
   \[
    \vec{\mu'} = {\min}_D \{ \vec{\upsilon} \in D^*\vec{\tilde{\mu}} : (\vec{\upsilon}, \vec{\tilde{\nu}}) \mbox{ forms a loop}\}.
   \]
   If $\vec{\mu'} \prec_D \vec{\tilde{\mu}} $, make the update  $ \vec{\mu'} \to \vec{\tilde{\mu}}$.
   \item[(A3)] Repeat steps (A1) and (A2) 
   until both steps yield no further update of the endpoints. 
   The loop $(\vec{\tilde{\mu}}, \vec{\tilde{\nu}})$ obtained in this 
   way is the maximal loop containing  the initial loop $(\vec{\mu}, \vec{\nu})$. 
   
 \end{itemize}

\end{algorithm}

Note that given a loop $(\vec{\mu}, \vec{\nu})$, except possibly for the endpoints, for all other states in 
$\mathcal{R}_{\vec{\mu}, \vec{\nu}}$ the 
transitions under $U$ and $D$ lead again to states in $\mathcal{R}_{\vec{\mu}, \vec{\nu}}$. 
Moreover, by Prop.~\ref{prop:munuorbitLemma} all these states leave the loop via its endpoints, unless the endpoint is an absorbing state. 
Thus a loop can communicate with the states of a larger loop only along states on the orbits $U^*\vec{\nu}$ or $D^*\vec{\mu}$. Step (A1) searches for 
the $U$-largest state $\vec{\nu'}$ on the orbit $U^*\vec{\tilde{\nu}}$ whose orbit $D^*\vec{\nu'}$ returns to $(\vec{\tilde{\mu}}, \vec{\tilde{\nu}})$. 
If this happens, the nesting property Prop.~\ref{prop:nestinglemma}, as illustrated in Fig.~\ref{fig:Nesting} (b) and (c), asserts that 
$D^*\vec{\nu'}$ must intersect the $D$-boundary of the loop $(\vec{\tilde{\mu}}, \vec{\tilde{\nu}})$ and therefore contain its lower endpoint 
$\vec{\tilde{\mu}}$. Hence $(\vec{\tilde{\mu}}, \vec{\tilde{\nu'}})$
forms a loop and moreover, $\mathcal{R}_{\vec{\tilde{\mu}}, \vec{\tilde{\nu}}} \subset \mathcal{R}_{\vec{\tilde{\mu}}, \vec{\tilde{\nu'}}}$. 
Step (A2) does the same by attempting to extend the lower endpoint of the loop.

The reasoning leading to the standard partition procedure of Theorem \ref{prop:mert} implies that the growth of a given loop via steps (A1) or (A2) of the 
algorithm can be actually viewed as merger of two loops. This is illustrated in Fig.~\ref{fig:LoopCompI}(a). Starting with the loop $(\vec{\mu}, \vec{\nu})$, 
suppose step (A1) identifies the state $\vec{\nu'}$ as the $U$-largest state whose $D$-orbit contains $\vec{\mu}$. Thus there must exists a $D$-lowest state in 
$D^*\vec{\nu'}$ that is not yet part of the $D$-boundary of $(\vec{\mu}, \vec{\nu})$ but whose successor under $D$ is. In the figure this is the state $\vec{\mu'}$.  
By the $\ell$RPM property applied to 
the loop $(\vec{\mu}, \vec{\nu'})$, the pair $(\vec{\mu'}, \vec{\nu'})$ must also form a loop and step (A1) in effect merges the loops $(\vec{\mu}, \vec{\nu})$ and 
$(\vec{\mu'}, \vec{\nu'})$ to form the loop $(\vec{\mu}, \vec{\nu'})$. 

\begin{figure}[t!]
  \begin{center}
    \includegraphics[width = 3.2in]{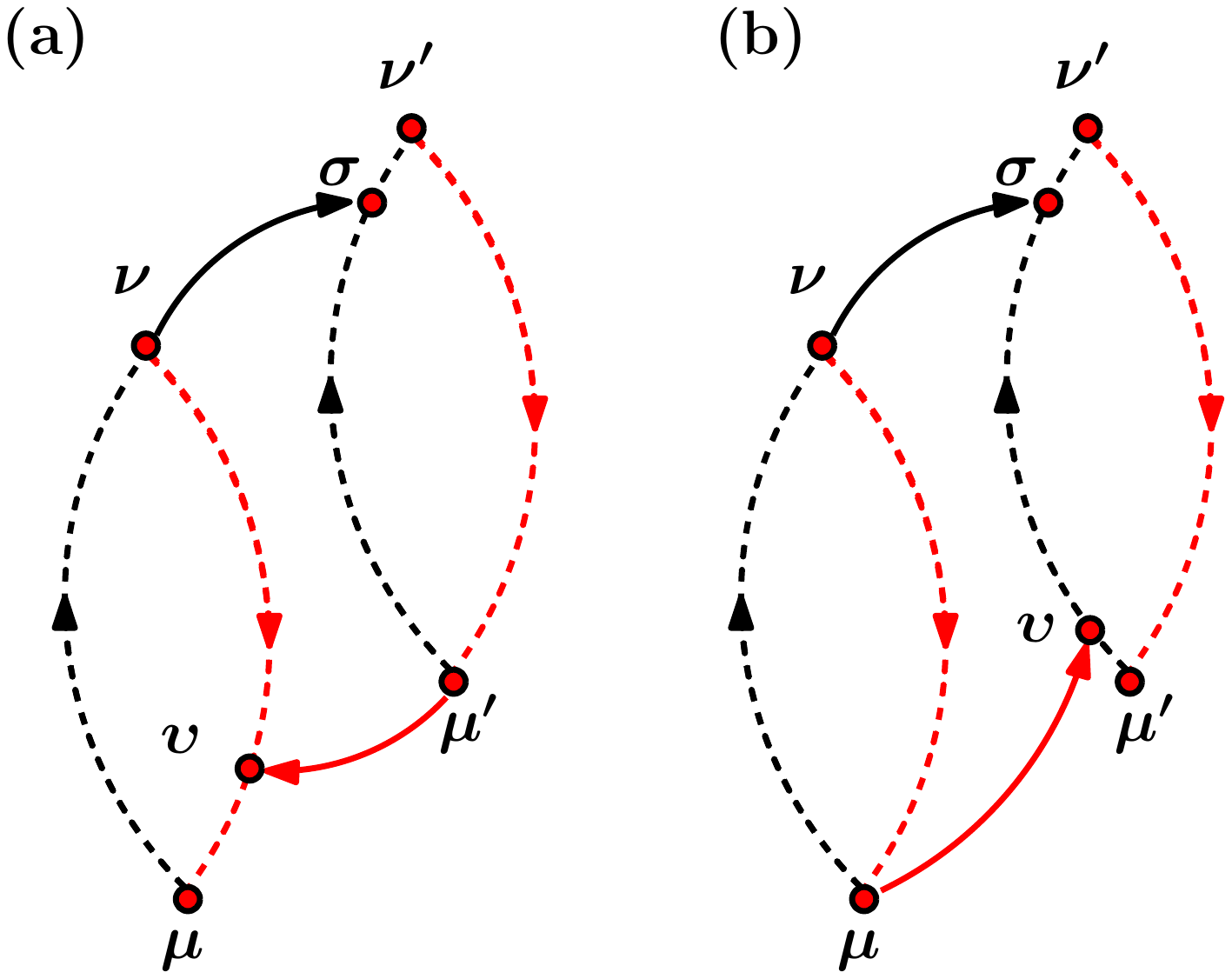} 
  \end{center}
  \caption{Possible ways of combining two loops .
  The boundary 
  nodes $\vec{\sigma}$ and $\vec{\upsilon}$ are glue points. (a) By gluing the endpoints to the opposite boundaries of the partner loop, we 
  obtain the augmented loop $(\vec{\mu}, \vec{\nu'})$. All states in this loop are mutually reachable and 
  $\mathcal{R}_{\vec{\mu}, \vec{\nu}} \subset \mathcal{R}_{\vec{\mu}, \vec{\nu'}}$. (b) By connecting both endpoints of one 
  loop to one of the boundaries of the partner loop, say the $U$-boundary of $(\vec{\mu'},\vec{\nu'})$, as depicted in the figure, mutual 
  reachability is lost. There are field histories that lead from the states in $\mathcal{R}_{\vec{\mu}\vec{\nu}}$ to those in $\mathcal{R}_{\vec{\mu'},\vec{\nu'}}$, 
  but the reverse is not true. States in $\mathcal{R}_{\vec{\mu},\vec{\nu}}$ are not reachable from $\mathcal{R}_{\vec{\mu'},\vec{\nu'}}$.
  } 
  \label{fig:LoopCompI}
\end{figure}

An example for the combination of two loops $(\vec{\mu}, \vec{\nu})$ and $(\vec{\mu'}, \vec{\nu'})$ where mutual reachability is not retained and consequently a maximal 
loop cannot be formed, is given in Fig.~\ref{fig:LoopCompI}(b). 
The states of $\mathcal{R}_{\vec{\mu}, \vec{\nu}}$ are not $(\vec{\mu'}, \vec{\nu'})$-reachable, since there is no field-history starting from $\vec{\mu'}$ or $\vec{\nu'}$, 
leading to a state in $\mathcal{R}_{\vec{\mu}, \vec{\nu}}$. Neither has the connection of the two loops resulted in a larger loop containing both as sub loops\footnote{ 
The state $\vec{\mu}$ cannot lie on the orbit $D^*\vec{\mu'}$, since we have $ \vec{\upsilon} \prec_D \vec{\mu}$, and from the $\ell$RPM property applied to 
the loop $(\vec{\mu'}, \vec{\nu'})$ we also have $\vec{\mu'} \prec_D \vec{\upsilon}$, so that $\vec{\mu'} \prec_D \vec{\mu}$.}. 
These considerations lead to the following definition of maximal loops: 

\begin{definition}[Maximal loops]\label{def:maxLoop}
 A loop $(\vec{\mu}, \vec{\nu})$ is said to be a maximal loop if the following two conditions hold:
 \begin{itemize}
  \item[(a)] $\forall \vec{\nu'} \in \mathcal{S}$, $\vec{\nu} \prec_U \vec{\nu'} \, \Rightarrow \, D^*\vec{\nu'} \cap \mathcal{R}_{\vec{\mu}, \vec{\nu}} = \emptyset$  
  \item[(b)] $\forall \vec{\mu'} \in \mathcal{S}$, $\vec{\mu'} \prec_D \vec{\mu} \, \Rightarrow \, U^*\vec{\mu'} \cap \mathcal{R}_{\vec{\mu}, \vec{\nu}} = \emptyset$  
 \end{itemize}
\end{definition}

In words, once a maximal loop has been left as a result of a monotonous field change, a subsequent field reversal will 
not bring the system back to a state of the loop. The following proposition asserts that Algorithm \ref{alg:maxloop} indeed terminates in a maximal loop.

\begin{proposition}[Termination of Algorithm  \ref{alg:maxloop} in a maximal loop]
\label{prop:algTermination}
 Let $(\vec{\mu'}, \vec{\nu'})$ be a maximal loop so that Def.~\ref{def:maxLoop} holds and consider a sub loop $(\vec{\mu}, \vec{\nu})$ 
 with $\vec{\mu}, \vec{\nu} \in \mathcal{R}_{\vec{\mu'}, \vec{\nu'}}$.
 Then Algorithm \ref{alg:maxloop} initialized with $(\vec{\mu}, \vec{\nu})$ will 
 reach 
 $(\vec{\mu'}, \vec{\nu'})$ and terminate. Moreover, Algorithm \ref{alg:maxloop} will lead to the same maximal loop $(\vec{\mu'}, \vec{\nu'})$ irrespective of the order of steps $(A1)$ and $(A2)$.
\end{proposition}

In terms of the analysis of the state transition graph associated with an  $\ell$AQS-A, a first step therefore is   
to identify all maximal loops $(\vec{\mu}, \vec{\nu})$ and their associated states $\mathcal{R}_{\vec{\mu}, \vec{\nu}}$. It is convenient 
to treat states $\vec{\sigma}$ that are not part of any loop as singleton maximal loops $(\vec{\sigma},\vec{\sigma})$. This is consistent 
with Definition \ref{def:maxLoop}. Denoting the set of maximal loops of an $\ell$AQS-A by $\mathcal{L}$, it then follows that the corresponding 
sets of reachable states form a (disjoint) partition of $\mathcal{S}$. The inter-loop transition among maximal loops can be 
represented in terms of a functional graph as follows. 

\begin{definition}[Inter-loop state transition graph of an $\ell$AQS-A]\label{def:interLoop}
 Given an $\ell$AQS-A $(\mathcal{S}, F^\pm, U, D)$, let 
 \[
  \mathcal{L} = \{ (\vec{\mu}, \vec{\nu}) \in \mathcal{S} \times \mathcal{S}: (\vec{\mu}, \vec{\nu}) \, \mbox{is a maximal loop} \}
 \]
be the set of its maximal loops. The inter-loop transition graph is the graph whose vertex set is $\mathcal{L}$ and whose set of directed 
edges represent the functional graphs of the inter-loop transitions under $U$ and $D$ as follows. 
For $(\vec{\mu}, \vec{\nu}), (\vec{\mu'}, \vec{\nu'}), (\vec{\mu''}, \vec{\nu''}) \in \mathcal{L}$,  
\begin{align}
 (\vec{\mu'}, \vec{\nu'}) &= U(\vec{\mu}, \vec{\nu}), \quad \Leftrightarrow \quad U\vec{\nu} \in \mathcal{R}_{\vec{\mu'}, \vec{\nu'}}, \label{eqn:interLoopU}\\
 (\vec{\mu''}, \vec{\nu''}) &= D(\vec{\mu}, \vec{\nu}), \quad \Leftrightarrow \quad D\vec{\mu} \in \mathcal{R}_{\vec{\mu''}, \vec{\nu''}}. \label{eqn:interLoopD}
\end{align}
The transitions \eqref{eqn:interLoopU} and \eqref{eqn:interLoopD} will be called $U$- and $D$-inter-loop transitions, respectively. 
\end{definition}

Note that the existence of the maximal loop $\mathcal{R}$ of $(\vec{\alpha}, \vec{\omega})$-reachable states is guaranteed by the absorbing properties 
of its endpoints $(\vec{\alpha}, \vec{\omega})$.  
The $\ell$RPM property does not imply the existence of any other non-reachable states. Even if the $\ell$AQS-A has non-reachable states these do not necessarily have to be 
part of non-singleton maximal loops. Therefore, if $(\vec{\alpha}, \vec{\omega})$ is the only non-singleton 
maximal loop, 
any non-trivial hysteretic properties will be due to the states of $\mathcal{R}$.  To illustrate this, we consider two examples in which 
$(\vec{\alpha}, \vec{\omega})$ is the only non-singleton maximal loop.

\begin{example}[$\ell$AQS-A and maximal loops I]
\label{ex:shortTrans} 
Suppose that apart form the reachable states $\mathcal{R}$ we have in addition $n$ non-reachable states $\vec{\sigma}_i$, with 
$i = 1, 2, \ldots, n$. 
We assign transitions of these states under $U$ and $D$ as follows. For $i = 1, 2, \ldots, n$,
\begin{align*}
 U\vec{\sigma}_i &\in  \mathcal{R}, \\
 D\vec{\sigma}_i &\in \mathcal{R}, 
\end{align*}
It is readily checked that the states $( \vec{\sigma}_i )_{1 \leq i \leq n}$ are indeed not $(\vec{\alpha}, \vec{\omega})$-reachable and that they do not form loops. Hence 
they are singleton maximal loops. 
Moreover, each of these states transits under one $U$ or $D$ step into $\mathcal{R}$ and 
remains there. The blue colored vertices on the right half of Fig.~\ref{fig:NonReach} illustrate this case.  

The corresponding inter-loop transition 
graph has the vertex set $\mathcal{L} = \{ (\vec{\alpha},\vec{\omega}) \} \cup_i \{ (\vec{\sigma}_i, \vec{\sigma}_i) \}$ and its directed edges are readily worked out 
from Definition \ref{def:interLoop} as
\begin{align*}
U(\vec{\alpha},\vec{\omega}) &= (\vec{\alpha},\vec{\omega}), \quad D(\vec{\alpha},\vec{\omega}) = (\vec{\alpha},\vec{\omega}),\\
U(\vec{\sigma}_i,\vec{\sigma}_i) &=  (\vec{\alpha},\vec{\omega}),\quad 
D(\vec{\sigma}_i,\vec{\sigma}_i) =  (\vec{\alpha},\vec{\omega}), \quad i = 1, 2, \ldots, n.
\end{align*}
$\maltese$
\end{example}

In the above example there is a ``transient'' of at most one step before the states are trapped in $\mathcal{R}$. 
The following is an example of an $\ell$AQS-A where this ``transient'' can be made arbitrarily long.

\begin{example}[$\ell$AQS-A and maximal loops II]
\label{ex:longTrans}
 Suppose that apart from the set of reachable states $\mathcal{R}$ and their transitions we have additionally $n^2$  singleton maximal loops whose   
 states we index by a pair of integer as $(\vec{\upsilon}_{i,j})_{1 \leq i,j \leq n}$. Let the transitions of these states be as follows:
 \begin{align*}
 U \vec{\upsilon}_{i,j} &= \left \{ 
 \begin{array}{cc}
  \vec{\upsilon}_{i + 1,j}, & 1 \leq i \leq n -1, \\
  \vec{\omega}, & i = n,
 \end{array}
 \right. 
 \end{align*}
 and 
 \begin{align*}
 D \vec{\upsilon}_{i,j} 
 &= \left \{ 
 \begin{array}{cc}
  \vec{\upsilon}_{i,j+1}, & 1 \leq j \leq n -1, \\
  \vec{\alpha}, & j = n.
 \end{array}
 \right. 
\end{align*}

It is readily checked that with increasing $i$  the states are increasing with respect to the partial order $\prec_U$, while with increasing $j$ they 
are decreasing with respect to $\prec_D$. The 
$U$ and $D$ orbits send the states  $\vec{\upsilon}_{i,j}$ into the corresponding absorbing states. All non-reachable states are 
singleton maximal loops and hence trivially satisfy the $\ell$RPM property, 
together with the set of 
reachable states $\mathcal{R}$ this system is an $\ell$AQS-A.
By construction, each non-reachable state gets 
trapped in $\mathcal{R}$ after at most $2n$ transitions. 
The case with $n = 5$ is illustrated in the left half of Fig.~\ref{fig:NonReach}. $\maltese$
\end{example}

\begin{figure}[t!]
  \begin{center}
    \includegraphics[width = 0.8\textwidth]{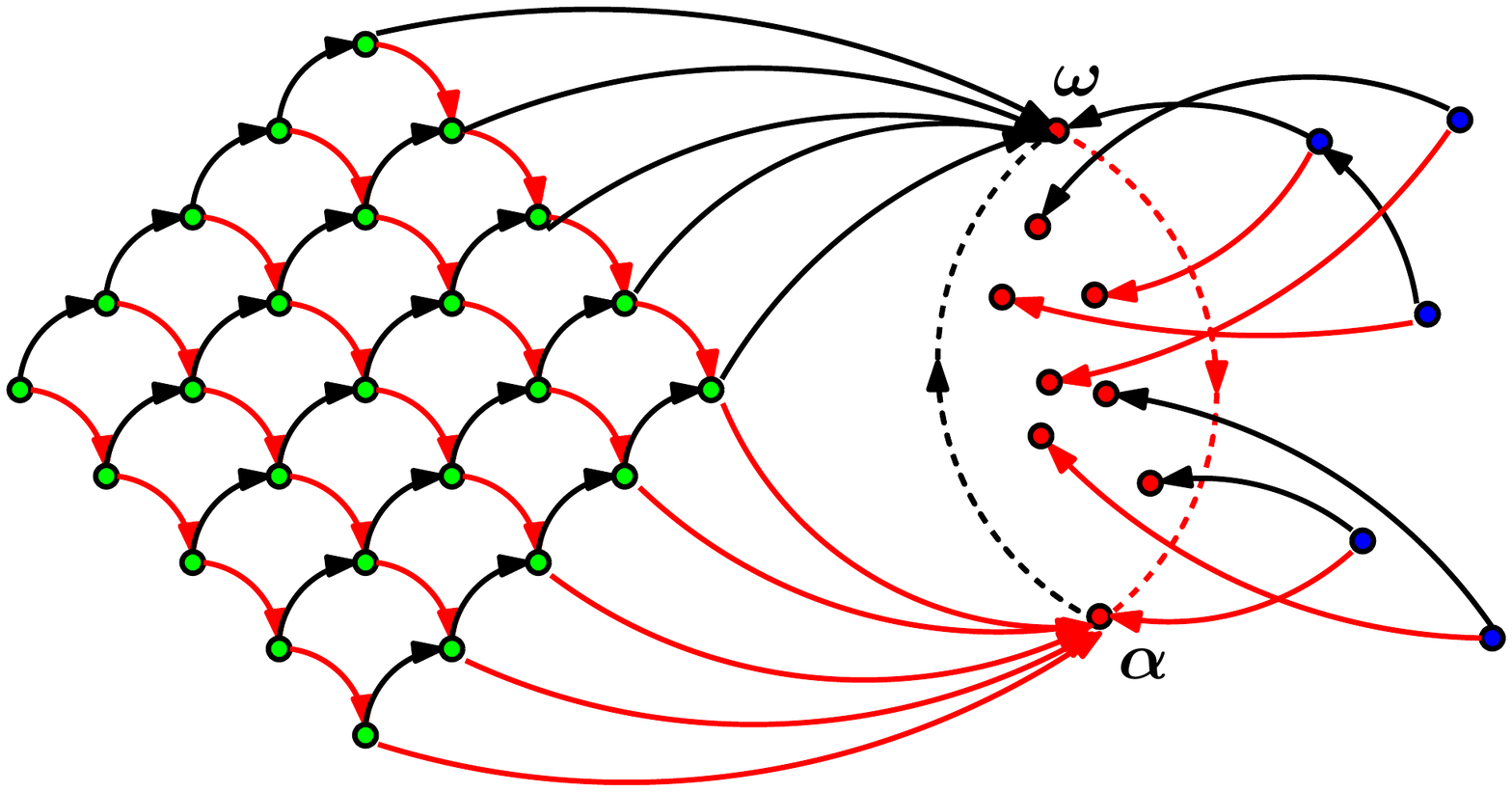} 
  \end{center}
  \caption{The set of reachable states (red) are 
  associated with the maximal loop that is formed by  
  the pair of absorbing states $\vec{\alpha}$ and $\vec{\omega}$. The non-reachable states (in blue and green) are not part of any non-singleton loops. These states 
  can transit directly into 
  reachable states, as illustrated by the states colored blue and placed to the right of the major loop. Alternatively, they  
   can also be connected  to each other in a way avoiding the formation of loops. An example for such a connection pattern is given by the states in green   
   left of the loop $(\vec{\alpha},\vec{\omega})$. 
  } 
  \label{fig:NonReach}
\end{figure}

These examples demonstrate that the $\ell$RPM property does not severely restrict the connectivity pattern at the inter-loop level, {\em i.e.} the 
structure of the inter-loop transition graph. This is to be compared with the connection pattern at the intra-loop level, 
which is strongly constrained by the $\ell$RPM property. In the next section we turn to a description of $\ell$AQS-A subject to 
time-periodic forcing and show that  
it is the features of the inter-loop transition graph that determine the length of transients.
Note that the choice of the non-reachable states as singleton maximal loops in the two examples was only for reasons of simplicity. 
The singleton maximal loops could have been replaced by arbitrary maximal loops and the $U$- and $D$-transitions would become inter-loop transitions among these.  

\pagebreak

\section{$\ell$AQS-A response to time-periodic forcing}
\label{sec:PF}

The $\ell$AQS-A response to monotonous changes of the driving force has already been defined via \eqref{eqn:FUaction} and \eqref{eqn:FDaction}. 
Let us consider now an $\ell$AQS-A subject to 
time-periodic forcing $(F_t)_{t \geq 0}$ with period $T$ and starting at $t = 0$ from some value $\hat{F}^-$, rising monotonously to  $\hat{F}^+$, then 
decreasing monotonously back to $\hat{F}^-$ at time $T$, as shown in Fig.~\ref{fig:LoopTransGraph}(a). We will call such a forcing 
{\em simple periodic forcing} with amplitudes $\hat{F}^\pm$.    
In particular, we are interested in the one-period map $\mathcal{F}: \mathcal{S} \to \mathcal{S}$, 
\begin{equation}
 \vec{\sigma}_T = \mathcal{F}[\hat{F}^-,\hat{F}^+] \vec{\sigma}_0 = \mathcal{D}[\hat{F}^-] \mathcal{U}[\hat{F}^+] \vec{\sigma}_0,
\end{equation}
and its orbit $\mathcal{F}^*[\hat{F}^-,\hat{F}^+] \vec{\sigma}_0 = ( \vec{\sigma}_0, \vec{\sigma}_T, \vec{\sigma}_{2T}, \ldots )$. 
In the following, in order to keep the main ideas simple, we will assume that the forcing is simple. Many of the results can be adapted to the case 
when the periodic forcing is arbitrary and bounded between $\hat{F}^\pm$. 

\begin{definition}[Trapping loop]
\label{def:trapping}
Let a simple periodic forcing $(F_t)_{t \geq 0}$ with amplitudes $\hat{F}^\pm$ be given. We call a loop $(\vec{\mu}, \vec{\nu})$ trapping
for the forcing $(F_t)_{t \geq 0}$, if 
\begin{equation}
  F^-(\vec{\mu}) < \hat{F}^- <  \hat{F}^+ < F^+(\vec{\nu}). 
  \label{eqn:Fconst} 
\end{equation}

\end{definition}

The following is useful. 

\begin{lemma}[Confinement to a trapping loop]\label{lem:confinement}
Given a simple periodic forcing and a loop $(\vec{\mu}, \vec{\nu})$ that is trapping, Definition \ref{def:trapping},  the following hold:
\begin{itemize}
 \item [(a)] For all $\vec{\sigma} \in \mathcal{R}_{\vec{\mu}, \vec{\nu}}$,
 \[
  F^-(\vec{\mu}) \leq F^-(\vec{\sigma}) < F^+(\vec{\sigma}) \leq F^+(\vec{\nu}).
 \]
 \item [(b)] Let $\vec{\sigma}_0 \in \mathcal{R}_{\vec{\mu}, \vec{\nu}}$, then for all $t > 0$, 
 \[
  \vec{\sigma}_t \in \mathcal{R}_{\vec{\mu}, \vec{\nu}}.
 \]
\end{itemize}
\end{lemma}

Thus if the $\ell$AQS-A is subjected to simple forcing and 
it starts out in a state of a loop that is trapping, it will remain confined to the loop for all subsequent times.
It is useful to make the following definition:
\begin{definition}[Loop-traversing forcing]
\label{def:looptraversing}
Let a loop $(\vec{\mu}, \vec{\nu})$ and a 
simple $T$-periodic forcing $(F_t)_{t \geq 0} $ with amplitudes $\hat{F}^\pm$ be given such that the loop $(\vec{\mu}, \vec{\nu})$  
is traversed as a result of the forcing: 
\begin{align*}
 \mathcal{U}[\hat{F}^+] \vec{\mu} &= \vec{\nu}, \\
 \mathcal{D}[\hat{F}^-] \vec{\nu} &= \vec{\mu}. 
\end{align*}
We shall call a such a forcing $(\vec{\mu}, \vec{\nu})${\em -traversing}.
\end{definition}

We are interested in periodic response to simple forcing and in particular marginality \cite{Tangetal87}, defined 
as follows:

\begin{definition}[Loop marginality]
\label{def:loopmarginality}
The response of the loop to  simple forcing is {\em marginal}, if for any $(\vec{\mu}, \vec{\nu})${\em -traversing} simple forcing 
$\hat{F}^\pm$ and all amplitudes $F$ such that $\hat{F}^- \leq F \leq \hat{F}^+$, the following hold: 
\begin{align}
  \mathcal{D}[\hat{F}^-] \mathcal{U}[F] \vec{\mu} &= \vec{\mu}, \label{eqn:FRPM1}\\
  \mathcal{U}[\hat{F}^+] \mathcal{D}[F] \vec{\nu} &= \vec{\nu}. \label{eqn:FRPM2}
\end{align}
The  simple forcings with amplitudes $(\hat{F}^-,F)$ and $(F, \hat{F}^+)$ thus traverse  
sub loops of $(\vec{\mu}, \vec{\nu})$ containing $\vec{\mu}$, respectively $\vec{\nu}$, as one of their endpoints.   
\end{definition}

\begin{definition}[$\ell$AQS-A marginal stability]
\label{def:lAQSAmarginality}
 We say that an $\ell$AQS-A exhibits marginal stability, if the loop marginality property, Def.~\ref{def:loopmarginality} holds 
 for all loops $(\vec{\mu},\vec{\nu})$ and simple forcings traversing these. 
\end{definition}

While the $\ell$RPM condition, Def.~\ref{def:RPM} is clearly necessary for marginality it is not sufficient. Specifically, 
$\ell$RPM does neither imply \eqref{eqn:FRPM1} nor \eqref{eqn:FRPM2}. Consider the first condition  \eqref{eqn:FRPM1}. The 
state $\mathcal{U}[F]\vec{\mu}$ clearly is on the $U$-boundary of the loop $(\vec{\mu}, \vec{\nu})$ and by the $\ell$RPM property, 
$\vec{\mu} \in D^*(\mathcal{U}[F]\vec{\mu})$. However, the latter together with $\hat{F}^- \leq F \leq \hat{F}^+$ only implies that 
\[
 \vec{\mu} \preceq_D \mathcal{D}[\hat{F}^-] \mathcal{U}[F] \vec{\mu},
\]
but for marginal loop response we require equality. Next, if equality does not hold, by a similar argument we would only have that
\[
 \mathcal{U}[F] \mathcal{D}[\hat{F}^-] \mathcal{U}[F] \vec{\mu} \preceq_U \vec{\nu}. 
\]
Continuing in this way, we see that in the absence of marginality a transient of many driving periods can occur before a periodic response is achieved. 
While the limit cycle will be a sub loop of $(\vec{\mu}, \vec{\nu})$, it does not need to share any of its boundary points. The following is a necessary 
and sufficient condition for the marginal stability of a loop. 

\begin{lemma}[Loop marginality condition]
 \label{lem:loopmarginality}
 Let $(\vec{\mu}, \vec{\nu})$ be a loop and consider its $U$- and $D$-boundary points  $(\vec{\nu}_i)_{1 \leq i \leq n}$ and 
 $(\vec{\mu}_j)_{0 \leq i \leq m -1}$. Denote by $\vec{\mu}^{(i)} \in D^*\vec{\nu}_i$ and $\vec{\nu}^{(j)} \in U^*\vec{\mu}_j$ the 
 states preceding the corresponding endpoints, so that 
 \begin{align}
  \vec{\mu} &= D\vec{\mu}^{(i)}, \quad 1 \leq i \leq n, \label{eqn:miact}\\
  \vec{\nu} &= U\vec{\nu}^{(j)}, \quad 0 \leq j \leq m-1. \label{eqn:nujact}
 \end{align}
The loop $(\vec{\mu}, \vec{\nu})$ is marginal, if 
\begin{align}
 F^-[\vec{\mu}^{(i)}] &\geq F^-[\vec{\mu}^{(n)}], \quad 1 \leq i \leq n, \label{eqn:FMcond1}\\
 F^+[\vec{\nu}^{(j)}] &\leq F^+[\vec{\nu}^{(0)}], \quad 0 \leq j \leq m-1. \label{eqn:FMcond2}
\end{align}
\end{lemma}

We then have the following necessary and sufficient condition for the marginal response of an $\ell$AQS-A. 
\begin{proposition}[$\ell$AQS-A marginal stability]\label{prop:AQSmarginality}
 Given an $\ell$AQS-A, the following two are equivalent:
 \begin{itemize}
  \item[(a)] The $\ell$AQS-A exhibits marginal stability.
  \item[(b)] For any loop $(\vec{\mu}, \vec{\nu})$ and predecessor states, $(\vec{\nu}_i)_{1 \leq i \leq n}$ and 
 $(\vec{\mu}_j)_{0 \leq i \leq m -1}$, 
  the sequences $(F^-[\vec{\mu}^{(i)}])_{1 \leq i \leq n}$  
  and $(F^+[\vec{\nu}^{(j)}])_{0 \leq j \leq m-1}$, as defined in Lemma \ref{lem:loopmarginality}, are non-increasing.   
 \end{itemize}
\end{proposition}

The no-passing property, Def.~\ref{def:NP}, implies marginal stability. To see this, consider a loop $(\vec{\mu}, \vec{\nu})$ and a force $F^{(1)}_t$ 
monotonously decreasing from some $\hat{F}^+ $ to $\hat{F}^-$ such that it evolves $\vec{\nu}$ to $\vec{\mu}$. Apply now
to an intermediate boundary state $\vec{\nu}_i$ a force  $F^{(2)}_t$ that decreases monotonously from some value  $F < \hat{F}^+$ to $\hat{F}^-$ .  
Choosing the time dependence so that $F^{(2)}_t \le F^{(1)}_t$ for all $t$, and since $\vec{\nu}_i \prec_U \vec{\nu}$, the no-passing property implies that 
$F^{(2)}_t$ evolves $\vec{\nu}_i$ to $\vec{\mu}$.

\begin{theorem}[Transients in trapping loops under simple periodic driving]
\label{thm:trans}
 Consider an $\ell$AQS-A that possesses marginal stability and let $(F_t)_{t \geq 0} $ be a $T$-periodic simple forcing with 
 amplitudes $\hat{F}^\pm$. Suppose that there is a loop $(\vec{\mu}, \vec{\nu})$   
 trapping the forcing so that condition \eqref{eqn:Fconst} holds. Let $\vec{\sigma}_0$ be any $(\vec{\mu}, \vec{\nu})$-reachable 
 initial state and consider the orbit of the one-period map
 \[
  \mathcal{F}^*[\hat{F}^-,\hat{F}^+] \vec{\sigma}_0 = ( \vec{\sigma}_0, \vec{\sigma}_T, \vec{\sigma}_{2T}, \ldots ).
 \]
 Then, for $n \geq 1$, 
 \[
  \vec{\sigma}_{nT} = \vec{\sigma}_T,  
 \]
 meaning that the trajectory reaches a limit-cycle after a transient of at most one period. 
\end{theorem}
In particular, note that starting in a $(\vec{\alpha}, \vec{\omega})$-reachable state, a periodically forced $\ell$AQS-A  will reach 
a limit-cycle after at most one period, since by definition, the maximal loop $(\vec{\alpha}, \vec{\omega})$ is trapping for any forcing.

Apart from marginality, the restriction in Theorem \ref{thm:trans} to an initial state belonging to a loop $(\vec{\mu}, \vec{\nu})$ that traps 
the forcing is essential. Key to its proof is the observation that the standard partition of a 
loop via its tree representation, Thm.~\ref{prop:mert}, associates with every state $\vec{\sigma} \in \mathcal{R}_{\vec{\mu}, \vec{\nu}}$ a sequence of 
nested sub loops, each containing its predecessor and all containing $\vec{\sigma}$. In terms of the forcing this sequence of loops 
traps a monotonously increasing range of forces. The condition \eqref{eqn:Fconst} on $(\vec{\mu}, \vec{\nu})$ then  implies that 
there exists a sub loop in this sequence all of whose successors -- including $(\vec{\mu}, \vec{\nu})$ -- are trapping.

Conversely, when the forcing is such that the maximal loop containing 
the initial state is not trapping, the dynamics will be such that several maximal loops are entered and left as the forcing reaches its extreme values 
$\hat{F}^\pm$, until a trapping loop is found. Since the endpoints of the loop $(\vec{\alpha}, \vec{\omega})$ are absorbing and this loop traps 
any forcing, the finiteness of $\mathcal{S}$ ensures that a trapping loop will be found in a finite time. 
Moreover by Theorem \ref{thm:trans}, once this happens the system settles into a limit cycle after at most 
one further period of the driving. As already evident from Example \ref{ex:longTrans}, the length of the transient  
can be arbitrarily long. We now show how the structure of the inter-loop transition graph 
is related to the length of such transients.

Assume that a simple periodic forcing 
with amplitudes $\hat{F}^\pm$ is applied to a state of a maximal loop $(\vec{\mu}, \vec{\nu})$ which moreover cannot trap the forcing so that 
\eqref{eqn:Fconst} does not hold. 
The system trajectory is bound to leave this loop through one of its endpoints, entering another maximal loop. Thus the 
forcing induces inter-loop transitions between maximal loops and we can follow these transitions as a path on the inter-loop transition graph, 
as given by Def.~\ref{def:interLoop}. 
Let us suppose the maximal $(\vec{\mu}, \vec{\nu})$ is left through its upper endpoint $\vec{\nu}$, so that we have a sequence of  $U$-inter-loop transitions, 
as defined in \eqref{eqn:interLoopU}. These transitions continue until we reach the first maximal loop $(\vec{\mu'}, \vec{\nu'})$, 
satisfying $\hat{F}^+ < F^+(\vec{\nu'})$.
Denote by $\vec{\sigma}$ the state of the system at the time the forcing has reached its maximum value $\hat{F}^+$. The state $\vec{\sigma}$ 
necessarily belongs 
to $(\vec{\mu'}, \vec{\nu'})$. Applying Lemma~\ref{lem:confinement}(a), it must be that 
\begin{equation}
 F^+(\vec{\nu}) < \hat{F}^+ < F^+(\vec{\sigma}) \leq F^+(\vec{\nu'}).
\end{equation}
Observe that the AQS compatibility conditions 
\eqref{eqn:stabilityII} - \eqref{eqn:stabilityV} do not permit to establish any order relation between the lower trapping fields of the two 
loops $(\vec{\mu},\vec{\nu})$ and $(\vec{\mu'},\vec{\nu'})$. In particular, if $\hat{F}^- < F^-(\vec{\mu'})$, the maximal loop $(\vec{\mu'},\vec{\nu'})$ 
will be left through its lower endpoint before the first driving period completes. By the time the forcing reaches its lower extreme value 
$\hat{F}^-$, a new maximal loop $(\vec{\mu''},\vec{\nu''})$ must have been entered that can trap the lower amplitude so that $F^-(\vec{\mu''}) < \hat{F}^-$. 
This will happen after a sequence of inter-loop $D$-transitions, 
\eqref{eqn:interLoopD}. In this way a system 
subjected to periodic forcing,   might find itself entering and leaving a sequence of 
maximal loops until encountering a trapping one, giving thereby rise to a transient that can last multiple driving periods. 

We can graphically visualize this response, if we place the vertices representing the maximal loops $(\vec{\mu}, \vec{\nu})$ 
of the inter-loop transition graph in a 
2d Cartesian coordinate system, choosing as coordinates   
$x = -F^-(\vec{\mu})$ and $y = F^+(\vec{\nu})$. Consider now a $U$-transition between two maximal loops   
$(\vec{\mu},\vec{\nu})$ and $(\vec{\mu'},\vec{\nu'})$, as defined by \eqref{eqn:interLoopU}.  

Referring to Fig.~\ref{fig:LoopTransGraph}(b), by the condition 
\eqref{eqn:stabilityII}, $(\vec{\mu'}, \vec{\nu'}$) must lie in quadrant I or IV relative  to $(\vec{\mu}, \vec{\nu})$. We shall call this 
a type I or type IV transition. 
Likewise, for the $D$-transition, as defined by 
\eqref{eqn:interLoopD}, the condition \eqref{eqn:stabilityIV} implies that $(\vec{\mu''}, \vec{\nu''}$) must lie in quadrant I or II, so that 
we have a type I or type II transition. 
Thus the AQS dynamics does not permit transitions into quadrant III. Type I transitions, on the other hand, are special, since they lead to  maximal loops 
$(\vec{\mu'''}, \vec{\nu'''})$ satisfying 
\begin{equation}
 (F^-(\vec{\mu}),F^+(\vec{\nu}))  \subset ((F^-(\vec{\mu'''}),F^+(\vec{\nu'''})).
 \label{eqn:looptrapping}
\end{equation}
In type I transitions the destination loop maintains or improves its trapping in at least one direction, while at the same time maintaining or improving its trapping in the other. 
In II or IV transitions on the other hand, the trapping is improved in one direction but deteriorates in the other. 
The possibility of transitions of type  II and IV  can hence give rise to long transients. 
\begin{figure}[t!]
  \begin{center}
    \includegraphics[width = \textwidth]{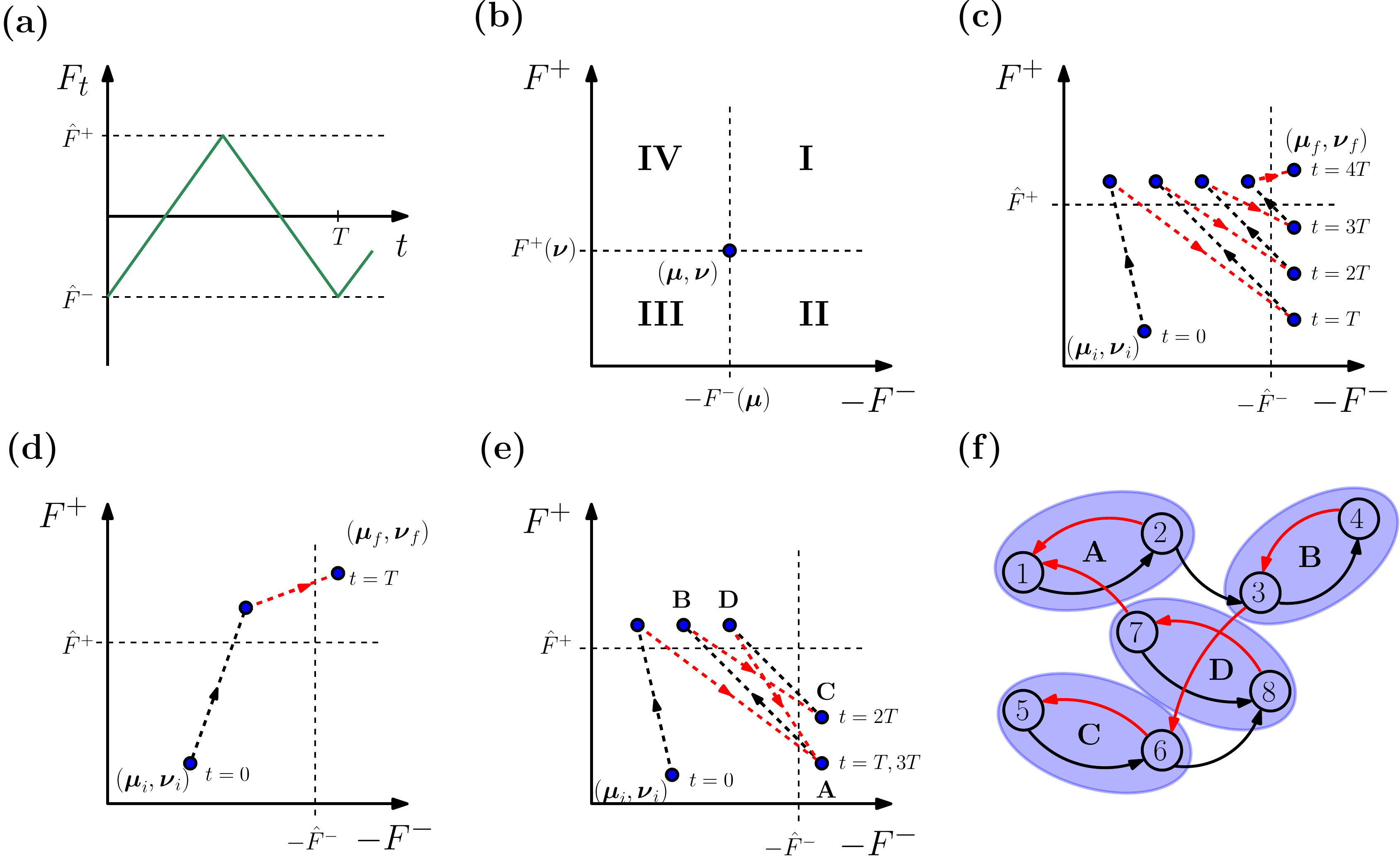} 
  \end{center}
  \caption{(a) Simple periodic forcing with period $T$. (b) The directions of inter-loop transitions. The AQS dynamics implies that transitions into quadrant III are impossible.(c)-(e): 
  Inter-loop transitions among maximal loops and their limit cycles under the simple forcing of (a). (c)-(d) harmonic response. In the case of no-passing only 
  transitions into quadrant I are possible. The transient is therefore at most one period long.   (e) An example of subharmonic response limit-cycle, where the system returns to it 
  initial state after two driving periods. (f) State transition graph generating the 
  the two-period limit-cycle shown in (e). The labeling of the loops in (f) and (e) coincide. 
  } 
  \label{fig:LoopTransGraph}
\end{figure}

Consider a system under a simple periodic forcing, as shown Fig.~\ref{fig:LoopTransGraph}(c), and starting out 
in a maximal loop $(\vec{\mu}_i, \vec{\nu}_i)$ that cannot trap 
the forcing. As the force rises from $\hat{F}^-$ to $\hat{F}^+$, the system transits through a sequence of maximal loops, 
moving monotonously upwards in the loop transition graph  
until it encounters the first maximal 
loop located above the horizontal line $F^+ =  \hat{F}^+$. This is shown by the dashed black line segments, representing the 
$U$-inter-loop transitions under a force increase. The subsequent force 
decrease from $\hat{F}^+$ to $\hat{F}^-$ causes a transition through maximal loops that continues monotonously towards the right 
until the first maximal loop to the right of the 
vertical line at $- \hat{F}^-$ is reached. These $D$-inter-loop transitions are indicated by the dashed red segments in the figure. 
The transitions between maximal loops continue until we encounter 
a maximal loop whose trapping fields $F^\pm$ satisfy both $F^+ > \hat{F}^+$ and 
$F^- < \hat{F}^-$. This is a maximal loop $(\vec{\mu}_f, \vec{\nu}_f)$, that is respectively located above and to the right of the 
horizontal and vertical dashed lines, as shown 
in Fig.~\ref{fig:LoopTransGraph}(c).
Once this loop is reached, the periodic forcing cannot induce any further inter-loop transitions and 
the loop $(\vec{\mu}_f, \vec{\nu}_f)$ therefore traps the forcing. Condition \eqref{eqn:Fconst} of Thm.~\ref{thm:trans} is thereby 
established and subsequently a periodic response sets in. 

Recall that the no-passing property implies that when subjected to periodic forcing, the  transient is at most one period long \cite{Sethna93}. 
The response depicted in Fig.~\ref{fig:LoopTransGraph}(c) can therefore not happen when the dynamics obeys the no-passing property. 
Reverting this argument suggests that a sufficient condition to ensure one period 
transients is to admit  only type I transitions  
so that \eqref{eqn:looptrapping} holds, and consequently the system trajectory follows a trajectory similar to the one depicted in 
Fig.~\ref{fig:LoopTransGraph}(d). Indeed, it turns out that no-passing implies type I transitions only and we have the following  

\begin{theorem}[No-passing and inter-loop transitions]\label{thm:NPInterloop}
 Consider an AQS-A  satisfying the no-passing property Def.~\ref{def:NP} and let 
 $(\vec{\mu},\vec{\nu})$, $(\vec{\mu'},\vec{\nu'})$ and $(\vec{\mu''},\vec{\nu''})$ be maximal loops with their $U$- and $D$-inter-loop transitions 
 as given in \eqref{eqn:interLoopU} and \eqref{eqn:interLoopD}, respectively. 
Then the following hold: 
 \begin{itemize}
  \item [(a)]  
  \begin{align*}
    F^+(\vec{\nu'}) &> F^+(\vec{\nu}), \\
    F^-(\vec{\mu'}) &\leq F^-(\vec{\mu}),
  \end{align*}
  \item [(b)]  
  \begin{align*}
    F^+(\vec{\nu''}) &\geq F^+(\vec{\nu}), \\
    F^-(\vec{\mu''}) &< F^-(\vec{\mu}).
  \end{align*}
 \end{itemize}
\end{theorem}

Theorem \ref{thm:NPInterloop} states that under no-passing all $U$- and $D$-inter-loop transitions are type I. Thus under simple forcing,  
the sequence of type I $U$-inter-loop transition ensures that the resulting maximal loop traps $\hat{F}^+$. The subsequent sequence of 
$D$-inter-loop transitions ensure that $\hat{F}^-$ is trapped as well. A trapping loop is therefore always reached by the end of the 
first driving period. 

In the absence of no-passing, $\ell$AQS-A limit-cycles involving more than one maximal loop are possible and give rise to {\em subharmonic} response \cite{Deutschetal2003,Nagel2017}. 
An  example is shown in Fig.~\ref{fig:LoopTransGraph} (e) and (f) and we conclude with a realization of this.

\begin{example}[A subharmonic limit-cycle]
\label{ex:subharmonic} 

Consider the limit-cycle shown in Fig.~\ref{fig:LoopTransGraph} (f) and choose the trapping fields of the states  as
\begin{eqnarray}
 F(1) &= (-5,0), \ \ \ F(2) = (-3,1), \nonumber \\ 
 F(3) &= (-2,3), \ \ \ F(4) = (-1,5), \nonumber \\
 F(5) &= (-5,1), \ \ \ F(6) = (-3,2), \nonumber \\
 F(7) &= (-3,4), \ \ \ F(8) = (-2,5), \nonumber
\end{eqnarray}
where the first (second) values on the right hand sides give the trapping field in the negative (positive) direction. 
The trapping fields for the maximal loops $A, B, C$, and $D$ are therefore:
\begin{eqnarray}
 F(A) &= (-5,1), \ \ \ F(B) = (-2,5), \nonumber \\
 F(C) &= (-5,2), \ \ \ F(D) = (-3,5). \nonumber
\end{eqnarray}
When we periodically drive the system with fields $\hat{F}^\pm=\pm4$ and start in state $1$ of the maximal loop $A$, the system will pass through the sequence of maximal loops given by
$\mathcal{F}[\hat{F}^-, \hat{F}^+]A=(A,B,C,D,A,\cdots)$, 
returning to the maximal loop $A$ every second driving 
cycle, as is readily checked. The period of the limit cycle is thus $2T$ and we have subharmonic response. 
Examples of subharmonic response with higher multiples of the driving period are readily constructed. $\maltese$
\end{example}

\section{Summary and Discussion}
\label{sec:discussion}

We have developed a graph theoretical approach to describe discrete state transitions in driven athermal disordered systems. 
Two maps, $U$ and $D$, capture the transitions between quasi-static states under increases and decreases of a uniform driving 
force. The athermal quasi-static (AQS) nature of the dynamics, 
implies that the functional graphs of these maps are acyclic. We have termed the set of quasi-static states along with the two maps 
capturing the adiabatic response under force changes, {\em AQS automata} (AQS-A). The hysteretic response of the AQS-A depends on the way in which the 
two acyclic maps $U$ and $D$ act on the common set of states. 

In particular, our interest has been in the return-point-memory (RPM) property. In  
the context of magnets, RPM is defined with respect to an intrinsic partial order on the spin configurations along with a dynamics that 
preserves this partial order, the no-passing (NP) property\cite{NoPassing,Sethna93}. While NP is a sufficient condition for RPM, 
it is not necessary. As we have demonstrated, the acyclic maps $U$ and $D$ provide an alternative partial order, the {\em dynamic partial 
order} Def.~\ref{def:dynPO}, in terms of which the RPM property can be defined, Def.~\ref{def:RPMII}. The RPM property defined in this 
way by-passes the need of invoking the NP property. The case when RPM follows as a result of NP becomes just a special case. 
We have called this the $\ell$RPM property in order to emphasize its detachment from NP. 

We then turned to the analysis of the state transitions graphs 
of AQS-A with the $\ell$RPM property. Central to our analysis is the notion of a loop: a pair of states $\vec{\mu}$ and 
$\vec{\nu}$ mapped into each other under repeated applications of $U$, respectively $D$, Def.~\ref{def:loopdef}. 
The $\ell$RPM property prescribes the 
hierarchy of sub loops and the transition among the states associated with a loop, the {\em intra-loop} transitions. This hierarchy can be understood in terms of 
a recursive partition procedure that breaks a loop down into disjoint constituent loops and has a natural representation in terms of 
an ordered tree, Thm.~\ref{prop:mert}. Thus $\ell$RPM governs the intra-loop structure of the state transition graph. 
We analyzed next the inter-loop transitions of the AQS-A. Basic to this description is the notion of a maximal 
loop, Def.~\ref{def:maxLoop}. A maximal loop is a loop that is not a sub loop of any other loop. Each configuration of the AQS-A 
is associated with precisely one maximal loop and thus the maximal loops partition the set of quasi-static configurations. 
This naturally leads to a higher level description of the AQS-A in terms of transitions between maximal loops, the inter-loop transition 
graph, Def.~\ref{def:interLoop}. We showed that while the $\ell$RPM property strongly constrains the intra-loop transitions, it does 
not severely restrict the inter-loop transitions. Such restrictions turn out to be the consequence of additional features, such as the 
NP property.  

In particular, the NP property not only establishes RPM, but also imposes additional constraints on the inter-loop transitions, 
Thm.~\ref{thm:NPInterloop}. At the same time, we have demonstrated that features of the 
inter-loop transition graph bear directly on the length of the transient response when the AQS-A is subject to a time-periodic driving force. 
As is well-known \cite{Sethna93}, AQS-A with the NP property exhibit extremely short transients, lasting at most one period.  
On the other hand, when $\ell$RPM is present but the NP property does not hold, the transients can be arbitrarily long and the response can be even 
{\em subharmonic}, {\em cf.} panels (c) and (e) of Fig.~\ref{fig:LoopTransGraph}. 

Thus the following picture emerges for the dynamics of AQS-A with the $\ell$RPM property: $\ell$RPM organizes state transition at the 
``local'' intra-loop level of maximal loops and their hierarchy of sub-loops, while leaving transitions at the inter-loop largely unconstrained.

With the possibility of experimentally observing microstates in a variety of driven disordered systems, such as arrays of ferromagnetic 
nanoislands \cite{Gilbert2016}, or colloidal spin ice systems 
\cite{Hanetal2008,Libaletal2012}, it should be possible to determine how closely the intra-loop structure resembles 
predictions of the corresponding $\ell$AQS-A, Theorems \ref{prop:mert} and \ref{thm:planar}. This would be particularly interesting, since some of 
these disordered systems exhibit RPM, but apparently not NP. 

Our results indicate how in such systems one could test for the simultaneous absence of NP and  
presence of RPM,  namely by identifying $\ell$RPM loops for which the 
marginality condition, Def.~\ref{def:loopmarginality}, is violated, 
since -- as we have seen -- NP implies marginality. 
Moreover, the systems mentioned above exhibit RPM  only for a range of tunable parameters. 
One could also track experimentally how the corresponding state transition graph changes as system parameters are varied and 
these system are thereby moved  in and out of a regime where RPM holds. 


From a mathematical point of view it is also of interest to ask for the properties of the decomposition trees 
describing  the standard partition of  maximal loops into their sub loops. Since the mappings $U$ and $D$ underlying 
these trees are random objects, so are the trees themselves. In particular, what happens to the hierarchical organization of 
the intra-loop structure in the limit that the system size becomes arbitrarily large?
What are the statistical properties of such trees?

\subsection*{Acknowledgments}

The authors would like to thank M. I\c{s}eri and D.C. Kaspar for many fruitful exchanges. They also acknowledge discussions with 
B. Behringer, A. Bovier, K. Dahmen, N. Keim, J. Krug, C. Maloney, A. A. Middleton, S. Nagel, A. Rosso, S. Sastry, K. Sekimoto, 
J. Sethna, D. Vandembroucq, T. A. Witten, and A. Y{\i}lmaz. 
Many of these took place during the Memory Formation in Matter program of KITP and the authors thank KITP for 
its kind hospitality.


\bibliographystyle{unsrt}%
\bibliography{toy_model}%

\appendix

\section{Proofs}
\label{app:proofs}

\begin{proof}[Lemma \ref{lem:NP}]
 We shall only prove part (a) of the lemma, since the proof of part (b) follows the same reasoning. By the stability condition 
 \eqref{eqn:stabilityIV}, there exist forces $F_T < F_0$ such that $\vec{\mu}$ is stable at $F_T$ and $\vec{\nu}$ is stable 
 at $F_0$. Let $(F_3(t))_{0 \leq t \leq T}$ be a forcing monotonously decreasing from $F_0$ to $F_T$. Likewise let 
 $F_1(t))_{0 \leq t \leq T}$ be a constant force with $F_1(t) = F_T$. Then for all $0 \leq t \leq T$,  
 $F_1(t) \leq F_3(t)$. Apply $F_1(t)$ to $\vec{\mu}$ and $F_3(t)$ to $\vec{\nu}$. By the choice for $F_3(t)$, the 
 condition that $\vec{\mu} \in D^*\vec{\nu}$ and \eqref{eqn:FDaction}, it follows that at 
 time $t = T$ we have $\vec{\mu}(T) = \vec{\mu} = \vec{\nu}(T)$. Now apply $F_1(t)$ to $\vec{\sigma}$. The no-passing 
 property implies that $\vec{\mu}(T) \preceq \vec{\sigma}(T)$.  Likewise, from the application of $F_3(t)$ we have 
 $ \vec{\sigma}(T) \preceq \vec{\nu}(T)$. But $\vec{\mu}(T) = \vec{\nu}(T) = \vec{\mu}$ and it follows that  
 $\vec{\sigma}(T) = \vec{\mu}$.
\end{proof}

\begin{proof}[Lemma \ref{lem:mutreach}]
 Note that since the two states $\vec{\sigma}_1, \vec{\sigma}_2 \in \mathcal{R}$ are reachable, there exist field histories 
 $\Phi_1$ and $\Phi_2$ that map $\vec{\omega}$ to these states. However, $\vec{\omega}$ is an absorbing  state and 
 there exists a positive integer $n$, such that $U^n \vec{\sigma}_1 = \vec{\omega}$, thus starting from $\vec{\sigma}_1$ and applying 
 the field-history $U^n$ followed by $\Phi_2$, will lead to the state $\vec{\sigma}_2$. By the 
same argument, $\vec{\sigma}_1$ is field-reachable from $\vec{\sigma}_2$, so that any two reachable states are mutually 
field-reachable. 

 To show that mutual reachability is an equivalence relation, we need to show reflexivity, symmetry and transitivity. Reflexivity is trivial
  and symmetry has already been shown above.  Transitivity is shown similarly: if $\vec{\sigma}_2$ is field-reachable 
  from $\vec{\sigma}_1$, and $\vec{\sigma}_3$ is field-reachable from $\vec{\sigma}_2$, then it follows that $\vec{\sigma}_1$ 
  is field-reachable from $\vec{\sigma}_3$: we simply consider a field-history that leads takes $\vec{\sigma}_1$ via $\vec{\sigma}_2$ to 
  $\vec{\sigma}_3$. Using the symmetry of the relation, it follows that $\vec{\sigma}_1$ and $\vec{\sigma}_3$ are mutually reachable. 
  Thus the set $\mathcal{R}$ is an equivalence class under mutual reachability. 
  Closeness of $\mathcal{R}$ under $U$ and $D$ follows by considering field histories to which we append a $U$ or $D$ step. 
  The appended field-history necessarily produces a reachable state. 
\end{proof}


In preparation for the proof of Proposition~\ref{prop:munuorbitLemma}, we need to 
first establish some intermediate results. Given a loop $(\vec{\mu},\vec{\nu})$ we first construct sequences of nested sub-loops and derive the order 
relations amongst its endpoints. Nested sub loop sequences can be associated with $(\vec{\mu},\vec{\nu})$-confined field-histories and we show that these 
satisfy Proposition~\ref{prop:munuorbitLemma}. The proof of the proposition then follows by observing that any $(\vec{\mu},\vec{\nu})$-reachable 
state has a field-history that is associated with a nested sub loop sequence and the result follows.  

Let $(\vec{\mu},\vec{\nu})$ be the endpoints of a loop and  
consider an intermediate node $\vec{\mu}_1$ on its $D$-boundary. By the $\ell$RPM property $\vec{\mu}_1$ is the endpoint of 
the loop $(\vec{\mu}_1,\vec{\nu})$. Next, consider an intermediate state on the $U$-boundary of this loop and label it $\vec{\nu}_1$. The pair 
$(\vec{\mu}_1,\vec{\nu}_1)$ again forms the endpoints of a loop. Then, consider again a site $\vec{\mu}_2$ on its $D$-boundary to obtain the loop 
$(\vec{\mu}_2,\vec{\nu}_1)$. Continuing in this manner, by alternately picking intermediate boundary sites from the $U$ and $D$-boundaries of the 
previous loop results in the sequence of nested loops 
\begin{equation}
 (\vec{\mu},\vec{\nu}), (\vec{\mu}_1,\vec{\nu}), (\vec{\mu}_1,\vec{\nu}_1), (\vec{\mu}_2,\vec{\nu}_1), \ldots.
 \label{eqn:loopSeq}
\end{equation}
Since the subsequent state picked,  $\vec{\mu}_k$ or $\vec{\nu}_k$, is required to be an intermediate boundary state and the number of 
configuration is finite, this procedures has to terminate in a {\em 2-loop},  a pair of states 
 $(\vec{\upsilon}, \vec{\sigma})$ such that $U\vec{\upsilon} = \vec{\sigma}$ and $D\vec{\sigma} = \vec{\upsilon}_m$, whose boundary 
 does not posses intermediate states. The states $\vec{\mu}_i$, or $\vec{\nu}_i$ appearing in the above loop sequence will be called 
 {\em switch-back states}. It is convenient to regard the endpoints $\vec{\mu}$ and $\vec{\nu}$ 
as switch-back states as well and we set $\vec{\mu}_0 = \vec{\mu}$ and $\vec{\nu}_0 = \vec{\nu}$.
 The above nested loop sequence terminates in the 2-loop  $(\vec{\mu}_n, \vec{\nu}_n)$ or $(\vec{\mu}_{n+1}, \vec{\nu}_{n})$, for some integer $n \geq 0$.
Alternatively, we could have started with an intermediate node $\vec{\nu}_1$ on the $U$-boundary of $(\vec{\mu},\vec{\nu})$, leading 
to the nested sequence of loops
\begin{equation}
 (\vec{\mu},\vec{\nu}), (\vec{\mu},\vec{\nu}_1), (\vec{\mu_1},\vec{\nu}_1), (\vec{\mu_1},\vec{\nu}_2), \ldots. 
 \label{eqn:loopSeqomega}
\end{equation}  
Again, this sequence terminates with a pair of switch-back states forming a 2-loop whose endpoints are either $(\vec{\mu}_n, \vec{\nu}_n)$ or 
$(\vec{\mu}_n, \vec{\nu}_{n+1})$ for some integer $n \geq 0$. The nested sequence of loops is illustrated in Fig.~\ref{fig:LoopDec}

\begin{figure}[t!]
  \begin{center}
    \includegraphics[width = 0.8\textwidth]{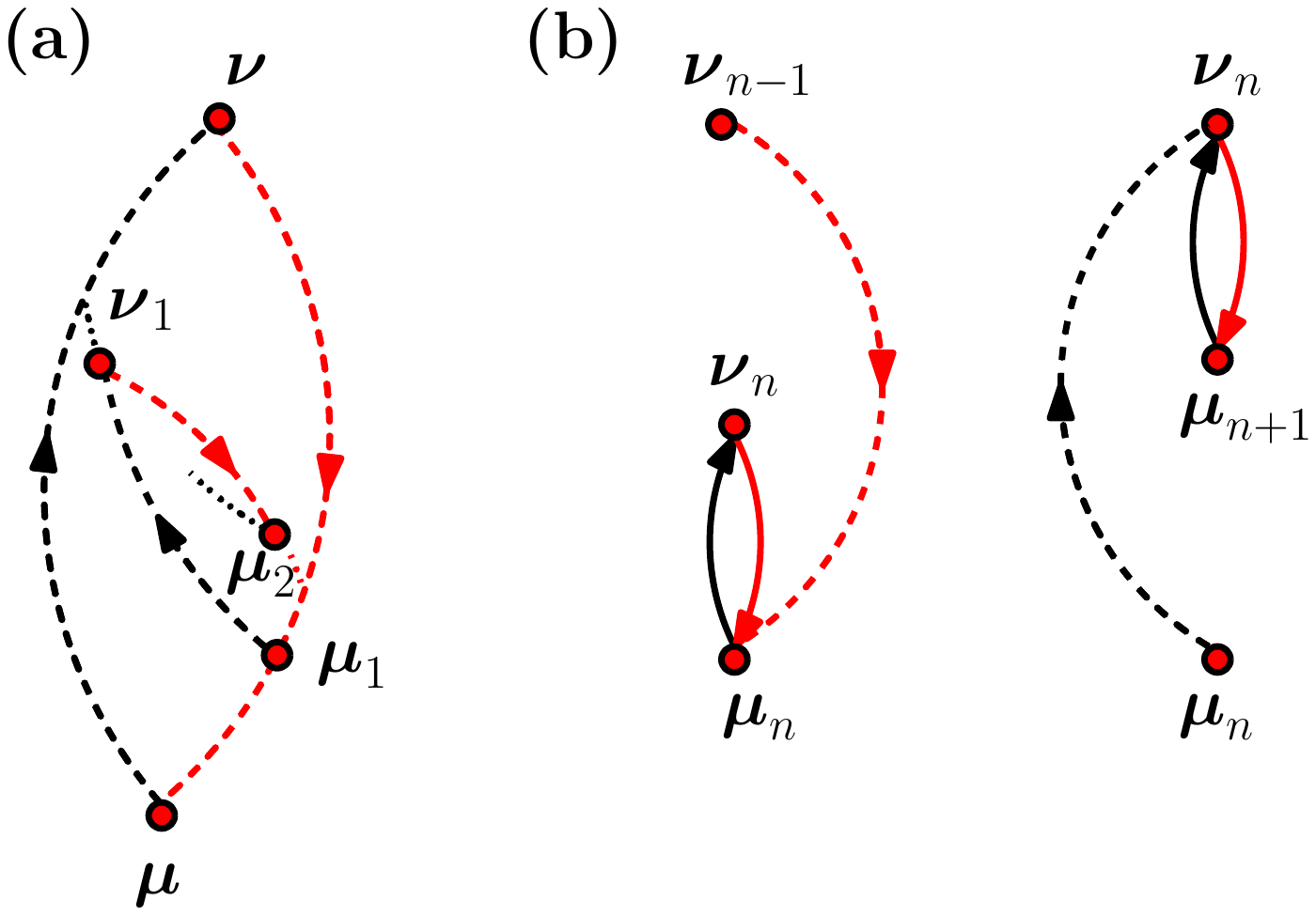} 
  \end{center}
  \caption{A nested sequence of sub loops associated with the loop $(\vec{\mu},\vec{\nu})$. (a) By the $\ell$RPM property, the pairs $(\vec{\mu}_1,\vec{\nu})$, 
  $(\vec{\mu}_1,\vec{\nu}_1)$, $(\vec{\mu}_2,\vec{\nu}_1), \ldots$, form sub loops, one nested inside the other. Due to the finiteness 
  of the number of states, this nesting must terminate at a {\em $2$-loop}. (b) The two possible terminations for the 
  decomposition of (a), depending on whether the $2$-loop is entered on a $U$ step (black solid arrow) or $D$ step (red solid arrow). 
  } 
  \label{fig:LoopDec}
\end{figure}

One can think of the sequence of nested loops as providing a field-history for the switch-back states. For example, considering the loop 
sequence \eqref{eqn:loopSeqomega} and assuming that it 
terminates in the 2-loop $(\vec{\mu}_p,\vec{\nu}_p)$, we have 
\begin{equation}
 \vec{\mu}_p = D^{m_{p}}U^{n_{p}} D^{m_{p-1}} \cdots U^{n_2} D^{m_1} U^{n_1} \vec{\mu}, 
 \label{eqn:FH1}
\end{equation}
for some non-negative integers $m_1,n_1, m_2, n_2, \ldots, m_p, n_p$.   
We will call the transit states  
\begin{align*}
 \vec{\nu}_1 &= U^{n_1} \vec{\mu}, \\
 \vec{\nu}_2 &= U^{n_2} D^{m_1} U^{n_1} \vec{\mu}, \\
     &\cdots \\
 \vec{\nu}_{p} &= U^{n_{p}} D^{m_{p-1}} \cdots U^{n_2} D^{m_1} U^{n_1} \vec{\mu}, 
\end{align*}
and 
\begin{align*}
 \vec{\mu}_1 &= D^{m_1} U^{n_1} \vec{\mu}, \\
 \vec{\mu}_2 &= D^{m_2}U^{n_2} D^{m_1} U^{n_1} \vec{\mu}, \\
     &\cdots \\
 \vec{\mu}_{p} &= D^{m_{p}}U^{n_{p}} D^{m_{p-1}} \cdots U^{n_2} D^{m_1} U^{n_1} \vec{\mu},
\end{align*}
the {\em switch-back states} of the loop sequence \eqref{eqn:loopSeqomega}. By construction, the field-history \eqref{eqn:FH1} is 
$(\vec{\mu}, \vec{\nu})$-confined.

\begin{lemma}[Ordering of switch-back states]\label{lem:sworder}
Let $(\vec{\mu},\vec{\nu})$ form the end points of a loop and assume that 
a sequence of nested loops  of the form \eqref{eqn:loopSeq} or \eqref{eqn:loopSeqomega} is given.   
The partial order of the switch-back states obeys  
\begin{equation}
 \vec{\mu} \prec_D \vec{\mu}_1 \prec_D \vec{\mu}_2 \prec_D \cdots  \prec_D \vec{\mu}_p < \vec{\nu}_q \prec_U \cdots 
 \prec_U \vec{\nu}_2 \prec_U \vec{\nu}_1 \prec_U \vec{\nu}, 
 \label{eqn:switchbackorder}
\end{equation}
where $(\vec{\mu}_p, \vec{\nu}_q)$ is the terminating 2-loop and $(p,q) \in \{(n,n), (n,n+1), (n+1,n)\}$.
\end{lemma}

\begin{proof}
 From the construction of the nested loop sequence, \eqref{eqn:loopSeq} or \eqref{eqn:loopSeqomega}, it immediately follows that $\vec{\mu}_{i+1}$ is on the 
 $D$-boundary of a loop whose lower endpoint is $\vec{\mu}_i$, hence $\vec{\mu}_{i} \prec_D \vec{\mu}_{i+1}$. A similar 
 argument yields that $\vec{\nu}_{i+1} \prec_U \vec{\nu}_i$. Finally, the terminating $2$-loop of the decomposition has as lower and upper endpoints 
 the states $\vec{\mu}_p$ and $\vec{\nu}_q$, respectively. Thus $\vec{\mu}_p < \vec{\nu}_q$ and \eqref{eqn:switchbackorder} follows.
\end{proof}

We consider the $U$ and $D$ orbits of the switch-back states associated with a nested loop sequence and show that they have to pass 
through the corresponding endpoints of the confining loop. 

\begin{lemma}[Orbits of switch-back states]\label{lem:switchbackorbit}
Assume that the pair of states $(\vec{\mu},\vec{\nu})$ forms a loop. Let a nested loop sequence be given for which $\vec{\sigma}$ is a switch-back state. 
Then 
 \begin{equation}
 \vec{\nu} \in U^*\vec{\sigma}, \quad \mbox{and} \quad \vec{\mu} \in D^*\vec{\sigma}.
\end{equation}
\end{lemma}

\begin{proof}
Since $\vec{\sigma}$ is a switch-back state it is either the lower or upper endpoint of at least one loop in the sequences 
\eqref{eqn:loopSeq} or \eqref{eqn:loopSeqomega}. Assume first that it is a lower endpoint, so that 
\[
 \vec{\sigma} = \vec{\mu}_j
\]
for some $j \geq 0$ and the sub loop containing it be of the form $(\vec{\sigma}, \vec{\nu}_k)$, for some corresponding integer $k$. 
By Lemma \ref{lem:sworder} it follows then that $\vec{\mu} \prec_D \vec{\mu_1} \prec_D \cdots \prec_D \vec{\mu_j} = \vec{\sigma}$. 
But this is equivalent to $\vec{\mu} \in D^*\vec{\sigma}$. Likewise, since $(\vec{\sigma}, \vec{\nu}_k)$ forms a loop 
$\vec{\sigma} \prec_U \vec{\nu}_k$ and again from Lemma \ref{lem:sworder} and the transitivity of $\prec_U$ it follows that 
$\vec{\sigma} \prec_U \vec{\nu}$ and hence $\vec{\nu} \in D^*\vec{\sigma}$. 
 
The case where $\vec{\sigma}$ is the upper endpoint of a sub loop in the sequence proceeds along similar lines. 
\end{proof}

By construction, given a loop $(\vec{\mu}, \vec{\nu})$ and a sequence of nested sub loops, the corresponding switch-back states are 
all field-reachable from the endpoints. In preparation for Lemma \ref{lem:switchbackreachable} and Prop.~\ref{prop:algTermination}, it
is useful to introduce the notion of a {\em reduced field-history}. 
\begin{definition}[Reduced field-history of switch-back states]\label{def:reducedFH}
 Let a loop $(\vec{\mu}, \vec{\nu})$ and a sequence of nested sub loops be given for which $\vec{\sigma}$ is a switch-back state. The  
 field-history of $\vec{\sigma}$ will be called a {\em reduced field-history}. It has one of the four forms below, 
 \begin{align*}
  \vec{\sigma} = U^{n_{k}} D^{m_{k-1}} \cdots U^{n_2} D^{m_1} U^{n_1} \vec{\mu} &\equiv U^{n_{k}} \vec{\mu}_{k-1},  \\
  \vec{\sigma} = D^{m_{k}}U^{n_{k}}  \cdots U^{n_2} D^{m_1} U^{n_1} \vec{\mu} &\equiv D^{m_{k}} \vec{\nu}_k,   \\
  \vec{\sigma} = U^{n_{k}} D^{m_{k}} \cdots U^{n_2} D^{m_1}  \vec{\nu} &\equiv U^{n_{k}} \vec{\nu}_{k},   \\
  \vec{\sigma} = D^{m_{k}}U^{n_{k-1}} \cdots U^{n_2} D^{m_1} \vec{\nu} &\equiv D^{m_{k}} \vec{\nu}_{k-1},
 \end{align*}
  for some integer $k > 0$, and denoting its switch-back states  as $\vec{\mu}, \vec{\mu}_1, \ldots, \vec{\mu}_p$ and $\vec{\nu}, \vec{\nu}_1, \ldots, \vec{\nu}_q$, 
  their partial ordering is given by, {\em cf.} \eqref{eqn:switchbackorder}, 
  \begin{equation}
 \vec{\mu} \prec_D \vec{\mu}_1 \prec_D \vec{\mu}_2 \prec_D \cdots  \prec_D \vec{\mu}_p \prec_D \vec{\sigma} \prec_U \vec{\nu}_q \prec_U \cdots 
 \prec_U \vec{\nu}_2 \prec_U \vec{\nu}_1 \prec_U \vec{\nu},    
  \end{equation}
  with $(p,q) \in \{(k-1,k-1), (k-1,k), (k,k-1)\}$.
\end{definition}

From Def.~\ref{def:munuconfined} it follows that any switch-back state $\vec{\sigma}$, 
due to its association with some nested sub loop sequence of $(\vec{\mu},\vec{\nu})$, 
belongs to $\mathcal{R}_{\vec{\mu}, \vec{\nu}}$. We show next that the  converse is true as well: 

\begin{lemma}[Reachable states are switch-back states]\label{lem:switchbackreachable}
Let the pair $(\vec{\mu}, \vec{\nu})$ be the endpoints of a loop. The following two are equivalent:
\begin{itemize}
 \item [(a)] $\vec{\sigma} \in \mathcal{R}_{\vec{\mu}, \vec{\nu}}$.
 \item [(b)] $(\vec{\mu}, \vec{\nu})$ has admits a sequence of nested sub loops containing $\vec{\sigma}$ as a switch-back state.
\end{itemize}
\end{lemma}

\begin{proof}
 $(b) \Rightarrow (a)$: This has already been shown: the sequence of nested sub loops defines a $(\vec{\mu}, \vec{\nu})$-confined field-history that is moreover in the reduced form of 
 Def.~\ref{def:reducedFH}, hence 
$\vec{\sigma}$ is $(\vec{\mu},\vec{\nu})$-reachable. 

$(a) \Rightarrow (b)$: Let us start with the special case when $\vec{\sigma}$ is an 
endpoint of the loop, $\vec{\sigma} = \vec{\mu}$ or $\vec{\sigma} = \vec{\nu}$. 
By definition, $\vec{\mu}, \vec{\nu} \in \mathcal{R}_{ \vec{\mu}, \vec{\nu} }$, and in addition   
the endpoints are switch-back states of any loop decomposition. The claim follows. 

Now let us assume that $\vec{\sigma} \in \mathcal{R}_{ \vec{\mu}, \vec{\nu}}$ is not an endpoint of the loop. 
There exists a $(\vec{\mu},\vec{\nu})$-confined history leading to $\vec{\sigma}$. We show that this field-history 
can always be converted into a reduced field-history so that every $(\vec{\mu},\vec{\nu})$-reachable state $\vec{\sigma}$ 
has a reduced field-history. The claim follows then from Lemma \ref{lem:switchbackorbit}. 

Suppose first that the $(\vec{\mu},\vec{\nu})$-confined history
is of the form 
\eqref{eqn:sigmaFH}, with its switch-back states given by \eqref{eqn:sigmaSB}. The case where the field-history 
starts from $\vec{\nu}$, \eqref{eqn:sigmaFH2} proceeds along similar lines and we will therefore omit it. 
Consider the following algorithm:
\begin{itemize}
 \item[(i)] Start from the switch-back state $\vec{\nu}_1$ and find the first switch-back state of the $(\vec{\mu},\vec{\nu})$-confined field 
 history \eqref{eqn:sigmaFH} for which one of the inequalities $\vec{\nu}_i \prec_U \vec{\nu}_{i-1}$, or $\vec{\mu}_i \prec_D \vec{\mu}_{i-1}$ 
 does not hold. 
 \item[(ii)] Due to the $\ell$RPM property Def.~\ref{def:RPM}, $\vec{\nu}_{i-1} \in U^* \vec{\mu}_{i-1}$ and $\vec{\nu}_{i} \in U^* \vec{\mu}_{i-1}$. 
 Therefore, $\vec{\nu}_{i} \in U^* \vec{\nu}_{i-1}$. In other words we do not need the switch-back states $\vec{\nu}_{i-1}$ and $\vec{\mu}_{i-1}$ 
 and we bypass them. The new field-history becomes,
  \begin{equation}
    \vec{\sigma} = U^{n_p} D^{m_{p-1}} \cdots U^{n_{i}'} D^{m_{i-2}}  \cdots D^{m_1} U^{n_1} \vec{\mu},  
    \label{eqn:sigmaFH1}
 \end{equation}
 where $\vec{\nu}_{i-1}$ and $\vec{\mu}_{i-1}$ are no longer switch-back states. 
 \item[(iii)] Repeat the first two steps until the remaining switch-back states satisfy \eqref{eqn:switchbackorder}. 
\end{itemize}
The pruned field-history obtained in this way can be associated naturally with the initial steps of a nested sub loop sequence for which 
$\vec{\sigma}$ is a switch-back state. It is a reduced field-history. 
We have therefore shown that a loop-confined field-history can always be converted into a field-history 
associated with nested sub loop sequence. The claim follows.  
\end{proof}

\begin{proof}[Proposition \ref{prop:munuorbitLemma}]
The proof of the proposition follows from the application of Lemmas \ref{lem:switchbackorbit} and \ref{lem:switchbackreachable}. 
\end{proof}

\begin{proof}[Proposition \ref{prop:nestinglemma}]
 We only prove part (a) of the proposition. Part (b) follows from a similar argument. 
 
 (i) is a direct consequence of the RPM 
 property, Def.~\ref{def:RPM}, applied to the $(\vec{\mu}, \vec{\nu})$ -- loop. 
 
(ii) Let $\vec{\upsilon}$ be an interior node of the 
$(\vec{\mu}, \vec{\sigma})$ -- loop. By Corollary \ref{prop:munuorbitLemma}, we have that 
$\vec{\sigma} \in U^*\vec{\upsilon}$, moreover, since $\vec{\eta} \in U^*\vec{\sigma}$, we also have
\[
   \vec{\eta} \in U^*\vec{\upsilon}.
\]
Let us now assume that $\vec{\upsilon} \in D^*\vec{\eta}$. We show that this implies that 
$D^*\vec{\eta}$ passes through the boundary of the loop $(\vec{\mu},\vec{\sigma})$: by assumption, the pair 
of nodes $(\vec{\upsilon}, \vec{\eta})$ form the endpoints of a loop, but we have $\vec{\upsilon} \prec_U 
\vec{\sigma} \prec_U \vec{\eta}$. Invoking the RPM property it therefore follows that $\vec{\upsilon}  \in D^*\vec{\sigma}$, implying 
that $\vec{\upsilon}$ is a boundary node of the $D$-boundary of the $(\vec{\mu}, \vec{\sigma})$ -- loop. 
This is a contradiction. We therefore conclude that if 
$\vec{\upsilon} \in \mathcal{R}_{\vec{\mu}, \vec{\sigma}}$ and $\vec{\upsilon} \in D^*(U\vec{\sigma})$ then this implies 
that $\vec{\upsilon} \in D^*\vec{\sigma}$. 

\end{proof}

\begin{proof}[Proposition \ref{prop:bondarymon}] 
(i) Apply the Definition \ref{def:psidef} to the boundary states $\vec{\mu}$ and $\vec{\nu}$, noting that they are the endpoints of the loop. (ii) is a direct consequence 
of Prop. \ref{prop:nestinglemma}. In the following we assume that the boundary nodes $\vec{\nu}_i$ and $\vec{\mu}_k$ are 
indexed as in \eqref{eqn:nuidef} and \eqref{eqn:muidef}. Consider the map $\psi_-$ and note that 
when $\vec{\nu}_i = \vec{\nu}$ then there is no boundary point $\vec{\nu}_j$ on the loop such that $\vec{\nu}_i \prec_U \vec{\nu}_j$. 
Consequently, there is nothing to prove. We therefore assume that  $\vec{\nu}_i$ is an intermediate boundary point and suppose that $\psi_-(\vec{\nu}_i) = \vec{\mu}_k$ for some $D$-boundary point $\vec{\mu}_k$.  Consider now the $D$-orbit of $\vec{\nu}_{i+1} = U \vec{\nu}_i$. By Prop.~\ref{prop:nestinglemma}, 
the orbit $D^*\vec{\nu}_{i+1}$ can (a) either contain states on the $D$-boundary of the $(\vec{\mu}, \vec{\nu}_i)$ 
(but none of its interior states), or (b), be   
incident on the $D$-boundary of $(\vec{\mu}, \vec{\nu})$ at a state $\vec{\mu}_{k'} \succ_D \vec{\mu}_k$. In the former case we have 
$\psi_-(\vec{\nu}_i) = \psi_-(\vec{\nu}_{i+1}) = \vec{\mu}_k$, while in the latter case 
$\psi_-(\vec{\nu}_{i+1}) = \vec{\mu}_{k'} \succ_D \vec{\mu}_k = \psi_-(\vec{\nu}_i)$. 
In either case it follows that 
\begin{equation}
 \psi_-(\vec{\nu}_i) \preceq_D \psi_-(\vec{\nu}_{i+1}).
 \label{eqn:iip1}
\end{equation}
By part (i) of the proposition, inequality \eqref{eqn:iip1} holds also when $\vec{\nu}_{i+1} = \vec{\nu}$ is the endpoint. 
Using induction, it is then readily shown that inequality \eqref{eqn:iip1} holds when replacing $i$ 
successively  $i + 1, i+ 2, \ldots, j - 1$, that is, when considering the pair of adjacent boundary states $\vec{\nu}_{i+1}$ and $\vec{\nu}_{i+2}$, $\vec{\nu}_{i+2}$ and $\vec{\nu}_{i+3}$ {\em etc.} Thus the first part of claim (ii) is proven.  
The proof of the second part of (ii) follows by applying a similar argument 
for successive states $\vec{\mu}_i$ and $\vec{\mu}_{i+1}$ on the $D$-boundary of the loop $(\vec{\mu}, \vec{\nu})$. 
\end{proof}

\begin{proof}[Lemma \ref{lem:iminter}] 
(i) $\Rightarrow$ (ii): Since $\vec{\kappa}$ is an intersection point, it lies on both boundaries. Using Def. \ref{def:psidef} it is then shown 
that this implies (ii) (and (iii)). 
The implication (ii) $\Rightarrow$ (iii) and (iii) $\Rightarrow$ (i) follow likewise from Def. \ref{def:psidef}.   
\end{proof}

\begin{proof}[Lemma \ref{lem:invImages}]
 By definition, $\psi^{-1}_-(\vec{\mu}_i)$ is the set of boundary points  on the $U$-boundary that are mapped into 
 $\vec{\mu}_i$ under $\psi_-$. Among these boundary points there must be a least element with respect to the order 
 $\prec_U$. Suppose this is $\vec{\nu}_{k'}$. Now consider the image of the adjacent boundary point 
 $\vec{\nu}_{k'+1} = U \vec{\nu}_{k}$ under $U$. By Prop.~\ref{prop:bondarymon}, either $\psi_-(\vec{\nu}_{k'+1}) = 
 \vec{\mu}_i$, or $\psi_-(\vec{\nu}_{k'+1}) \succ_D \vec{\mu}_i$. In the first case $\vec{\nu}_{k'+1} \in \psi^{-1}_-(\vec{\mu}_i)$, in 
 the second case $\vec{\nu}_{k'+1} \notin \psi^{-1}_-(\vec{\mu}_i)$. The latter implies however that 
 $\vec{\nu}_{s} \notin \psi^{-1}_-(\vec{\mu}_i)$ for any boundary state $\vec{\nu}_{s} \succeq_U \vec{\nu}_{k'+1}$. Hence there 
 must be a $k \geq k'$ such that for all $k' \leq s \leq k$, $\psi_-(\vec{\nu}_{s}) = \vec{\mu}_i$. The pre-image  
 $\psi^{-1}_-(\vec{\mu}_i)$ is thus an interval. The proof of \eqref{eqn:downTriangleLoop_2} proceeds along similar lines and 
 assertion (i) is proven. 

 To prove assertion (ii) of the Lemma, observe that the pair $(\vec{\mu}, \vec{\nu}_{k'})$ forms a loop, since by the RPM 
 property $\vec{\mu} \in D^*\vec{\nu}_{k'}$. From \eqref{eqn:downTriangleLoop} it follows that $\vec{\mu}_i \in D^*\vec{\nu}_{k'}$, 
 but $\vec{\mu}_{i}$ is on the $D$-boundary of the loop $(\vec{\mu}, \vec{\nu}_{k'})$ and hence using the RPM property again, 
 $\vec{\nu}_{k'} \in U^*\vec{\mu}_i$ and hence the pair $(\vec{\mu}_i, \vec{\nu}_k)$ forms a loop. The proof that $(\vec{\mu}_{l'}, \vec{\nu}_j)$
 forms a loop is proven in an analogous manner. 
\end{proof}

\begin{proof}[Lemma \ref{lemma:stacking}]
The situation is illustrated in panel (a) of Fig.~\ref{fig:StackingProof}. By assumption, $\psi_-(\vec{\nu}_k) = \vec{\mu}_i$  
and by the monotonicity of the boundary map $\psi_-$, Prop.~\ref{prop:bondarymon}, it must be that  
$\psi_-(\vec{\nu}_{k+1}) \succ_D \vec{\mu}_i$. Suppose that 
$\psi_-(\vec{\nu}_{k+1}) = \vec{\mu}_l$, for some $l > i$. Let us assume that $l > i + 1$, so that there are boundary points on the $D$-boundary between 
$\vec{\mu}_i$ and $\vec{\mu}_l$. By the RPM property $\vec{\nu}_{k+1} \in U^*\vec{\mu}_l$. This implies that 
\[
\psi_+(\vec{\mu}_l) \preceq_U \vec{\nu}_{k+1}.
\]
Again, using the RPM property on the loop $(\vec{\mu}_l,\vec{\nu}_{k+1})$, it is easily shown that if $\vec{\nu}_k \in U^*\vec{\mu}_l$ this implies that 
$\vec{\nu}_k \in \psi_-^{-1}(\vec{\mu}_l)$, which contradicts our assumption, $\vec{\nu}_k \in \psi_-^{-1}(\vec{\mu}_i)$. Thus $\vec{\nu}_k \notin U^*\vec{\mu}_l$, 
and therefore $\psi_+(\vec{\mu}_l) \succ_U \vec{\nu}_k$. Combining both inequalities, we conclude that $\psi_+(\vec{\mu}_l) = \vec{\nu}_{k+1}$. 

Next, consider the 
orbits $U^*\vec{\mu}_{r}$ for $i < r < l$. Repeating the same argument, it is shown that 
$\psi_+(\vec{\mu}_r) \preceq_U \vec{\nu}_{k+1}$ and 
$\psi_+(\vec{\mu}_r) \succ_U \vec{\nu}_k$. Thus we conclude that $\psi_+(\vec{\mu}_r)= \vec{\nu}_{k+1}$ and in particular, 
$\psi_+(\vec{\mu}_{i+1}) = \vec{\nu}_{k+1}$. 

Along the way we have also shown that $\vec{\nu}_{k+1} \in U^*\vec{\mu}_{i+1}$ and $\vec{\mu}_{i+1} \in D^*\vec{\nu}_{k+1}$. Hence the pair $(\vec{\mu}_{i+1},\vec{\nu}_{k+1})$ are the  endpoints of a loop.  It remains to be shown that the loops $(\vec{\mu}_{i+1},\vec{\nu}_{k+1})$ and $(\vec{\mu},\vec{\nu}_{k})$ 
are disjoint. This is equivalent to showing that there is no field-history confined by the endpoints of one loop, {\em cf.} Defs.~\ref{def:munuconfined} and 
\ref{def:munureachable} that leads to a state of the other loops. As we have shown above, the set of boundary states of the loops are disjoint. By 
Lemma \ref{lem:switchbackreachable} and Prop.~\ref{prop:munuorbitLemma}, any $(\vec{\mu}, \vec{\nu}_{k+1})$-confined field-history that leads from a state of 
one loop to that of the other must pass  through one of the endpoints of the loop of origin and hence enter the destination loop through its corresponding endpoint. 
The order relations \eqref{eqn:NCU} and \eqref{eqn:NCD} in their strict form 
then imply  that there is no field-history from one loop to the other that is field-confined to the endpoints of the loop of origin and hence the two loops 
are disjoint. 
\end{proof}

\begin{proof}[Lemma \ref{lem:stacking}] 
The proof of the Lemma follows from adapting the proof of Lemma \ref{lemma:stacking} and we omit it.   
\end{proof}

\begin{proof}[Theorem \ref{prop:mert}]
 Note that the states $\vec{\nu}_i$ and $\vec{\mu}_j$ in \eqref{eqn:psimim} and \eqref{eqn:psiplm} are either endpoints or intermediate 
 states, hence by Lemma \ref{lem:iminter} they are not intersection states. We distinguish the following cases:
 
 \begin{itemize}
  \item [(a)] $U\vec{\nu}_i = \vec{\nu}$. Then by Lemma \ref{lemma:stacking} we have two disjoint loops, $(\vec{\mu}, \vec{\nu}_i)$ 
  and $(\vec{\mu}_1,\vec{\nu})$ with the possibility included that the latter loop is a singleton, $\vec{\mu}_1 = \vec{\nu}$. These are the 
  loops $\ell_-$ and $\ell_+$ of the proposition.  
  
  \item [(b)] $U\vec{\nu}_i \prec_U \vec{\nu}$, so that $U\vec{\nu}_i = \vec{\nu}_{i+1}$ is an intermediate state. Denote by 
  $\vec{\nu}_{n-1}$ the boundary point for which $U\vec{\nu}_{n-1} = \vec{\nu}$. It follows that $\vec{\nu}_{n-1} \succeq_U \vec{\nu}_{i+1}$. 
  The existence of an intermediate boundary point  $\vec{\nu}_{i + 1}$ implies in turn that the $D$-boundary of the $(\vec{\mu}, \vec{\nu})$-loop 
  has at least one intermediate point and thus $\vec{\mu} \prec_D \vec{\mu}_1 \prec_D \vec{\nu}$. 
  Applying Lemmas \ref{lemma:stacking} and \ref{lem:stacking} we have 
  \begin{align}
   \psi_+(\vec{\mu}_1) &= \vec{\nu}_{i+1}, \\
   \psi_-(\vec{\nu}_{n-1}) &= \vec{\mu}_{1},
  \end{align}
  and hence the pair $(\vec{\mu}_1, \vec{\nu}_{n-1})$ forms a loop $\ell_0$. By Lemma \ref{lemma:stacking} this loop is disjoint from the loop 
  $\ell_-$. By a similar argument, it is shown that $\ell_0$ and $\ell_+$ are disjoint as well. 
   \end{itemize}
  
  We shown next that $\ell_-$ and $\ell_+$ are disjoint. Consider the loop $(\vec{\mu}, \vec{\nu}_{n-1})$ this is the disjoint union of 
  $\ell_-$ and $\ell_0$. By an application of Lemma \ref{lemma:stacking} to $(\vec{\mu}, \vec{\nu}_{n-1})$, 
  it follows that $(\vec{\mu}, \vec{\nu}_{n-1})$ and 
  $\ell_+$ are disjoint as well. Hence $\ell_-$ and $\ell_+$ are disjoint. 
  
  What remains to be shown is that the standard partition procedure has a representation in terms of an ordered tree. The standard 
  partition decomposes the parent loop $(\vec{\mu}, \vec{\nu})$ into the sub loops $\ell_\pm$ and possibly $\ell_0$. We 
  can think of the relation among the loops as one between a parent and its off-springs and represent this relation 
  as a genealogical tree, as depicted in 
  Fig.~\ref{fig:CanonicalPart} (c) and (d). Comparing with panels (a) and (b) of the figure, it is clear that the ordering of 
  off-springs matters. Applying the partition procedure in turn to each of the off-spring loops and continuing in this manner, we 
  obtain an ordered tree. The leaves of this tree are singleton-loops and hence cannot be further decomposed. 
  At the same time they form also the set of states $\mathcal{R}_{\vec{\mu}, \vec{\nu}}$. The immediate parents of these 
  states are either  $2$ or $3$-loops, {\em cf.} Fig.~\ref{fig:LoopComp} (a) and (b). 
\end{proof}

\begin{proof}[Theorem~\ref{thm:planar}]
Consider the tree representation associated with the standard partition of the loop $(\vec{\mu}, \vec{\nu})$. 
From the standard partition, as depicted in Fig.~\ref{fig:CanonicalPart} we see that if the off-spring loops are planar, 
so is the parent loop. Continuing in this manner towards the leaves of the tree, since the set $\mathcal{R}_{\vec{\mu}, \vec{\nu}}$ is finite, 
after a finite number of steps we reach the leaves, whose parent loops are either $2$- or $3$ loops, 
as depicted in Fig.~\ref{fig:LoopComp} (a) and (b). These loops are clearly planar, hence so are any parent loops formed from these. 
\end{proof}

\begin{proof}[Proposition \ref{prop:algTermination}]
 Since the endpoints of sub loop $(\vec{\mu}, \vec{\nu})$ are $(\vec{\mu'}, \vec{\nu'})$-reachable, by Lemma \ref{lem:switchbackreachable} and Def.~\ref{def:reducedFH} 
 there exist reduced field-histories for $\vec{\mu}$ and $\vec{\nu}$. In the first step of the proof we use these field histories to show that 
 the repeated application of step (A1) and (A2) of the algorithm move the endpoints $\vec{\tilde{\nu}}$ and $\vec{\tilde{\mu}}$ of the loop monotonously outwards  
along the orbits $U^*\vec{\nu}$ and $D^*\vec{\mu}$, respectively. This leads to the termination condition. In the second step we argue that the termination condition 
would hold if the algorithm had been started with (A2) instead of (A1). What remains to be shown then is that, given a loop $(\vec{\mu}, \vec{\nu})$, the maximal loops found by Algorithm~\ref{alg:maxloop} does not depend on whether it is started with step (A1) or (A2). We will show this in the third step by proving that any loop $(\vec{\mu}, \vec{\nu})$ evolves under Algorithm~\ref{alg:maxloop} into a unique maximal loop. 

 {\bf Step 1:} Consider the reduced field-history of $\vec{\nu}$ and assume that it is of the form 
 \begin{equation}
  \vec{\nu} = D^{m_{k}} U^{n_{k}} D^{m_{k-1}} \cdots U^{n_2} D^{m_1} U^{n_1} \vec{\mu'},
  \label{eqref:mlnuFH}
 \end{equation}
 for some positive integers, $k, n_1, m_1, \ldots, m_{k-1}, n_k$. 
 This is the first of the four possible forms given in Def.~\ref{def:reducedFH}. The proof for the other three cases follows similar lines and will be therefore omitted. 
 Following Def.~\ref{def:reducedFH}, the switch-back states of the reduced field-history are  
 denoted by $\vec{\mu'}, \vec{\mu}_1, \ldots \vec{\mu}_{k-1}$ and $\vec{\nu'}, \vec{\nu_1}, \ldots, \vec{\nu}_{k-1}$, so that 
  \begin{equation}
 \vec{\mu}_0 \prec_D \vec{\mu}_1  \prec_D \cdots  \prec_D \vec{\mu}_{k-1} \prec_D \vec{\nu} \prec_U \vec{\nu}_{k-1} \prec_U \cdots 
  \prec_U \vec{\nu}_1 \prec_U \vec{\nu'},     
 \label{eqn:ml1}
  \end{equation}
  where we have set $\vec{\mu}_0 = \vec{\mu'}$. 
 Moreover, since $\vec{\mu} \prec_D \vec{\nu}$, it must be that either $\vec{\mu}_{k-1} \preceq_D \vec{\mu} \prec_D \vec{\nu}$ or there exists an integer $ 0 < j < k-1$ such
 that $\vec{\mu}_{j-1} \preceq_D \vec{\mu} \prec_D \vec{\mu}_j$. Thus the partial ordering \eqref{eqn:ml1} with $\vec{\mu}$ included takes one of the two 
 forms 
 \begin{equation}
 \vec{\mu}_0 \prec_D \vec{\mu}_1  \prec_D \cdots  \prec_D \vec{\mu}_{k-1} \preceq_D \vec{\mu} < \vec{\nu} \prec_U \vec{\nu}_{k-1} \prec_U \cdots 
 \prec_U  \vec{\nu}_1 \prec_U \vec{\nu}_0,     
 \label{eqn:muembed1}
 \end{equation}
 or
 \begin{equation}
 \vec{\mu}_0   \prec_D  \cdots \vec{\mu}_{j-1} \preceq_D \vec{\mu} \prec_D \vec{\mu}_j \cdots  \prec_D 
 \vec{\mu}_{k-1} \prec_D\vec{\nu} \prec_U \vec{\nu}_{k-1} \prec_U \cdots 
 \prec_U   \vec{\nu}_0.    
 \label{eqn:muembed2}
 \end{equation}
 Note that with the choice of reduced field-history \eqref{eqref:mlnuFH} and according to \eqref{eqn:loopSeqomega}, we have the corresponding sequence of nested sub-loops 
 \begin{equation}
 (\vec{\mu}_0,\vec{\nu}_0), (\vec{\mu}_0,\vec{\nu}_1), (\vec{\mu_1},\vec{\nu}_1), (\vec{\mu_1},\vec{\nu}_2), \ldots, (\vec{\mu}_{k-1}, \vec{\nu}_k), 
 \label{eqn:loopSeqomegaFH}
\end{equation}  
where we have set $\vec{\nu}_k = \vec{\nu}$. 

\begin{figure}[t!]
  \begin{center}
    \includegraphics[width = 4.2in]{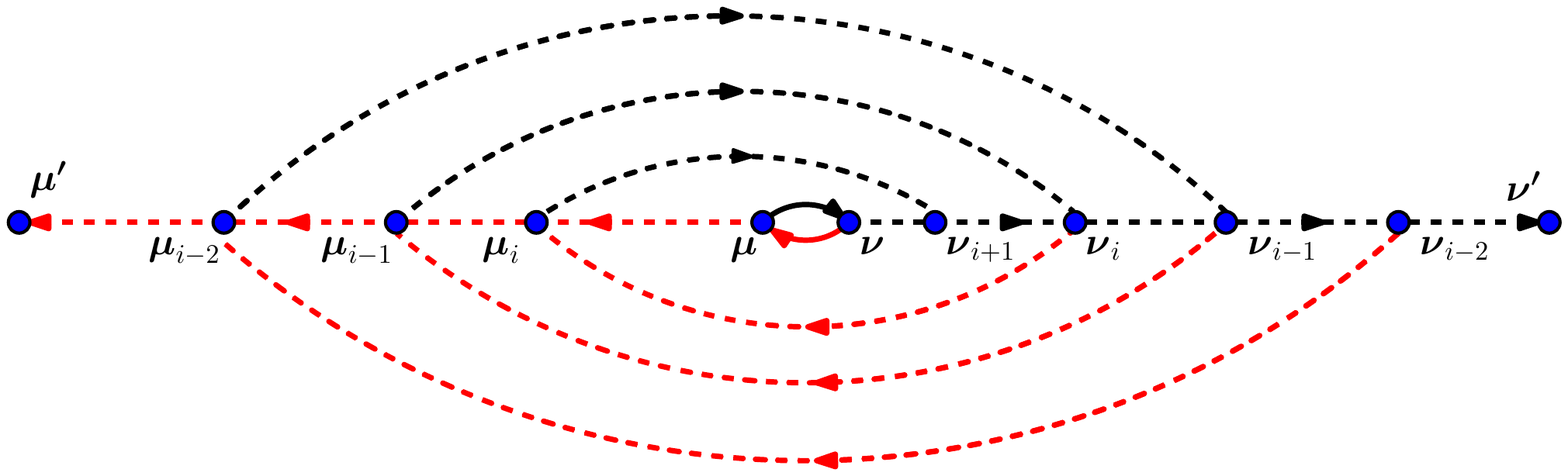} 
  \end{center}
  \caption{Proof of Prop.~\ref{prop:algTermination}. The loop $(\vec{\mu}, \vec{\nu})$ is a sub-loop of a maximal loop $(\vec{\mu'}, \vec{\nu'})$. The 
  states $\vec{\mu}_i$ and $\vec{\nu}_i$ are switch-back states of a reduced field-history of $\vec{\nu}$ starting at  $\vec{\mu'}$. The black and red 
  dashed arrows are $U$- and $D$-orbits, respectively. By Prop.~\ref{prop:munuorbitLemma}, these orbits flow through the end points of the maximal loop 
  $(\vec{\mu'}, \vec{\nu'})$. The figure corresponds to the partial ordering given in \eqref{eqn:muembed1}. Refer to the text for further details.
  } 
  \label{fig:MLProof}
\end{figure}

Referring to Fig.~\ref{fig:MLProof}, suppose that the maximal loop Algorithm \ref{alg:maxloop} has reached the loop $(\vec{\mu}_i, \vec{\nu}_i)$ and assume first that 
this happened at the end of its (A1) step, so that $\vec{\tilde{\mu}} = \vec{\mu}_i$. Then $\vec{\tilde{\nu}}$ necessarily has to satisfy 
\[
  \vec{\nu}_i \preceq_U \vec{\tilde{\nu}} \prec_U \vec{\nu}_{i-1}, 
\]
that is, $\vec{\tilde{\nu}}$ has to lie on the segment of the $U$ orbit starting at $\vec{\nu}_i$ but not containing  $\vec{\nu}_{i-1}$, since the $D$-orbit of the latter, by construction 
of the reduced field-history of $\vec{\mu}$ does not contain $\vec{\mu}_{i-1}$. By Prop.~\ref{prop:nestinglemma} applied to the loop $(\vec{\mu}_{i-1},\vec{\nu}_{i-1})$ and 
$\vec{\nu}_i$ playing the role of $\vec{\sigma}$, there exists a $U$-largest state $\vec{\eta} = \vec{\tilde{\nu}}$ with $ \vec{\nu}_i \preceq_U \vec{\tilde{\nu}} \prec_U \vec{\nu}_{i-1} $, 
whose $D$-orbit contains $\vec{\mu}_{i}$. This could be the state $\vec{\tilde{\nu}} = \vec{\nu}_i$, whose $D$-orbit contains $\vec{\mu}_{i}$. Without loss of generality we may as well assume 
that $\vec{\tilde{\nu}} = \vec{\nu}_i$, since the orbit and incidence relations as depicted in Fig.~\ref{fig:MLProof} remain the same. In other words,  the state $\vec{\nu}_i$ can be replaced by 
$\vec{\tilde{\nu}}$ without changing the incident relations of the $U$ and $D$-orbits shown. Now consider the result of applying the (A2) step. Since $\vec{\mu}_{i-1} \prec_U \vec{\tilde{\nu}}$, it is 
clear that the new lower endpoint $\vec{\tilde{\mu}}$ has to satisfy
\[
\vec{\mu}_{i-2} \prec_D \vec{\tilde{\mu}} \preceq_D \vec{\mu}_{i-1},
\]
since, again by construction, $\vec{\tilde{\nu}} \notin U^*\vec{\mu}_{i-2}$. Again without loss of generality we can replace $\vec{\mu}_{i-1}$ by $\vec{\tilde{\mu}}$ in the figure. We have in 
effect shown that the repeated application of step (A1) and (A2) of the algorithm move the endpoints $\vec{\tilde{\nu}}$ and $\vec{\tilde{\mu}}$ of the loop monotonously outwards respectively 
along the orbits $U^*\vec{\nu}$ and $D^*\vec{\mu}$. These are the horizontal dashed lines of Fig.~\ref{fig:MLProof}. By Prop.~\ref{prop:munuorbitLemma} these orbits 
contain $\vec{\mu'}$ and $\vec{\nu'}$, and since the latter form the 
endpoints of a maximal loop, the algorithm will terminate once these have been reached. 

{\bf Step 2:} If Algorithm~\ref{alg:maxloop} had been defined as starting with step (A2) instead of (A1), the definition of a maximal loop, Def.~\ref{def:loopmarginality} would remain the same and we can use the above line of reasoning to show that given a maximal 
loop $(\vec{\mu'}, \vec{\nu'})$ and a sub-loop $(\vec{\mu}, \vec{\nu})$, Algorithm \ref{alg:maxloop} initiated with (A2) would terminate 
with $(\vec{\mu'}, \vec{\nu'})$. 

{\bf Step 3:} We now prove that given a loop $(\vec{\mu}, \vec{\nu})$, the maximal loops found by Algorithm~\ref{alg:maxloop} does not depend on whether it is started with step (A1) or (A2). Suppose that $(\vec{\mu'}, \vec{\nu'})$ is a loop that contains the sub-loop $(\vec{\mu},\vec{\nu})$ such that the endpoints $\vec{\mu}$ and 
$\vec{\nu}$ are $(\vec{\mu'}, \vec{\nu'})$-reachable. Then any $(\vec{\mu'}, \vec{\nu'})$ reachable pair of states $\vec{\sigma}$ and $\vec{\eta}$ such that $\vec{\sigma} \in U^*\vec{\nu}$ and $\vec{\eta} \in D^*\vec{\mu}$ 
must satisfy
\[
 \vec{\sigma} \preceq_U \vec{\nu}^\prime \ \ \ \mbox{and} \ \ \ \vec{\mu}^\prime \preceq_D \vec{\eta}, 
\]
as follows from Prop. \ref{prop:munuorbitLemma}.  

To show the uniqueness of the maximal loops, we first assume that the loop $(\vec{\mu}, \vec{\nu})$ is both $(\vec{\mu'}_1, \vec{\nu'}_1)$-and $(\vec{\mu'}_2, \vec{\nu'}_2)$-reachable. Given any $(\vec{\mu'}, \vec{\nu'})$-reachable state $\vec{\sigma}$, by Prop. \ref{prop:munuorbitLemma}, 
$\vec{\nu'} \in U^*\vec{\sigma}$ and $\vec{\mu'} \in D^*\vec{\sigma}$. Thus for the maximal loops $(\vec{\mu'}_1, \vec{\nu'}_1)$ and $(\vec{\mu'}_2, \vec{\nu'}_2)$
we have $\vec{\mu'}_1,\vec{\mu'}_2 \preceq_D \vec{\mu}$ and $ \vec{\nu} \preceq_U \vec{\nu'}_1, \vec{\nu'}_2$. Without loss of generality assume that   
that $\vec{\mu'}_2 \preceq_D \vec{\mu'}_1$. Then since the two maximal loops are distinct it must be that either (i) 
$\vec{\mu'}_2 \prec_D \vec{\mu'}_1$, or (ii) $\vec{\mu'}_2 = \vec{\mu'}_1$, in which case we either have $\vec{\nu'}_1 \prec_U \vec{\nu'}_2$ or 
$\vec{\nu'}_2 \prec_U \vec{\nu'}_1$. Thus the pair of lower end points or the pair of upper endpoints must satisfy a strict partial order relation. Now suppose that $\vec{\mu'}_2 \prec_D \vec{\mu'}_1$. This implies that $\vec{\mu'}_1$ is $(\vec{\mu'}_2, \vec{\nu'}_2)$-reachable so $\vec{\mu'}_1$ cannot be the endpoint of a maximal loop. The case where a strict order exists between $\vec{\nu'}_1$ and $\vec{\nu'}_2$ proceeds similarly, and thus if both $(\vec{\mu'}_1, \vec{\nu'}_1)$ and $(\vec{\mu'}_2, \vec{\nu'}_2)$ are maximal loops they must be identical. 
%
%
\end{proof}

\begin{proof}[Lemma \ref{lem:confinement}]
 This is a direct consequence of the $\ell$RPM property. Part (a) of the assertion follows from Proposition \ref{prop:munuorbitLemma} and the 
 compatibility conditions \eqref{eqn:stabilityII} and \eqref{eqn:stabilityIV}. Part (b) follows from \eqref{eqn:FUaction} and \eqref{eqn:FDaction} 
 and part (a), since the loop $(\vec{\mu}, \vec{\nu})$ can only be left through its endpoints and the trapping conditions \eqref{eqn:Fconst} ensure 
 that even if the end points are reached, transitions out of the loop cannot occur. 
\end{proof}

\begin{proof}[Lemma \ref{lem:loopmarginality}]
 Let $\hat{F}^\pm$ be a simple forcing that is $(\vec{\mu}, \vec{\nu})$-traversable, so that 
 \[
  \mathcal{U}[\hat{F}^+]\vec{\mu} = \vec{\nu} \quad \mbox \quad \mathcal{D}[\hat{F}^-]\vec{\nu} = \vec{\mu}.
 \]
It must therefore be that $F^-[\vec{\mu}] < \hat{F}^- < \hat{F}^+ < F^+[\vec{\nu}]$. Now consider the state $\vec{\mu}^{(n)}$. By definition this is the 
state on the $D$-boundary of the loop $(\vec{\mu}, \vec{\nu})$ preceding $\vec{\mu}$, \eqref{eqn:miact}. Since the forcing is $(\vec{\mu}, \vec{\nu})$-traversable, 
the amplitude $\hat{F}^-$ is bounded as 
\begin{equation}
 F^-[\vec{\mu}] < \hat{F}^- \leq F^-[\vec{\mu}^{(n)}].
 \label{eq;Fhatm1}
\end{equation}
In order for the loop $(\vec{\mu}, \vec{\nu})$ to be marginal, we require that for any $\hat{F}^-$ satisfying 
the above inequality, the following inequality holds
\begin{equation}
 \hat{F}^- \leq F^-[\vec{\mu}^{(i)}], \label{eq;Fhatm2}
\end{equation}
for all $1 \leq i \leq n$. We thus conclude that it must be that 
\[
F^-[\vec{\mu}^{(i)}] \geq  F^-[\vec{\mu}^{(n)}], 
\]
for all $i = 1, 2, \ldots, n-1$. The inequalities for the trapping fields $F^+[\vec{\nu}^{(j)}]$ are obtained in a similar manner and we 
omit the proof. 
\end{proof}

\begin{proof}[Proposition~\ref{prop:AQSmarginality}]
$(b) \Rightarrow (a)$: given any loop $(\vec{\mu}, \vec{\nu})$, the associated monotonous sequences $(F^-[\vec{\mu}^{(i)}]_{1 \leq i \leq n}$ 
and $(F^+[\vec{\nu}^{(j)}])_{0 \leq m \leq m-1}$, satisfy conditions \eqref{eqn:FMcond1} and \eqref{eqn:FMcond2} of 
Lemma~\ref{lem:loopmarginality} and hence the loop is marginal. 

$(a) \Rightarrow (b)$: Consider any loop $(\vec{\mu}, \vec{\nu})$. By assumption, this loop is marginal. Let the associated predecessor states 
$(\vec{\mu}^{(i)})_{1 \leq i \leq n}$ and $(\vec{\nu}^{(j)})_{0 \leq m \leq m-1}$ be defined as in Lemma~\ref{lem:loopmarginality}. Thus 
we have the inequalities,  
\begin{align*}
 F^-[\vec{\mu}^{(i)}] &\geq F^-[\vec{\mu}^{(n)}], \quad 1 \leq i \leq n, \\
 F^+[\vec{\nu}^{(j)}] &\leq F^+[\vec{\nu}^{(0)}], \quad 0 \leq j \leq m-1. 
\end{align*}
Consider next the loops $(\vec{\mu},\vec{\nu}_{n-1})$ and 
$(\vec{\mu}_{1},\vec{\nu})$. Their predecessor states are subsets of those of the loop $(\vec{\mu}, \vec{\nu})$. We shall therefore  
retain their labels so that they are given by 
$(\vec{\mu}^{(i)})_{1 \leq i \leq n-1}$ and $(\vec{\nu}^{(j)})_{1 \leq m \leq m-1}$. The marginality of these sub loops implies that 
\begin{align*}
 F^-[\vec{\mu}^{(i)}] &\geq F^-[\vec{\mu}^{(n-1)}], \quad 1 \leq i \leq n-1, \\
 F^+[\vec{\nu}^{(j)}] &\leq F^+[\vec{\nu}^{(1)}], \quad 1 \leq j \leq m-1.
\end{align*}
Proceeding inductively, $(F^-[\vec{\mu}^{(i)}])_{1 \leq i \leq n}$ 
  and 
$(F^+[\vec{\nu}^{(j)}])_{0 \leq j \leq m-1}$ are shown to be non-increasing sequences.
\end{proof}

\begin{proof}[Theorem~\ref{thm:trans}]
 By Lemma \ref{lem:confinement}, the trajectory is confined to the states of the loop $(\vec{\mu}, \vec{\nu})$. This loop can be represented 
 by a tree with the initial state $\vec{\sigma}_0$ as one of its leaves, as defined in Theorem.~\ref{prop:mert}. There is a unique path from the leaf $\vec{\sigma}_0$ to the 
 root -- the loop $(\vec{\mu}, \vec{\nu})$ -- of the tree. Index the sequence of loops along this path by the integer $p = 0, 1, 2, \ldots, q$ so that 
 \begin{align*}
    (\vec{\mu}_0,\vec{\nu}_0) &= (\vec{\sigma}_0,\vec{\sigma}_0). \\
    (\vec{\mu}_q,\vec{\nu}_q) &= (\vec{\mu}, \vec{\nu}).
 \end{align*}
 From the standard partition, Thm.~\ref{prop:mert}, the following are easily 
 shown ({\em cf. } Fig.~\ref{fig:CanonicalPart}): 
 \begin{itemize}
  \item [(i)] $\vec{\mu}_{p+1} \preceq_D \vec{\mu}_p$,
  \item [(ii)] $\vec{\nu}_{p} \preceq_U \vec{\nu}_{p+1}$,
  \item [(iii)] two consecutive loops,  $(\vec{\mu}_p, \vec{\nu}_p)$ and $(\vec{\mu}_{p+1}, \vec{\nu}_{p+1})$, 
  do have at most one common endpoint.
 \end{itemize}
 Thus the sequences $(\vec{\mu}_p)$ and $(\vec{\nu}_p)$ are non-increasing and non-decreasing with respect to the partial order 
 $\prec_D$ and $\prec_U$, respectively. From the properties of the 
 trapping fields, \eqref{eqn:stabilityII} and \eqref{eqn:stabilityIV}, it follows then that the sequences $F^-(\vec{\mu}_p)$ and $F^+(\vec{\nu}_p)$ 
 are non-increasing, respectively, non-decreasing. 
 
 As a result, there exists a {\em smallest} $p$ such that the corresponding loop $(\vec{\mu}_p, \vec{\nu}_p)$ traps the forcing $F_t$, implying that 
 the following holds:  
 \begin{equation}
  F^-(\vec{\mu}_p) < \hat{F}^- < \hat{F}^+ < F^+(\vec{\nu}_p). 
  \label{eqn:ForceComp}
 \end{equation}
 Assume now that the loop partition of the trapping loop $(\vec{\mu}_p, \vec{\nu}_p)$ yields three loops. Label these as  $\ell^-_{p-1}, \ell^0_{p-1}$ and 
 $\ell^+_{p-1}$, for the left, middle and right loop of the partition. The case when the loop $(\vec{\mu}_p, \vec{\nu}_p)$ has only 
 two off-spring loops proceeds similarly. Since the off-spring loops are disjoint, $\vec{\sigma}_0$ must belong to one and only one of these loops. 
 The key observation underlying the remainder of the proof is to 
 note that once the off-spring loop is left, for all later times $t$ the system trajectory $\vec{\sigma}_t$ is confined to the {\em boundary} of 
 some initial condition-independent loop 
 that contains this off-spring loop. A trajectory confined to the boundary of some loop has to be periodic. What remains to be shown then is 
 that due to marginal stability the boundary of such a confining loop is reached after a transient of at most one period. 
Let us now consider each of the three cases one by one.

\noindent {\bf Case $\vec{\sigma}_0 \in \ell^-_{p-1}$:}
  This implies that 
  \begin{align*}
  \vec{\mu}_{p-1} &= \vec{\mu}_p, \\ 
  \vec{\nu}_{p-1} &\prec_U \vec{\nu}_p, 
  \end{align*}
  Consequently, 
  from \eqref{eqn:ForceComp} we see that when $F_t \geq F^+(\vec{\nu}_{p-1})$, the state trajectory $\vec{\sigma}_t$ must leave the loop 
  $(\vec{\mu}_{p-1}, \vec{\nu}_{p-1})$ sometime during the {\em first} period and while the field is still rising. 
  From Prop.~\ref{prop:munuorbitLemma} we know that it can leave this loop only through 
  the upper endpoint $\vec{\nu}_{p-1}$. Thus under the action of the increasing segment of $F_t$ the state vector eventually leaves 
  the loop $\ell^-_{p-1}$. Note however that once this happens, $\vec{\sigma}_t$ is confined to the $U$-boundary of the loop 
  $(\vec{\mu}_{p-1}, \vec{\nu}_p) = (\vec{\mu}_{p}, \vec{\nu}_p)$, {\em cf.} Fig.~\ref{fig:CanonicalPart}. 
  
  The state $\vec{\sigma}_t$ will evolve along the $U$-boundary of this loop until it reaches the first 
   boundary state $\vec{\nu'}$ for which  $\hat{F}^+ <  F^+(\vec{\nu'})$. The existence of such a state is guaranteed by condition \eqref{eqn:ForceComp}.
  By the RPM property, the pair $(\vec{\mu}_p, \vec{\nu'})$ forms a loop. Once $F_t$ starts to decrease, the state vector will move along 
  the $D$-boundary of this loop until it reaches the state $\vec{\mu}'$ at $\hat{F}^-$ and $\vec{\mu}_p \preceq_D \vec{\mu}'$.
  By the RPM property, the pair $(\vec{\mu}', \vec{\nu'})$ forms a loop.
  The subsequent increment of $F$ now causes the 
  state vector to move along the $U$-boundary of the loop $(\vec{\mu}', \vec{\nu'})$. Since $(\vec{\mu}', \vec{\nu'})$ is a sub loop of 
  $(\vec{\mu}_p, \vec{\nu'})$, and the state $\vec{\nu'}$ has been visited at the initial rise of the force to $\hat{F}^+$, the marginal stability 
  property ensures that the state $\vec{\nu'}$ will be revisited by the time the force has become again $\hat{F}^+$.  
  A periodic limit-cycle has therefore set in before the second forcing cycle completed.

  \noindent {\bf Case $\vec{\sigma}_0 \in \ell^+_{p-1}$:} This implies that
  \begin{align*}
     \vec{\mu}_{p} &\prec_D \vec{\mu}_{p-1}, \\ 
      \vec{\nu}_{p-1} &= \vec{\nu}_p. 
  \end{align*}
  By the time the force has reached its maximum value in the first period of the forcing, the trajectory will have reached a 
  state $\vec{\sigma_t} \in \ell^+_{p-1}$. The subsequent force decrease to $\hat{F}^-$ will move the system trajectory to 
  a state $\vec{\mu'} \prec_D \vec{\mu}_{p-1}$ which is on the $D$-boundary of the loop $(\vec{\mu}_p, \vec{\nu}_p)$. 
  The next increase of the force to the value $\hat{F}^+$ move the system to a state $\vec{\nu}'$. By the RPM property, 
  the pair $(\vec{\mu}', \vec{\nu'})$ forms a loop and due to the marginal stability property,
  the system   returns to the state $\vec{\mu}'$ upon decrease of the force to $\hat{F}^-$. From there on the trajectory will trace out the 
  boundary of the loop $(\vec{\mu'}, \vec{\nu}')$ and hence a limit-cycle has been attained. 

  \noindent {\bf Case $\vec{\sigma}_0 \in \ell^0_{p-1}$:} This implies that 
  \begin{align*}
  D\vec{\mu}_{p-1} &= \vec{\mu}_p, \\ 
  U\vec{\nu}_{p-1} &= \vec{\nu}_p, 
  \end{align*}
  as the standard partition glues the upper and lower endpoints of the middle loop $\ell^0_{p-1}$ to the upper endpoint $\vec{\nu}_p$ of the right loop 
  $\ell^+_{p-1}$ and the lower endpoint $\vec{\mu}_p$ of the left loop $\ell^-_{p-1}$, respectively. 
  Sine the loop $\ell^0_{p-1}$ is non-trapping, there are three possibilities: (i) only the upper endpoint is non-trapping, 
  (ii) only the lower endpoint is non-trapping, (iii) both endpoints 
  are non-trapping. The proof of (i) and (ii) proceeds along the same lines as that  of the cases $\vec{\sigma}_0 \in \ell^-_{p-1}$, 
  and $\vec{\sigma}_0 \in \ell^+_{p-1}$, respectively. We therefore omit these proofs and consider 
  only possibility (iii) where both endpoints of $\ell^0_{p-1}$ are non-trapping. As the force is increasing in the first cycle, 
  the state vector will leave the loop $\ell^0_{p-1}$ through its upper endpoint and arrive at $\vec{\nu}_p$, where it will remain 
  until the reaches its maximum value $F^+_0$, since $(\vec{\mu}_p, \vec{\nu}_p)$ is trapping. 
  On the  the subsequent decrease of the 
  driving force, the trajectory must move along the $D$-boundary of the loop, passing through the state $\vec{\mu}_{p-1}$ and reaching the subsequent 
  state $(\vec{\mu}_{p}$, once the minimum value $F^-_0$ is attained.  
  As a result of the driving the trajectory traverses thus the boundary of the 
  loop  $(\vec{\mu}_p, \vec{\nu}_p)$ and a limit cycle will have set in by the time the force attains its maximum for the second time and the 
  state $\vec{\nu}_p$ is reached. 
\end{proof}

\begin{figure}[t!]
  \begin{center}
    \includegraphics[width = 3.8in]{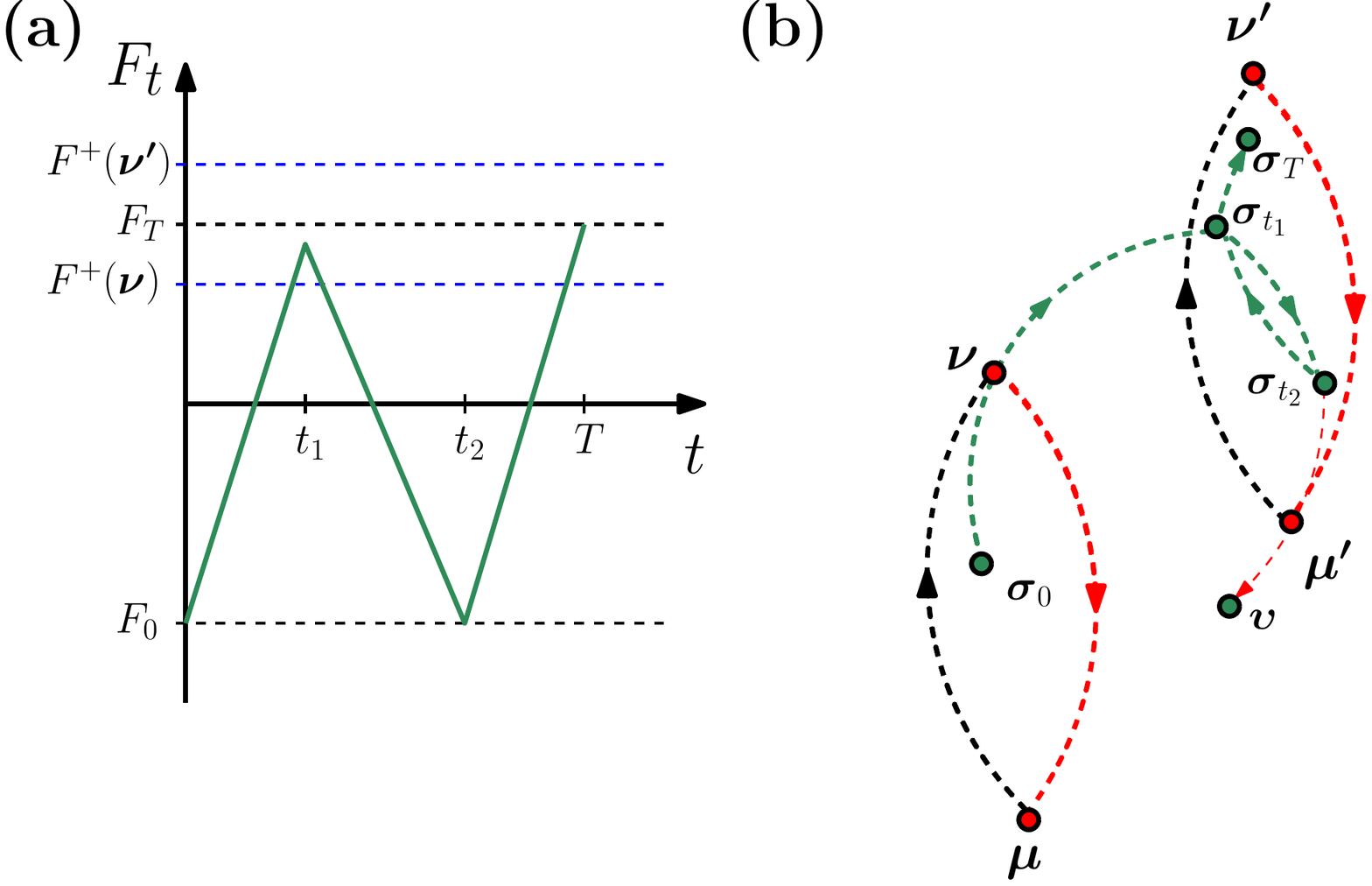} 
  \end{center}
  \caption{The No-passing property and trapping maximal loops, corresponding to part (a) of the proof of Thm.~\ref{thm:NPInterloop}. The 
  system starts out in state $\vec{\sigma}_0$ belonging to the maximal loop $(\vec{\mu}, \vec{\nu}$. A force $(F_t)_{0 \leq t \leq T}$,
  monotonously increasing from $F_0$ to $F_{t_1}$ and then decreasing back to $F_0$ a time $T$ is applied. Panel (b) shows the 
  possible system trajectory. Since $F_{t_1} > F^+(\vec{\nu})$, $\vec{\sigma}_t$ must leave the maximal loop $(\vec{\mu}, \vec{\nu}$ at 
  some time $t \leq t_1$ and enter the maximum loop containing $\vec{\sigma}_{t_1}$. With a subsequent decrease of the force during 
  $t_1 < t < t_2$ the system follows the orbit $D^*\vec{\sigma}_{t_1}$ until reaching the state $\vec{\sigma}_{t_2}$. The NP property 
  prevents this orbit to leave the loop, and it therefore must be that $F^-(\vec{\mu'}) \leq F_0$. 
  } 
  \label{fig:NPproof}
\end{figure}

\begin{proof}[Theorem~\ref{thm:NPInterloop}]
 The idea of the proof is very similar to the RPM proof of \cite{Sethna93}. The key additional ingredient is that once a maximum loop 
 $(\vec{\mu}, \vec{\nu})$ has been left through one of its endpoints, a field-reversal will not cause the trajectory to re-enter it. 
 We then use the no-passing property to construct a situation where once a maximal loop has been entered it cannot be left, showing that 
 the only way in which this can be assured is via conditions (a) or (b) of the theorem. The situation is illustrated in 
 Fig.~\ref{fig:NPproof}. We will only prove part (a) of the theorem, since the proof of (b) follows a similar reasoning. 

 Suppose we start out with a state $\vec{\sigma}_0$ belonging to the maximal loop $(\vec{\mu}, \vec{\nu})$ and suppose we apply a forcing 
 $(F_t)_{0 \leq t \leq T}$, as depicted in Fig.~\ref{fig:NPproof}(a), rising first to a value $F_{t_1} > F^+(\vec{\nu})$ at $t_1$, then lowering 
 back to $F_0$ at $t_2$ before increasing to some final value $F_T > F^+(\vec{\nu})$. Thus by the time $t_1$ the trajectory $\vec{\sigma}_t$  
 will have left the maximal loop $(\vec{\mu}, \vec{\nu})$ through its upper endpoint. Since $F_t$ is increasing up to time $t_1$, this means 
 that for $0 \leq t \leq t_1$, 
 the trajectory follows the orbit $U^*\vec{\sigma}_0$. Suppose that the state $\vec{\sigma}_{t_1}$ reached at time $t_1$ 
 belongs to the maximal loop $(\vec{\mu'}, \vec{\nu'})$, so that $\vec{\sigma}_{t_1} \in \mathcal{R}_{\vec{\mu'}, \vec{\nu'}}$. It follows  
 that
 \begin{equation}
  F^+(\vec{\nu'}) \geq F^+(\vec{\sigma}_{t_1}) > F_{t_1},
  \label{eqn:NPproofa1}
 \end{equation}
 and $(\vec{\mu'}, \vec{\nu'})$ can indeed trap the state at $t_1$. The subsequent field reduction 
 will cause the trajectory to follow the orbit $D^*\vec{\sigma}_{t_1}$. We then use the no-passing property to show that at time $t_2$ the 
 trajectory could not have passed through the lower endpoint $\vec{\mu'}$ of this maximal loop and hence it must be that 
 $F^-(\vec{\mu'}) \leq F_0$. Since the initial state is arbitrary and we could in particular have chosen $\vec{\sigma}_0 = \vec{\mu}$ as an 
 initial states and thus 
 $F_0 = F^-(\vec{\mu})$, it follows that
 \begin{equation}
  F^-(\vec{\mu'}) \leq F^-(\vec{\mu}).
  \label{eqn:NPproofa2}
 \end{equation}
  Equations \eqref{eqn:NPproofa1} and \eqref{eqn:NPproofa2} are assertion (a) of the Theorem.   
  
  Finishing the proof of (a) we now show that NP implies $F^-(\vec{\mu'}) \leq F_0$. Suppose we have two additional forcings 
  $(F^{(1)}_t)_{0 \leq t \leq T}$ and $(F^{(2)}_t)_{0 \leq t \leq T}$ such that 
  \begin{align*}
   F^{(1)}_t &= \min_{t \leq t' \leq T}  F_{t'}, \\
   F^{(2)}_t &= \max_{0 \leq t' \leq t}  F_{t'}.
  \end{align*}
  So that for all $0 \leq t \leq T$
  \[
   F^{(1)}_t \leq F_t \leq F^{(2)}_t.
  \]
  with $F^{(1)}_0 = F_0 = F^{(2)}_0$ and $F^{(1)}_T = F_T = F^{(2)}_T$. By construction both forcings are non-decreasing. 
  Apply these to $\vec{\sigma}_0$ and denote the resulting trajectories 
  $\vec{\sigma}^{(1)}_t$  and $\vec{\sigma}^{(1)}_t$. Both of these trajectories must follow the orbit $U^*\vec{\sigma}_0$, since 
  the forces are non-decreasing. It thus follows that $\vec{\sigma}^{(1)}_T = \vec{\sigma}^{(2)}_T$. From the NP property we moreover 
  find that $\vec{\sigma}_T = \vec{\sigma}^{(1)}_T$, so that all three trajectories end up in the same state at time $T$. Since 
  $(\vec{\mu'}, \vec{\nu'})$ is a maximal loop, for $t_1 \leq t \leq t_2$ the trajectory $\vec{\sigma}_t$ could not have left the 
  maximal loop. To see this note that  $t_1 \leq t \leq t_2$, $\vec{\sigma}_t$ follows the orbit $D^*\vec{\sigma}_{t_1}$. If 
  $\vec{\sigma}_{t_2} \prec_D \vec{\mu'}$ then the subsequent force increment in $t_2 \leq t \leq T$, which will follow the orbit
  $U^*\vec{\sigma}_{t_2}$ cannot return to the maximal loop, since if it did, we would have found a loop $(\vec{\sigma}_{t_2}, \vec{\nu'})$ 
  containing the maximal loop $(\vec{\mu'}, \vec{\nu'})$, which is a contradiction. The only possibility is that 
  $\vec{\mu'} \preceq_D \vec{\sigma}_{t_2}$ and therefore $F^-(\vec{\mu'}) \leq F_0$.   
\end{proof}

\end{document}